\documentclass[a4paper,reqno]{amsart}

\usepackage[table,xcdraw]{xcolor}

\usepackage{float}
\usepackage{colortbl}
\usepackage[T1]{fontenc}
\usepackage[margin=2cm]{geometry}
\usepackage[foot]{amsaddr}
\usepackage[colorlinks, linkcolor = blue, citecolor = blue, urlcolor = blue]{hyperref}

\usepackage{amsmath,amsthm,amssymb,amsfonts}
\usepackage{ytableau}
\usepackage{pifont}
\newcommand{\hollowstar}{\text{\ding{73}}}

\usepackage{mathtools}

\usepackage{thmtools}

\usepackage{mathrsfs}

\usepackage{cleveref}

\usepackage[shortlabels]{enumitem}

\usepackage{tabularx}
\usepackage{makecell}
\usepackage{multirow}
\usepackage{hhline}
\usepackage{diagbox}
\usepackage{booktabs}
\usepackage{slashbox}
\usepackage{cases}

\usepackage{caption}

\usepackage[group-separator={,},group-minimum-digits=4]{siunitx}

\usepackage{epigraph} 

\usepackage{graphicx}
\usepackage{tikz}
\usetikzlibrary{quantikz2}
\usetikzlibrary{backgrounds}
\usetikzlibrary{calc}
\usetikzlibrary{shapes}
\usetikzlibrary{decorations.pathreplacing}
\usetikzlibrary{external}

\usepackage{nicematrix}
\NiceMatrixOptions{cell-space-top-limit=5pt,cell-space-bottom-limit=5pt}
\usepackage{nicefrac}

\usepackage{ytableau,varwidth}

\usepackage[backend=biber,style=alphabetic,maxnames=9,maxalphanames=5,isbn=false,backref]{biblatex}

\newcommand{\legendbox}[1]{%
  \fcolorbox{black}[HTML]{#1}{\rule{0ex}{1ex}\rule{1ex}{0ex}}
}

\usepackage{adjustbox}

\newtheorem*{theorem*}{Theorem}
\newtheorem{theorem}{Theorem}
\newtheorem{lemma}[theorem]{Lemma}
\newtheorem{definition}[theorem]{Definition}

\newtheorem{remark}[theorem]{Remark}

\crefname{lemma}{Lemma}{Lemmas}
\crefname{definition}{Definition}{Definitions}
\crefname{theorem}{Theorem}{Theorems}
\crefname{conjecture}{Conjecture}{Conjectures}
\crefname{section}{Section}{Sections}
\crefname{claim}{Claim}{Claims}
\crefname{appendix}{Appendix}{Appendices}
\crefname{figure}{Fig.}{Figs.}
\crefname{table}{Table}{Tables}
\crefname{proposition}{Proposition}{Propositions}
\crefname{corollary}{Corollary}{Corollaries}
\crefname{example}{Example}{Examples}
\crefname{remark}{Remark}{Remarks}

\providecommand\given{}
\newcommand\SetSymbol[1][]{%
    \nonscript\:#1\vert
    \allowbreak
    \nonscript\:
    \mathopen{}}
\DeclarePairedDelimiterX\Set[1]\{\}{%
    \renewcommand\given{\SetSymbol[\delimsize]}
    #1
}

\DeclarePairedDelimiter{\set}{\lbrace}{\rbrace}
\DeclarePairedDelimiter{\abs}{\lvert}{\rvert}
\DeclarePairedDelimiter{\norm}{\lVert}{\rVert}

\DeclarePairedDelimiter{\of}{\lparen}{\rparen}
\DeclarePairedDelimiter{\sof}{\lbrack}{\rbrack}

\newcommand{\defeq}{\vcentcolon=}

\renewcommand{\leq}{\leqslant}
\renewcommand{\geq}{\geqslant}

\DeclareMathOperator*{\argmax}{arg\,max}

\renewcommand{\bra}[1]{\langle{#1}\rvert}
\renewcommand{\ket}[1]{\lvert{#1}\rangle}
\renewcommand{\braket}[2]{\langle{#1}|{#2}\rangle}
\newcommand{\ketbra}[2]{\ket{#1}\bra{#2}}
\renewcommand{\proj}[1]{\ketbra{#1}{#1}}

\newcommand{\mx}[1]{\begin{pmatrix}#1\end{pmatrix}}

\newcommand{\ct}{^{\dagger}}

\newcommand{\x}{\otimes}
\newcommand{\xp}[1]{^{\otimes #1}}

\newcommand{\C}{\mathbb{C}} 
\newcommand{\cont}{\mathrm{cont}} 
\newcommand{\len}{\ell} 
\newcommand{\Brat}{\mathscr{B}} 
\newcommand{\A}{\mathcal{A}} 
\DeclareMathOperator{\poly}{poly} 
\DeclareMathOperator{\polylog}{polylog} 

\newcommand{\Usch}{U_{\mathrm{Sch}}} 

\newcommand{\Uirrep}[1]{\mathcal{W}_{#1}} 
\newcommand{\Airrep}[1]{\mathcal{H}_{#1}} 

\newcommand{\0}{\varnothing}

\let\S\relax
\DeclareMathOperator{\S}{S} 
\DeclareMathOperator{\Tr}{Tr} 

\newcommand{\Irr}[1]{\widehat#1} 
\newcommand{\Paths}{\mathrm{Paths}} 

\newcommand{\U}[1]{\mathrm{U}_{#1}} 

\newcommand{\psucc}{p_{\mathrm{succ}}} 
\newcommand{\N}{\mathcal{N}} 
\renewcommand{\AA}{\bar{A}} 
\newcommand{\BB}{\bar{B}} 

\newcommand{\GT}{\mathrm{GT}} 

\newcommand{\gtM}[1]{M_{[#1]}}

\newcommand{\pt}{\mathbin{\vdash}} 
\newcommand{\AC}{\mathrm{AC}} 
\newcommand{\RC}{\mathrm{RC}} 

\newcommand{\w}{0.5cm}

\newcommand{\bx}[3]{
  \draw[fill = white] #3 (#1*\w-\w/2,-#2*\w-\w/2) rectangle (#1*\w+\w/2,-#2*\w+\w/2);
}
\newcommand{\br}[1]{\of*{#1}}

\newcommand{\yd}[2][0.4]{%
  \begin{tikzpicture}[scale = #1, baseline={([yshift=-0.6ex]current bounding box.center)}]
    \foreach \li [count = \y] in {#2} {
      \foreach \x in {1,...,\li} {
        \bx{\x}{\y}{}
      }
    }
  \end{tikzpicture}
}

\newcommand\restr[2]{{
  \left.\kern-\nulldelimiterspace 
  #1 
  \vphantom{\big|} 
  \right|_{#2} 
  }}

\newcommand{\HD}{\ \ldots\ } 
\newcommand{\vd}{\vdots} 

\newcommand{\qb}{\setwiretype{b}} 
\newcommand{\qq}{\setwiretype{q}} 
\newcommand{\qn}{\setwiretype{n}}

\newcommand{\gateQFT}[1]{\gate{\mathrm{QFT}_{#1}}}
\newcommand{\gateQFTinv}[1]{\gate{\mathrm{QFT}_{#1}^\dagger}}

\title[Efficient quantum circuits for port-based teleportation]{Efficient quantum circuits for port-based teleportation}

\author{Dmitry Grinko\textsuperscript{1}\hspace{-.25em}}
\email{dmitry.grinko@cwi.nl}

\author{Adam Burchardt\textsuperscript{1}\hspace{-.25em}}
\address{\textsuperscript{1}Institute for Logic, Language, and Computation, University of Amsterdam and QuSoft, Amsterdam, The Netherlands}
\email{adam.burchardt@cwi.nl}

\author{Maris Ozols\textsuperscript{1,2}\hspace{-.25em}}
\address{\textsuperscript{2}Korteweg-de Vries Institute for Mathematics and Institute for Theoretical Physics, University of Amsterdam, The Netherlands}
\email{marozols@gmail.com}

\begin{document}

\begin{abstract}
Port-based teleportation (PBT) is a variant of quantum teleportation that, unlike the canonical protocol by Bennett et al., does not require a correction operation on the teleported state. Since its introduction by Ishizaka and Hiroshima in 2008, no efficient implementation of PBT was known. We close this long-standing gap by building on our recent results on representations of partially transposed permutation matrix algebras and mixed quantum Schur transform. We construct efficient quantum algorithms for probabilistic and deterministic PBT protocols on $n$ ports of arbitrary local dimension, both for EPR and optimized resource states. We describe two constructions based on different encodings of the Gelfand--Tsetlin basis for $n$ qudits: a standard encoding that achieves $\tilde{O}(n)$ time and $O(n \log(n))$ space complexity, and a Yamanouchi encoding that achieves $\tilde{O}(n^2)$ time and $O(\log(n))$ space complexity, both for constant local dimension and target error. We also describe efficient circuits for preparing the optimal resource states.
\end{abstract}

\maketitle

\tableofcontents

\newpage
\section{Introduction}\label{sec:intro}

\subsection{Background}

Quantum teleportation is a cornerstone of quantum information \cite{bennett1993teleporting}.
However, one potentially undesirable feature of the original teleportation protocol is that the receiving party needs to perform a correction operation on the received state.
\emph{Port-based teleportation} (PBT) gets around this limitation \cite{IshizakaHiroshima,ishizaka2009quantum}.
In PBT, Alice and Bob share an entangled resource state distributed among $n$ quantum systems called \emph{ports} on each side. 
To teleport an unknown quantum state, Alice measures it together with her share of the ports.
The measurement outcome, which she communicates to Bob, indicates the port on Bob's side to which the state has been teleported.
Bob does not need to perform any correction operation but simply retrieve the state from the correct port.
Each of the quantum systems involved has a fixed \emph{local dimension} $d$.

Our current understanding of PBT protocols is very detailed thanks to a long sequence of works \cite{IshizakaHiroshima,ishizaka2009quantum,beigi2011simplified,ishizaka2015remarks,studzinski2017port,mozrzymas2018optimal,leditzky2020optimality,christandl2021asymptotic}. In particular, \cite{studzinski2017port,mozrzymas2018optimal} were the first to obtain exact formulas for the asymptotic performance of PBT.
The resource requirements for PBT have been studied further in \cite{studzinski2022square,strelchuk2023minimal}.
The original PBT protocol has subsequently been extended to multi-port teleportation \cite{studzinski2020efficient,kopszak2020multiport,mozrzymas2021optimal} where several systems are teleported simultaneously.

A crucial feature of port-based teleportation is unitary equivariance \cite{grinko2022linear}, meaning that applying any unitary on Alice's input state is equivalent to applying the same unitary to all of Bob's ports.
Due to unitary equivariance, PBT can be seen as a concrete example of an approximate universal programmable quantum processor \cite{IshizakaHiroshima}.
The quantum no-programming theorem \cite{NoProgramming,Programmability,OptimalUniversalProgramming} implies that unitarily equivariant PBT protocols with a finite-dimensional resource state can only achieve approximate teleportation.
Nevertheless, certain PBT protocols are asymptotically faithful in the limit of a large number of ports.

PBT has diverse applications in non-local quantum computation and quantum communication \cite{beigi2011simplified,buhrman2016quantum,may2022complexity} with applications to quantum position verification \cite{allerstorfer2022role,allerstorfer2023relating}, channel discrimination \cite{pirandola2019fundamental}, channel simulation \cite{pereira2021characterising}, and holography in high-energy physics \cite{may2019quantum,may2022complexity}.

\subsection{Types of PBT protocols}

The two main types of PBT protocols considered are \emph{probabilistic exact} (pPBT) and \emph{deterministic inexact} (dPBT) \cite{studzinski2017port,mozrzymas2018optimal,leditzky2020optimality}.
``Exact'' means that the protocol achieves perfect entanglement fidelity $F = 1$, while ``inexact'' means that $F < 1$.
``Probabilistic'' refers to the fact that the protocol has some nonzero probability of failure, while ``deterministic'' highlights no possibility of failure, i.e., the average success probability of the protocol is $\psucc = 1$.\footnote{For brevity, in the rest of this section we will drop the terms ``exact'' and ``inexact'' in these cases since deterministic PBT protocols with finite resources are always inexact, and exact protocols cannot be deterministic due to \cite{NoProgramming}. Of course, there is still a possibility of having probabilistic inexact protocols.}
Besides pPBT and dPBT, one can also consider \emph{probabilistic inexact} protocols that interpolate between pPBT and dPBT.
One concrete example is minimal PBT (mPBT) which is a modified version of dPBT \cite{strelchuk2023minimal}.

Two types of resource states for PBT are typically considered: $n$ pairs of maximally entangled states and an arbitrary optimized state; we call these the \emph{EPR resource state} and the \emph{optimized resource state}, respectively.
Depending on the application, Alice's measurement is chosen either to maximize the entanglement fidelity $F$ (in dPBT) or the average success probability $\psucc$ (in pPBT).
While the optimized resource states are different for deterministic and probabilistic protocols, the optimal measurement for optimized states turns out to be the same in both cases.
\Cref{tab:pbt_summary} summarizes the four main types of PBT protocols and their optimal fidelity and probability of success. 

\begin{table}[ht]
    \centering
    \begin{NiceTabular}{|c||c|c|} \hline
        {\multirow{3}{*}{Resource state}} &   \multicolumn{2}{c}{Protocol type} \\ \cline{2-3}
            &  Deterministic inexact (dPBT) &  Probabilistic exact (pPBT)  \\ \hline \hline
         EPR&  \cellcolor[HTML]{FFFDD0} $\begin{aligned} F &= 1 - O\of*{\nicefrac{1}{n}} \\ \psucc &= 1\end{aligned}$& \cellcolor[HTML]{AFEEEE} $\begin{aligned}\nicefrac{F}{\psucc} &= 1 \\ \psucc &= 1-O\of*{\nicefrac{1}{\sqrt{n}}}\end{aligned}$\\ \hline 
         Optimized& \cellcolor[HTML]{FFFDD0} $\begin{aligned} F &= 1 - O\of*{\nicefrac{1}{n^2}} \\ \psucc &= 1\end{aligned}$& \cellcolor[HTML]{FDF5E6} $\begin{aligned}\nicefrac{F}{\psucc} &= 1 \\ \psucc &= 1-O\of*{\nicefrac{1}{n}}\end{aligned}$\\ \hline 
    \end{NiceTabular}
    \caption{Summary of different flavours of PBT protocols. Rows correspond to EPR and optimized resource states while columns correspond to deterministic (dPBT) \cite{studzinski2017port,mozrzymas2018optimal,leditzky2020optimality,christandl2021asymptotic} and probabilistic (pPBT) \cite{studzinski2017port,christandl2021asymptotic} protocols. Two figures of merit are used: the average success probability $\psucc$ and the entanglement fidelity $F$ (normalized by the success probability $\psucc$) as functions of the number of ports $n$, ignoring the dependence on the local dimension $d$. Cell colors correspond to different optimal POVMs: \legendbox{FDF5E6} is the standard PGM $E$ defined in \cref{def:PGM_PBT,def:PGM_PBT_schur}, \legendbox{FFFDD0} is the dPBT POVM $E^{\hollowstar}$ defined in \cref{def:dpbt_povm,def:dpbt_povm_schur} and \legendbox{AFEEEE} is the POVM $E^{\star}$ defined in \cref{def:povm_generic,def:G_lambda_ppbt}.}
    \label{tab:pbt_summary}
\end{table}

\subsection{Summary of our results}

While analytic expressions for the optimal measurement operators in PBT were known \cite{studzinski2017port,mozrzymas2018optimal,leditzky2020optimality}, efficient quantum circuits that implement them were not known until our work.
Our main result provides explicit quantum circuits for implementing PBT and analyzes their complexity.

\begin{theorem}\label{thm:pbt}
    The measurements for all four types of PBT protocols (deterministic/probabilistic and with optimized/EPR resource state) can be implemented in two ways with the following time and space complexities:
    \begin{enumerate}
        \item standard encoding: $n d^4 \polylog(n,d,1/\epsilon)$ time and $(n + d) d \log(n) \polylog(d,1/\epsilon)$ space,
        \item Yamanouchi encoding: $n^2 d^4 \polylog(n,d,1/\epsilon)$ time and $d^2 \log(n) \polylog(d,1/\epsilon)$ space,
    \end{enumerate}
    where $n$ is the number of ports, $d$ is the dimension of the teleported quantum state, and $\epsilon$ is the target error.
    In both cases, the total gate complexity is the same: $n^2 d^4 \polylog(n,d,1/\epsilon)$.
\end{theorem}

The setting of port-based teleportation is naturally suited for using the representation theory of the matrix algebra $\A_{n,1}^d$ of partially transposed permutations \cite{grinko2022linear,grinko2023gelfandtsetlin}. It is natural to work in the Gelfand--Tsetlin basis, which can be achieved by applying the mixed quantum Schur transform \cite{grinko2023gelfandtsetlin, nguyen2023mixed}. Therefore, to prove \Cref{thm:pbt} we first explain how to rewrite PBT measurements in the Gelfand--Tsetlin basis in \Cref{sec:std_pgm,sec:povms_dpbt_ppbt} based on representation theory of partially transposed permutations (\cref{sec:Representation theory of the partially transposed permutation algebra}) and mixed quantum Schur transform (\cref{sec:Mixed quantum Schur transform}). Next, we explain how to construct explicitly Naimark dilations of PBT measurements in \Cref{sec:easy_facts,sec:pgm:dilation}. Finally, we use these ingredients to construct efficient quantum algorithms for PBT in standard encoding in \Cref{sec:circuit_for_pgm,sec:dPBT_circuits,sec:pPBT_circuits,sec:G_PBT_circuits} and Yamanouchi encoding in \Cref{sec:yaman_enc} together with their respective complexity analysis.

We highlight that our construction works also for a large class of generic PBT protocols, see\cref{sec:pPBT_circuits}.
Moreover, we construct efficient quantum circuits for preparing optimized states for probabilistic PBT protocols in \cref{sec:opt_states}. The question of preparing efficiently optimized resource states for deterministic PBT is left open.

\subsection{Related work}

The main idea of Naimark's dilation behind the efficient construction of \emph{standard pretty good measurement (PGM)} for PBT was initially presented in version 1 of our paper \cite{grinko2023gelfandtsetlin}. This material is now presented in \cref{sec:pgm:dilation,sec:circuit_for_pgm} with more details and expanded analysis. The constructions of all PBT protocols presented in the current paper can be seen as a straightforward generalization of this main idea presented in \cite{grinko2023gelfandtsetlin}.

Our initial result for standard PGM from \cite{grinko2023gelfandtsetlin} was obtained independently and simultaneously with \cite{fei2023efficient}, which solves the problem for dPBT protocols with different methods. A relevant independent construction of mixed quantum Schur transform was obtained in \cite{nguyen2023mixed} together with our construction in \cite{grinko2023gelfandtsetlin}, which we use as a building block in PBT analysis. Moreover, a logarithmic space qubit Schur transforms described in \cite{kirby2018practical,wills2023generalised} inspired our Yamanouchi encoding construction presented in \cref{sec:yaman_enc}.

While writing this manuscript, we became aware of the independent work \cite{wills2023efficient} which tackles specifically the qubit case for all PBT protocols. In contrast to \cite{wills2023efficient}, our independent result solves the problem for dPBT and pPBT protocols for arbitrary local dimensions in full generality, achieving better time complexities. 
On the technical side, both results \cite{fei2023efficient,wills2023efficient} construct their PBT protocols by using general techniques of block encodings and amplitude amplification, while our construction requires neither technique and is based on the phase estimation primitive. This leads to more practical circuits and more appealing theoretical analysis. These features of our constructions are enabled by proper use of the representation theory of partially transposed permutation matrix algebras via the Gelfand--Tsetlin basis which we developed in \cite{grinko2023gelfandtsetlin,grinko2022linear}.

Finally, a relevant work \cite{Decker2004} discusses general methods for Naimark's dilations of rank one covariant POVMs. Our work analyses specific multi-rank covariant POVMs used in PBT and does not employ the techniques used in \cite{Decker2004}. However, the ideas presented in \cite{Decker2004} could be useful in understanding and generalizing our construction. 

\section{Preliminaries}

\subsection{Notation} 

Let $n \geq 0$ be an integer.
We write $\lambda \pt n$ to mean that $\lambda$ is a \emph{partition} of $n$, i.e.,
$\lambda = (\lambda_1,\dotsc,\lambda_d)$ for some $d$ where
$\lambda_1 + \dotsb + \lambda_d = n$ and
$\lambda_1 \geq \dots \geq \lambda_d \geq 0$.
The \emph{size} of partition $\lambda$ is $\abs{\lambda} = n$.
We say that $\lambda$ has \emph{length} $\len(\lambda) = k$ if $\lambda_k > 0$ and $\lambda_{k+1} = 0$. We use the notation $\lambda \pt_d n$ to indicate that $\lambda \pt n$ and $\len(\lambda) \leq d$.
The \emph{Young diagram} of partition $\lambda$ is a collection of $n$ \emph{cells} arranged in $\len(\lambda)$ rows, with $\lambda_i$ of them in the $i$-th row.
We will use these two notions interchangeably. Any cell $u \in \lambda$ can be specified by its row and column coordinates $i$ and $j$. 
The \emph{content} of cell $u = (i,j)$ is
\begin{equation}
    \cont(u) \defeq j - i.
    \label{eq:content}
\end{equation}
Note that content is constant on diagonals of $\lambda$ and increases by one when going right or up, and decreases by one when going left or down, see \Cref{fig:YD}. 

A cell $u \in \lambda$ is called \emph{removable} if the diagram $\lambda \setminus u$ obtained by removing $u$ from $\lambda$ is a valid Young diagram.
Similarly, a cell $u \notin \lambda$ is called \emph{addable} if the diagram $\lambda \cup u$ obtained by adding $u$ to $\lambda$ is a valid Young diagram.
We will denote the sets of all removable and addable cells of $\lambda$ by $\RC(\lambda)$ and $\AC(\lambda)$, respectively.
Furthermore, $\AC_d(\lambda) \subseteq \AC(\lambda)$ denotes the subset of addable cells $a$ of $\lambda$ such that $\lambda \cup a$ has at most $d$ rows:
\begin{equation}\label{def:AC_d}
    \AC_d(\lambda) \defeq \Set{ a \in \AC(\lambda) \given \len(\lambda \cup a) \leq d }.
\end{equation}

\begin{figure}
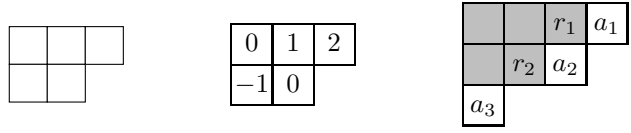

    \centering
    \begin{equation*}
    \yd[1]{3,2} \qquad \qquad
    \ytableausetup{centertableaux}
    \begin{ytableau}
        0 & 1 & 2  \\
       -1 & 0 
    \end{ytableau} \qquad \qquad
    \ytableausetup{centertableaux}
    \begin{ytableau}[*(lightgray)]
      & & r_1 & *(white) a_1 \\
      & r_2 & *(white) a_2 \\
      *(white) a_3
    \end{ytableau}
    \end{equation*}
    
    \caption{Example of a Young diagram $\lambda = (5,2)$ (left) together with its contents (middle) and the set $\RC(\lambda)$ of its two removable cells
    $r_1 = (1,3)$, $r_2 = (2,2)$, and the set $\AC(\lambda)$ containing three addable cells $a_1 = (1,4)$, $a_2 = (2,3)$, $a_3 = (3,1)$ (right, $\lambda$ shown in grey).}
    \label{fig:YD}
\end{figure}


\subsection{Representation theory of the partially transposed permutation algebra}
\label{sec:Representation theory of the partially transposed permutation algebra}

The setting of PBT is naturally suited for certain tools from representation theory, such as mixed Schur--Weyl duality for the matrix algebra of partially transposed permutations.
Here we present a short introduction to these topics and refer the reader to \cite{grinko2023gelfandtsetlin} for more details.

The matrix algebra $\A^d_{n,1}$ of \emph{partially transposed permutations} acts on $n+1$ qudits, each of local dimension $d$.
Its generators $\sigma_1, \dotsc, \sigma_n$ act on $(\C^d)\xp{n+1}$ in the following way:
for all $x_1, \dotsc, x_{n+1} \in [d] \defeq \set{1,\dotsc,d}$,
\begin{equation}
    \label{eq:Brauer action}
    \sigma_i \, \ket{x_1,\dotsc,x_{n+1}}
    \defeq
    \begin{cases}
        \ket{x_1,\dotsc,x_{i+1},x_i,\dotsc,x_{n+1}}, &
        \text{$i \neq n$}, \\
        \ket{x_1,\dotsc,x_{n-1}}
        \x
        \delta_{x_n,x_{n+1}}
        \sum_{k=1}^d \ket{k,k}
        , &
        \text{$i = n$}.
    \end{cases}
\end{equation}
In other words, $\sigma_i$ with $i < n$ are \emph{transpositions} that exchange qudits $i$ and $i+1$, while $\sigma_n$ is a \emph{contraction} that projects the last two qudits on the un-normalized maximally entangled state.
The irreducible representations or \emph{irreps} of $\A_{n,1}^d$ are labelled by the following pairs of Young diagrams of two types:
\begin{equation}
    \label{eq:IrrA}
    \Irr{\A_{n,1}^{d}} \defeq
    \Set[\Big]{
    (\lambda,\0) \given
    \lambda \pt n - 1, \;
    \len(\lambda) \leq d} \sqcup
    \Set[\Big]{
    (\mu,\yd[0.5]{1}) \given
    \mu \pt n , \;
    \len(\mu) \leq d-1},
\end{equation}
where $\0$ denotes the empty diagram and $\square = (1)$.
More generally, the irreducible representations of $\A_{n,m}^d$ for arbitrary $n,m \geq 0$ are labelled by pairs of partitions $\Lambda = (\lambda,\lambda')$ such that $\len(\lambda) + \len(\lambda') \leq d$ and $\lambda \pt n-k$ and $\lambda' \pt m-k$ for some $0 \leq k \leq \min(n,m)$ \cite{grinko2023gelfandtsetlin} (in our case $m=1$ and $k=0,1$).
We slightly abuse notation by not including $d$ as part of $\Lambda$ since $d$ is assumed to be fixed throughout.
However, knowing $d$ is necessary to unambiguously convert $\Lambda$ into a \emph{staircase} of length $d$, which is another convenient way of labelling the irreps of $\A_{n,m}^d$ \cite{stembridge1987rational,grinko2023gelfandtsetlin}.

To describe the representation theory of $\A_{n,1}^d$, we employ the so-called Okounkov--Vershik approach.
A central notion in this approach is the \emph{Bratteli diagram} $\Brat$
for the sequence of algebras
$\A_{0,0}^d \hookrightarrow \A_{1,0}^d \hookrightarrow \cdots \hookrightarrow \A_{n,0}^d \hookrightarrow \A_{n,1}^d$,
which is a directed acyclic simple graph obtained as follows.
The vertices of $\Brat$ are divided into $n+2$ \emph{levels} denoted by $i = 0,\dotsc,n+1$. These levels correspond to sets of irreducible representations $\Irr{\A_{0,0}^d},\Irr{\A_{1,0}^d},\dotsc,\Irr{\A_{n,1}^d}$ of the corresponding algebras $\A_{0,0}^d,\A_{1,0}^d,\dotsc,\A_{n,1}^d$.
The vertices at level $i \in \set{0,\dotsc,n}$ are labelled by Young diagrams $\lambda \pt_d i$, while the vertices at the last level $n+1$ are labelled by pairs of Young diagrams $\Lambda$ corresponding to the irreps $\Irr{\A_{n,1}^{d}}$, see \cref{eq:IrrA}.
When $i \in \set{0,\dotsc,n}$, the vertices $\mu \pt i-1$ and $\lambda \pt i$ are connected, denoted as $\mu \rightarrow \lambda$, if $\lambda$ can be obtained from $\mu$ by adding a cell, i.e., $\lambda = \mu \cup a$ for some $a \in \AC_d(\nu)$. 
Furthermore, a vertex $\mu \pt n$ at level $n$ is connected to $\Lambda$ at level $n+1$ iff $\Lambda = (\mu,\square)$ with $\len(\mu) \leq d-1$ or $\Lambda = (\lambda,\0)$ if $\mu$ can be obtained by adding a cell to $\lambda$, i.e., $\mu = \lambda \cup a$ for some $a \in \AC_d(\lambda)$. 
The Bratteli diagram $\Brat$ consists of all vertices from all levels and the directed edges between them (see \cref{fig:A51} for an example).
We denote the only vertex at level $0$ by $\0$ and call it \emph{root}, while the vertices at the last level of $\Brat$ we call \emph{leaves} (they correspond bijectively to the irreps $\Irr{\A_{n,1}^{d}}$).
For any leaf $\Lambda \in \Irr{\A_{n,1}^{d}}$, we denote by
\begin{equation}
    \Paths(\Lambda,\Brat) \defeq
     \Set[\Big]{(T^0,T^1,\dotsc,T^n,T^{n+1}) \given T^{k} \pt k \text{ for } k \leq n, \,
     T^{i-1} \rightarrow T^{i} \text{ for } i \in [n+1] \text{ and } T^{n+1} = \Lambda
     }
     \label{eq:paths}
\end{equation}
the set of all paths in $\Brat$ starting at the root $T^0 = \0$ and terminating at $\Lambda$. 
If $i \in [n]$ is an intermediate position on a path $T$, the diagram $T^i$ is obtained from $T^{i-1}$ by adding some addable cell $a \in \AC_d(T^{i-1})$,
hence $T^i \setminus T^{i+1}$ consists of a single cell.
For any path $T = (T^0,\dotsc,T^{n+1}) \in \Paths(\Lambda,\Brat)$ and $i \in [n]$ we define the \emph{content} of $i$ in $T$ as
\begin{equation}
    \cont_i(T) \defeq \cont(T^{i} \backslash T^{i-1})
    \label{eq:conti}
\end{equation}
and the \emph{axial distance} between $i$ and $i+1$ in $T$ as
\begin{equation}
    r_i(T) = \cont_{i+1}(T)-\cont_{i}(T).
    \label{eq:ri}
\end{equation}
Similar to \cref{eq:paths}, for any intermediate vertex $\lambda \pt k$ at level $k \leq n$ in $\Brat$, we use $\Paths_{k}(\lambda,\Brat)$ to denote the set of all paths in $\Brat$ terminating at $\lambda$.
Furthermore, we denote by $\Paths(\Brat)$ the set of all paths in $\Brat$, i.e.,
$\Paths(\Brat) \defeq \bigsqcup_{\Lambda \in \Irr{\A_{n,1}^{d}}} \Paths(\Lambda,\Brat)$.

The Bratteli diagram $\Brat$ plays an important role in the representation theory of $\A_{n,1}^{d}$ \cite{grinko2022linear,grinko2023gelfandtsetlin}.
Recall that the vertices on the last level of $\Brat$ correspond to irreducible representations $\Irr{\A_{n,1}^{d}}$ of $\A_{n,1}^{d}$.
For a given $\Lambda \in \Irr{\A_{n,1}^{d}}$, the corresponding irrep $\psi_\Lambda$ has a convenient explicit description in the so-called \emph{Gelfand--Tsetlin basis}
$\set{\ket{T} \mid T \in \Paths(\Lambda,\Brat)}$ which is spanned by all paths from $\0$ to $\Lambda$ in $\Brat$.
Recall from \cref{eq:Brauer action} that $\A_{n,1}^{d}$ is generated by $n-1$ transpositions $\sigma_1,\dotsc,\sigma_{n-1}$ and one contraction $\sigma_n$.
For any irrep $\Lambda \in \Irr{\A_{n,1}^{d}}$, a transposition $\sigma_i$ acts on a given path $T \in \Paths(\Lambda,\Brat)$ as follows:
\begin{align}
    \psi_\Lambda(\sigma_i) \, \ket{T}
    &= \frac{1}{r_i(T)} \, \ket{T} + \sqrt{1 - \frac{1}{r_i(T)^2}} \, \ket{\sigma_i T}
    \qquad
    \text{for $i \neq n$},
    \label{gtbasis:transpositions}
\end{align}
where $\sigma_i T$
denotes the path $T$ with vertex $T^i$ at level $i$ replaced by $T^{i-1} \cup (T^{i+1} \setminus T^i)$.
\Cref{gtbasis:transpositions} coincides with the well-known Young--Yamanouchi basis for the symmetric group $\S_n$.
The remaining generator $\sigma_n$ acts trivially in irreps $\Lambda=(\mu,\yd[0.5]{1})$:
\begin{align}
    \psi_{\of*{\mu,\yd[0.3]{1}}}(\sigma_n) \, \ket{T}
    &= 0.
    \label{gtbasis:contraction}
\end{align}
To describe the action of $\sigma_n$ in the remaining irreps $\Lambda = (\lambda,\0)$, notice that any path $T \in \Paths((\lambda,\0),\Brat)$ is of the form $T = S \circ (\lambda\cup a) \circ (\lambda,\0)$ for some prefix $S \in \Paths_{n-1}(\lambda,\Brat)$, partition $\lambda \pt n-1$, and addable cell $a \in \AC_d(\lambda)$, where $\circ$ extends the path by appending new vertices. 
The generator $\sigma_n$ acts on such path in the following way (see \cite[Theorem 3.2]{grinko2023gelfandtsetlin}):
\begin{align}
    \psi_{(\lambda,\0)}(\sigma_n) \, \ket{S\circ (\lambda\cup a) \circ (\lambda,\0)}
    &=\sum_{a'\in \AC_d(\lambda)} \frac{ \sqrt{ m_{\lambda\cup a} m_{\lambda\cup a'}}}{ m_{\lambda} } \ket{S\circ (\lambda\cup a') \circ (\lambda,\0)},
    \label{gtbasis:contraction1}
\end{align}
where $m_\lambda$ denotes the dimension of the $\lambda$-irrep of the unitary group $\U{d}$.
According to the well-known \emph{Weyl dimension}
\cite[eq.~(11.46)]{louck2008unitary} and \emph{hook-content} formulas, $m_{\lambda}$ can be computed for a given Young diagram $\lambda$ as
\begin{equation}
    m_{\lambda} = \prod_{1 \leq i < j \leq d}
    \frac{\lambda_i - \lambda_j + j - i}{j-i} = \prod_{c \in \lambda} \frac{d + \cont(c)}{h_{\lambda}(c)},
    \label{eq:unitary_dimension}
\end{equation}
where $c$ ranges over all cells in $\lambda$ and $h_{\lambda}(c)$ is the \emph{hook length} of $c$ in $\lambda$. Weyl dimension formula also works for a general staircase $\Lambda$:
\begin{equation}
    m_{\Lambda} = \prod_{1 \leq i < j \leq d}
    \frac{\Lambda_i - \Lambda_j + j - i}{j-i}.
\end{equation}

\subsection{Mixed quantum Schur transform}
\label{sec:Mixed quantum Schur transform}

A generalization of Schur--Weyl duality known as \emph{mixed Schur--Weyl duality} \cite{grinko2022linear,grinko2023gelfandtsetlin}. 
A variant of mixed Schur--Weyl duality partitions the space $(\C^d)\xp{n+1}$ into subspaces that are invariant under the natural $U^{\otimes n} \otimes \bar{U}$ action of $U \in \U{d}$ and the action \eqref{eq:Brauer action} of the matrix algebra $\A^d_{n,1}$. 
Moreover, there exists a unitary basis change $\Usch(n,1) \in \U{d^{n+1}}$ known as \emph{mixed quantum Schur transform} \cite{grinko2023gelfandtsetlin,nguyen2023mixed} that maps the computational basis $\set{\ket{x} \mid x \in [d]^{n+1}}$ of $(\C^d)\xp{n+1}$ to a new basis composed of irreducible representations $\Uirrep{\Lambda}$ and $\Airrep{\Lambda}$ of the aforementioned actions of $\U{d}$ and $\A^d_{n,1}$, respectively:
\begin{equation}
    \label{eq:mSW}
    \Usch(n,1) \colon
    (\C^d)\xp{n+1}
    \to
    \bigoplus_{\Lambda \in \Irr{\A_{n,1}^{d}}} \Airrep{\Lambda} \x \Uirrep{\Lambda}
    \quad\text{where}\quad
    \Airrep{\Lambda} \defeq \C^{\Paths(\Lambda,\Brat)}
    \text{ and }
    \Uirrep{\Lambda} \defeq \C^{\GT(\Lambda)}.
\end{equation}
Here the direct sum ranges over all irreducible representations $\Lambda$ of $\A_{n,1}^d$, see \cref{eq:IrrA}, and $\GT(\Lambda)$ denotes the set of, so-called, Gelfand--Tsetlin patterns of shape $\Lambda$ (they index a natural basis for the $\Lambda$-irrep of the unitary group $\U{d}$, see \cite{grinko2023gelfandtsetlin} for more details).
We will denote the dimensions of $\Airrep{\Lambda}$ and $\Uirrep{\Lambda}$ by
\begin{align}
    d_\Lambda
    &\defeq \dim \Airrep{\Lambda}
    = \abs{\Paths(\Lambda,\Brat)},
    \label{eq:d_lambda} \\
    m_{\Lambda}
    &\defeq \dim \Uirrep{\Lambda}
    = \abs{\GT(\Lambda)}.
    \label{eq:m_lambda}
\end{align}
Note that both irrep dimensions implicitly depend on the local dimension $d$ which is implicit in $\Lambda$.
Moreover, for any path $T \in \Paths(\Lambda,\Brat)$ and for any level $k$ but the last, the dimension $d_{T^{k}}$ coincides with the dimension of the corresponding $T^k$-irrep of the symmetric group $\S_k$.
As Schur transform \eqref{eq:mSW} is a basis transformation, the non-trivial combinatorial identity $d^{n+1} = \sum_{\Lambda \in \Irr{\A_{n,1}^{d}}} d_\Lambda \, m_{\Lambda}$ holds true. 
Any element $\sigma \in \A_{n,1}^d$, such as the generators $\sigma_i$ from \cref{eq:Brauer action}, transform by Schur transform as follows:
\begin{align}
    \Usch(n,1) \; \sigma \; \Usch\ct(n,1)
    &= \bigoplus_{\Lambda \in \Irr{\A_{n,1}^{d}}}
    \psi_\Lambda(\sigma) \x I_{m_{\Lambda}},
    \label{eq:Schur transform of sigma}
\end{align}
where $\psi_\Lambda$ denotes the corresponding $d_\Lambda$-dimensional irrep of $\A_{n,1}^{d}$ defined by \cref{gtbasis:transpositions,gtbasis:contraction,gtbasis:contraction1}.

\begin{figure}
\begin{tikzpicture}[> = latex,
cut/.style = {thick, dashed, blue, rounded corners = 12pt},
  every node/.style = {inner sep = 3pt, anchor = west},
  cut2/.style = {thick, dashed, gray, rounded corners = 12pt},
  every node/.style = {inner sep = 3pt, anchor = west},
 MT/.style = {draw = blue!40, line width = 3pt},
  every node/.style = {inner sep = 1pt}]
  \def\W{2.0cm}
  \def\H{1cm}
  \foreach \i/\p/\q in {0/0/0, 0.7/1/0, 1.5/2/0, 2.5/3/0, 3.5/4/0, 4.5/5/0, 5.7/5/1} {
    \node at (\i*\W,5.1*\H) {$\A^3_{\p,\q}$};
  }
  \draw[dashed] (5*\W,5.4*\H) -- (5*\W,-2.5*\H);
  \node (0)   at (0.0*\W, 0.0*\H) {$\0$};
  \node (1)   at (0.7*\W, 0.0*\H) {$\yd{1}$};
  \node (2)   at (1.5*\W, 1.0*\H) {$\yd{2}$};
  \node (11)  at (1.5*\W,-1.0*\H) {$\yd{1,1}$};
  \node (3)    at (2.5*\W, 2*\H) {$\yd{3}$};
  \node (21)   at (2.5*\W, 0*\H) {$\yd{2,1}$};
  \node (111)  at (2.5*\W,-2*\H) {$\yd{1,1,1}$};
  \node (4)    at (3.5*\W, 3.0*\H) {$\yd{4}$};
  \node (31)   at (3.5*\W, 1.5*\H) {$\yd{3,1}$};
  \node (22)   at (3.5*\W, 0.0*\H) {$\yd{2,2}$};
  \node (211)  at (3.5*\W,-1.5*\H) {$\yd{2,1,1}$};
  \node (1111) at (3.5*\W,-3.2*\H) {$\yd{1,1,1,1}$};
  \node (5)    at (4.5*\W, 4.0*\H) {$\yd{5}$};
  \node (41)   at (4.5*\W, 2.5*\H) {$\yd{4,1}$};
  \node (32)   at (4.5*\W, 1.0*\H) {$\yd{3,2}$};
  \node (311)  at (4.5*\W,-0.5*\H) {$\yd{3,1,1}$};
  \node (221)  at (4.5*\W,-2.0*\H) {$\yd{2,2,1}$};
  \node (2111) at (4.5*\W,-3.2*\H) {$\yd{2,1,1,1}$};
  \begin{scope}[shift={(0,1.2)}]
  \node (5!1)    at (5.7*\W, 3*\H) {$\br{\yd{5},\yd{1}}$};
  \node (4!0)    at (5.7*\W, 0*\H) {$\br{\yd{4},\0}$};
  \node (41!1)   at (5.7*\W, 2*\H) {$\br{\yd{4,1},\yd{1}}$};
  \node (31!0)   at (5.7*\W, -1*\H) {$\br{\yd{3,1},\0}$};
  \node (32!1)   at (5.7*\W, 1*\H) {$\br{\yd{3,2},\yd{1}}$};
  \node (22!0)   at (5.7*\W, -2*\H) {$\br{\yd{2,2},\0}$};
  \node (211!0)  at (5.7*\W,-3*\H) {$\br{\yd{2,1,1},\0}$};
  \end{scope}
  \draw[->] (0) -- (1);
  \draw[->] (1.north east) -- (2);
  \draw[->] (1.south east) -- (11);
  \draw[->] (2.north east) -- (3);
  \draw[->] (2.south east) -- (21);
  \draw[->] (11.north east) -- (21);
  \draw[->] (11.south east) -- (111);
  \draw[->] (3) -- (4);
  \draw[->] (3) -- (31);
  \draw[->] (21) -- (22);
  \draw[->] (21) -- (211);
  \draw[->] (21) -- (31);
  \draw[->] (111) -- (211);
  \draw[->] (111) -- (1111);
  \draw[->] (4) -- (5);
  \draw[->] (4) -- (41);
  \draw[->] (31) -- (41);
  \draw[->] (31) -- (32);
  \draw[->] (31) -- (311);
  \draw[->] (22) -- (32);
  \draw[->] (22) -- (311);
  \draw[->] (211) -- (311);
  \draw[->] (211) -- (221);
  \draw[->] (211) -- (2111);
  \draw[->] (1111) -- (2111);
  \draw[->] (5) -- (5!1.west);
  \draw[->] (5) -- (4!0.north west);
  \draw[->] (41) -- (41!1.west);
  \draw[->] (41) -- (31!0.north west);
  \draw[->] (41) -- (41!1.west);
  \draw[->] (32) -- (32!1.west);
  \draw[->] (32) -- (22!0.north west);
  \draw[->] (32) -- (31!0.west);
  \draw[->] (311) -- (31!0.south west);
  \draw[->] (311) -- (211!0.north west);
  \draw[->] (221) -- (22!0.west);
  \draw[->] (221) -- (211!0.west);
  \draw[->] (2111) -- (211!0.south west);
  \begin{scope}[on background layer]
    \draw[cut, rounded corners = 18pt] (-0.3*\W, 0*\H) -- ++(0*\W, -2.5*\H) -- ++(6.5*\W, 0*\H) -- ++(0*\W, 7.2*\H) -- ++(-6.5*\W, 0*\H)-- ++(0*\W, -4.7*\H);
    \node[text=blue] (B)  at (0.1*\W,-2*\H) {\LARGE $\Brat$};
    \draw[cut2, rounded corners = 18pt] (-0.4*\W, 0*\H) -- ++(0*\W, -3.7*\H) -- ++(6.7*\W, 0*\H) -- ++(0*\W, 8.5*\H) -- ++(-6.7*\W, 0*\H)-- ++(0*\W, -4.7*\H);
    \node[text=gray] (B)  at (0.1*\W,-3.2*\H) {\LARGE $\widetilde{\Brat}$};
  \end{scope}
  \draw [decorate,decoration={brace,amplitude=8pt,mirror,raise=4pt},yshift=0pt,thick]
(6.4*\W, -2.5*\H) -- (6.4*\W, 1.4*\H) node [black,midway,xshift=1.3cm] {$\Lambda=(\lambda,\0)$};
\draw [decorate,decoration={brace,amplitude=8pt,mirror,raise=4pt},yshift=0pt,thick]
(6.4*\W, 1.6*\H) -- (6.4*\W, 4.5*\H) node [black,midway,xshift=1.3cm] {$\Lambda=(\mu,\yd[0.5]{1})$};
\end{tikzpicture}
\caption{The Bratteli diagram $\Brat$ and the extended Bratteli diagram $\widetilde{\Brat}$ associated with the algebra $\A_{5,1}^3$. The extension of $\Brat$ to $\widetilde{\Brat} $ allows the implementation of optimal POVMs as PVMs, see \cref{sec:pgm:dilation}. 
Vertices in the last column correspond to the set $\Irr{\A_{5,1}^3}$ of irreducible representations of $\A_{5,1}^3$. 
A vector space corresponding to an irrep $\Lambda\in \Irr{\A_{5,1}^3}$ is spanned by $\Paths(\Lambda,\Brat)$ the set of all paths terminating at $\Lambda$, \cref{gtbasis:transpositions,gtbasis:contraction,gtbasis:contraction1} describe the action of the generators $\sigma_i$ of the algebra $\A_{5,1}^3$ in corresponding irrep. 
Irreps $\Lambda\in \Irr{\A_{5,1}^3}$ are of two different types: either $\Lambda=(\lambda,\0)$, or $\Lambda = (\mu,\yd[0.5]{1})$. Effects of POVMs implementing optimal measurements for PBT protocols corresponding to successful teleportation outcomes are supported only on irreps corresponding to $\Lambda=(\lambda,\0)$. 
}
\label{fig:A51}
\end{figure}

The mixed quantum Schur transform \eqref{eq:mSW} is a basis change that maps the computational basis into the mixed Schur basis and hence is a unitary operator. The mixed Schur basis is labelled by pairs of vectors $\ket{T,M}$, where $M$ is a Gelfand–Tsetlin pattern and $T=(T^1,\dotsc,T^{n+1}) \in \Paths(\Brat)$ is a path in the Bratteli diagram $\Brat$. Notice that the set of all paths in $\Brat$ is not a Cartesian product of allowed vertices in each level, i.e., $\Paths(\Brat) \neq \Irr{\A_{0,0}^d} \times \cdots \times \Irr{\A_{n,1}^d}$, as consecutive vertices in a valid path $T \in \Paths(\Brat)$ must differ only by one cell. However, the set of all paths in $\Brat$ is naturally embedded in the set $\Irr{\A_{0,0}^d} \times \cdots \times \Irr{\A_{n,1}^d}$. We call the mixed quantum Schur transform based on this embedding a \emph{mixed Schur isometry}
\begin{equation}
    \label{eq:quantum_schur_domain_range}
    U_{\text{Sch}} (n,1) \colon (\C^d)\xp{n+1} \to
    \underbrace{\C^{\Irr{\A_{0,0}^d}} \x \cdots \x \C^{\Irr{\A_{n,1}^d}}  }_T
    \otimes
    \underbrace{\C^{\GT((n,1),d)}}_{\gtM{d-1}},
\end{equation}
to distinguish it from the unitary transformation \eqref{eq:mSW}, where $\gtM{d-1}$ denotes the Gelfand--Tsetlin pattern $M$ without the top $d$-th row (the subscript means that we keep only the bottom $d-1$ rows of $M$).
In \cite{grinko2023gelfandtsetlin}, we described a quantum circuit implementing the mixed Schur isometry, which for any computational basis vector outputs the corresponding superposition of paths $T \in \Paths(\Brat)$ and Gelfand--Tsetlin patterns without top row $\gtM{d-1} \in \GT((n,1),d)$, where $\GT((n,1),d) \defeq \Set{ \gtM{d-1} \given M \in \bigsqcup_{\Lambda \in \Irr{\A_{n,1}^{d}}}\GT(\Lambda) }$. 
The complexity of implementing the mixed quantum Schur transform is $\widetilde{O}(nd^4)$ \cite{nguyen2023mixed,grinko2023gelfandtsetlin}. 
Notice from \cref{eq:quantum_schur_domain_range} that a path $T \in \Paths(\Brat)$ is stored as a tensor product state
\begin{equation}
    \ket{T} = \ket{T^2} \otimes \cdots \otimes  \ket{T^{n+1}},
    \label{eq:T standard encoding}
\end{equation}
where we have suppressed the registers $\ket{T^0}$ and $\ket{T^1}$ since they are one-dimensional ($T^0 = \0$ and $T^1 = \square$ for any path $T$). $T^{n+1}$ represents the irrep label, which we usually refer to as $T^{n+1} = \Lambda$.
We call \cref{eq:T standard encoding} the \emph{standard encoding} of $\ket{T}$.

We will also consider another more space-efficient isometry implementing the mixed quantum Schur transform. Notice that for a given path $T = (T^0, \dotsc, T^{n+1}) \in \Paths(\Lambda,\Brat)$, the vertex $T^i$ is uniquely determined by the previous vertex $T^{i-1}$ and the row number $y_i$ of the added (or removed) box $T^i \setminus T^{i-1}$. The sequence $(y_1,\dotsc,y_{n+1})$ is called the \emph{Yamanouchi word} of path $T$.
Since $y_i \in [d]$ for each $i$, encoding a path $T$ as a sequence of $y_i$ instead of $T^i$ is more space-efficient.
Indeed, each $y_i$ can be stored directly in the $i$-th input qudit without requiring additional memory, while storing each $T^i$ takes $O(d \log n)$ qubits and thus $O(n d \log n)$ auxiliary qubits for $T$ in total.
We call this more efficient encoding of the mixed Schur transform the \emph{Yamanouchi encoding}:
\begin{equation}
    \label{eq:quantum_schur_domain_range_yamanouchi}
    U_{\text{Sch}} (n,1)   \colon
        (\C^d)^{n+1} \to
    \underbrace{\C^{d} \x \cdots \x \C^{d}  }_{(y_1,\dotsc,y_{n+1})}
    \x
    \underbrace{\C^{\Irr{\A_{n,1}^d}}}_\Lambda
    \otimes
    \underbrace{ \C^{\GT((n,1),d)}}_{\gtM{d-1}}.
\end{equation}
More specifically, we store a path $T \in \Paths(\Brat)$ as the following tensor product state:
\begin{equation}
\label{eq:Yamanouchi_encoding}
    \ket{T} = \ket{y_2} \otimes \cdots \otimes \ket{y_n} \otimes \ket{y_{n+1}} \otimes \ket{\Lambda},
\end{equation}
where the first register $\ket{y_1}$ is suppressed since it is one-dimensional ($y_1 = 1$ for any path).
While the Yamanouchi encoding is more space-efficient, it makes certain operations less time-efficient.
For example, to recover the $i$-th vertex $T^i$ of a path $T$, one needs to perform a certain computation on $y_1, \dotsc, y_i$ stored in the first $i$ registers of $\ket{T}$, as opposed to directly looking up the $i$-th register $\ket{T^i}$ in the standard encoding \eqref{eq:T standard encoding}.

\subsection{PBT measurement and figures of merit}
\label{sec:F and p}

In this section, we briefly describe the general setting of PBT and the relevant figures of merit.
In a PBT protocol two parties, Alice and Bob, share a resource state distributed among $n$ quantum systems called ports, each of local dimension $d$.
We denote Alice's ports by $A_1,\dotsc,A_n$ and Bob's ports by $B_1,\dotsc,B_n$.
Alice has an additional $(n+1)$-th register $\AA$ on her side that contains an unknown input state $\ket{\psi}_{\AA} \in \C^d$ that she must teleport to Bob.

The main ingredient of every PBT protocol is a measurement that is performed by Alice on all her registers $A_1,\dotsc,A_n,\AA$.
The rest of the protocol consists of Alice transmitting the measurement outcome to Bob who uses it to locate the teleported state on one of his ports.
The most general measurement in quantum mechanics is known as \textit{positive operator-valued measure} (POVM), which is described mathematically as follows.

\begin{definition}
A positive operator-valued measure (POVM) is a set $E \defeq \set{E_k}_{k=1}^n$ of positive semi-definite matrices $E_k \succeq 0$ that sum to the identity matrix: $\sum_{k=1}^n E_k = I$.
\end{definition}

In dPBT, Alice's measurement has $n$ outcomes and the outcome $k \in [n]$ indicates the port where Bob should find the teleported state $\ket{\psi}$.
Probabilistic protocols, such as mPBT and pPBT, have an additional outcome $k=0$ corresponding to the failure of the teleportation task.

We can describe a general PBT protocol in terms of a quantum channel $\N$:
\begin{equation}
    \N_{\AA \to \BB}(\rho) \defeq
    \sum_{k=1}^n
    \Tr_{A^n \AA B_k'} \sof*{
      \of*{(\sqrt{E_k})_{A^n \AA} \x I_{B^n}}
      \of*{\Psi_{A^nB^n} \x \rho_{\AA}}
      \of*{\sqrt{E_k}_{A^n \AA} \x I_{B^n}}
    },
    \label{eq:N}
\end{equation}
where
$A^n \defeq A_1 \dots A_n$ and
$B^n \defeq B_1 \dots B_n$
denote Alice's and Bob's ports,
$B_k' \defeq B^n \setminus B_k$ denotes all of Bob's ports but the $k$-th,
$\Psi_{A^n B^n}$ is the resource state shared between Alice and Bob,
$\rho_{\AA}$ is the state to be teleported,
and $E_k$ are Alice's POVM operators.
Note that the sum in \cref{eq:N} omits the value $k = 0$, hence $\N$ is trace-decreasing if Alice's POVM contains a failure operator $E_0 \neq 0$.

The following two figures of merit are commonly used for characterizing the performance of PBT.
The \emph{entanglement fidelity} of the protocol is given by
\begin{equation}
    F \defeq \Tr \sof[\Big]{
      \Phi^+_{\BB R} \of{\N_{\AA \to \BB} \x I_R} \sof*{\Phi^+_{\AA R}}
    }
\end{equation}
where $\Phi^+_{\AA R}$ denotes the two-qudit maximally entangled state between $\AA$ and a reference system $R$.
The \emph{average success probability} is given by
\begin{equation}
    \psucc \defeq \Tr \sof[\big]{\N_{\AA \to \BB}(I/d)},
\end{equation}
which can be less than $1$ since in general $\N$ is trace-decreasing. 

Next, we describe the optimal measurements for different types of PBT \cite{studzinski2017port,mozrzymas2018optimal,leditzky2020optimality,strelchuk2023minimal} in the mixed Schur basis (or, equivalently, the Gelfand--Tsetlin basis). The key role in optimal POVM constructions is played by a special type of POVM called \emph{pretty good measurement} (PGM) \cite{PGM}. We describe this \emph{standard PGM} in \Cref{sec:std_pgm}. In \Cref{sec:povms_dpbt_ppbt}, we describe optimal measurements for dPBT and pPBT using the standard PGM from \Cref{sec:std_pgm}.

\subsection{The standard PGM}
\label{sec:std_pgm}

The fundamental POVM for PBT protocols is of a pretty good measurement type. We denote this POVM by $E = \set{E_k}_{k=0}^{n}$ and call it the \emph{standard PGM} for PBT. It is used in mPBT for both types of resource states \cite{strelchuk2023minimal} and it is defined as follows:
\begin{align}
    \label{def:PGM_PBT}
    E_k &\defeq \rho^{-1/2} \rho_k \rho^{-1/2} \text{ for every $k \in [n]$}, \quad \rho \defeq \sum_{k=1}^n \rho_k, \quad \rho_k \defeq \pi^{k} \sigma_n \pi^{-k}, \quad E_0 \defeq I-\sum_{k=1}^{n} E_k,
\end{align}
where $\rho^{-1}$ is the generalized inverse of $\rho$,
\begin{equation}
    \pi \defeq \sigma_{1}\sigma_{2}\dots\sigma_{n-2}\sigma_{n-1} \in \A^d_{n,1}
\end{equation}
is the cyclic shift permutation $(1\,2\,\dotsc\,n)$ on the first $n$ systems, and $\sigma_n \in \A^d_{n,1}$ is the contraction between systems $n$ and $n+1$ (it corresponds to the un-normalized projection onto the maximally entangled state between them). Since $\rho$ commutes with $\A_{n,0}^d$ and thus with $\pi$, the POVM elements $E_k$ for $k \in [n]$ can be written as
\begin{equation}
    E_k = \pi^{k} E_n \pi^{-k}
    \quad\text{where}\quad
    E_n = \rho^{-1/2} \sigma_n \rho^{-1/2},
    \label{eq:Ek}
\end{equation}
hence the standard PGM is group-covariant \cite{Decker2004} with respect to the cyclic group on $n$ elements.

The key to finding an efficient implementation of the measurement \eqref{def:PGM_PBT} is expressing the operators $E_k$ in the mixed Schur basis \eqref{eq:Schur transform of sigma}.
This requires deriving an explicit formula for $\psi_\Lambda(E_k)$ for any irrep $\Lambda \in \Irr{\A_{n,1}^d}$. Let's express $\rho$ operator in the mixed Schur basis first.

Since $\rho$ coincides with the shifted \emph{Jucys--Murphy element} $d-J_{n+1}$ of $\A_{n,1}^d$, its spectrum can be easily computed, see \cite{grinko2023gelfandtsetlin}.
More concretely, $\rho$ is nonzero and diagonal in the Gelfand--Tsetlin basis of each irrep $(\lambda,\0) \in \Irr{\A_{n,1}^d}$:
\begin{equation}
    \label{def:rho}
    \psi_{(\lambda,\0)}(\rho) = \sum_{a \in \AC_d(\lambda)} \of[\big]{d + \cont(a)} \Pi_{\lambda,a}, \quad \Pi_{\lambda,a} \defeq \sum_{S \in \Paths_{n-1}(\lambda,\Brat)} \proj{S \circ (\lambda \cup a) \circ (\lambda,\0)},
\end{equation}
and $\psi_{(\mu,\square)}(\rho) = 0$ for every $(\mu,\square) \in \Irr{\A_{n,1}^d}$. In particular, $E_0 E_k = 0$ for every $k \in [n]$ or, more precisely, we have for every $k \in [n]$:
\begin{align}
    \psi_{(\mu,\square)}(E_k) &= 0 \text{ for every } (\mu,\square) \in \Irr{\A_{n,1}^d}, \qquad
    \psi_{(\lambda,\0)}(E_0) = 0 \text{ for every } (\lambda,\0) \in \Irr{\A_{n,1}^d}.
\end{align}
Next, due to \cref{gtbasis:contraction1} the generator $\sigma_n$ in the Gelfand--Tsetlin basis of any irrep $(\lambda,\0) \in \Irr{\A_{n,1}^d}$ can be written as
\begin{equation}
    \label{def:sigma_irrep}
    \psi_{(\lambda,\0)}(\sigma_n) = \sum_{S \in \Paths_{n-1}(\lambda,\Brat)} \proj{v_{S,\lambda}}, \quad \text{where} \quad \ket{v_{S,\lambda}} \defeq \sum_{a \in \AC_d(\lambda)} \sqrt{\frac{m_{\lambda \cup a}}{ m_{\lambda}}} \, \ket{S \circ (\lambda \cup a) \circ (\lambda,\0)}
\end{equation}
and $\psi_{(\mu,\square)}(\sigma_n) = 0$ for every $(\mu,\square) \in \Irr{\A_{n,1}^d}$. Also note that due to \cite[Lemma B.1]{grinko2023gelfandtsetlin} for every $\lambda \pt_d n-1$ and $a \in \AC(\lambda)$ we have 
\begin{equation}
    \label{eq:cont_ratios}
    n \cdot \frac{d_{\lambda}}{m_{\lambda}} \cdot \frac{m_{\lambda \cup a}}{d_{\lambda \cup a}} = d + \cont(a).
\end{equation}
Using \cref{eq:cont_ratios,def:rho,def:sigma_irrep}, we can rewrite $E_n$ from \cref{eq:Ek} in the Gelfand--Tsetlin basis of irrep $\Lambda = (\lambda,\0) \in \Irr{\A_{n,1}^d}$ as follows:
\begin{align}
    \label{eq:pbt_main}
    \psi_{(\lambda,\0)}(E_n) = \psi_{(\lambda,\0)}(\rho^{-1/2} \sigma_n \rho^{-1/2}) &= \sum_{S \in \Paths_{n-1}(\lambda,\Brat)} \of*{\psi_{(\lambda,\0)}(\rho)}^{-1/2}\, \proj{v_{S,\lambda}}\, \of*{\psi_{(\lambda,\0)}(\rho)}^{-1/2} \\
    &= \sum_{S \in \Paths_{n-1}(\lambda,\Brat)} \proj{w_{S,\lambda}},
\end{align}
where for $\lambda \pt_d n-1$ and $S \in \Paths_{n-1}(\lambda,\Brat)$ we defined a vector $\ket{w_{S,\lambda}}$ in the irrep $(\lambda,\0) \in \Irr{\A_{n,1}^d}$ as
\begin{equation}
    \ket{w_{S,\lambda}} \defeq \sum_{a \in \AC_d(\lambda)} \sqrt{\frac{d_{\lambda \cup a}}{n \cdot d_{\lambda}}} \ket{S \circ \of{\lambda \cup a} \circ (\lambda,\0)}.
\end{equation}
Summarizing everything above, we can write the standard PGM $E$ in the Gelfand--Tsetlin basis for every irrep $\Lambda \in \Irr{\A_{n,1}^d}$ and every $k \in [n]$ as follows:
\begin{align}
    \label{def:PGM_PBT_schur}
    \psi_{\Lambda}(E_0) &=
    \begin{cases}
        I &
        \text{if $\Lambda = (\mu,\square)$}, \\
        0 &
        \text{if $\Lambda = (\lambda,\0)$},
    \end{cases} \quad
    \psi_{\Lambda}(E_k) =
    \begin{cases}
        0 &
        \text{if $\Lambda = (\mu,\square)$}, \\
        \psi_{(\lambda,\0)}(\pi^{k}E_n\pi^{-k}) &
        \text{if $\Lambda = (\lambda,\0)$},
    \end{cases} \\ \nonumber
    \psi_{(\lambda,\0)}(E_n) &= \sum_{S \in \Paths_{n-1}(\lambda,\Brat)} \proj{w_{S,\lambda}}, \quad \ket{w_{S,\lambda}} = \sum_{a \in \AC_d(\lambda)} \sqrt{\frac{d_{\lambda \cup a}}{n \cdot d_{\lambda}}} \ket{S \circ \of{\lambda \cup a} \circ (\lambda,\0)}.
\end{align}

\subsection{POVMs for deterministic and probabilistic PBT}\label{sec:povms_dpbt_ppbt}

We are now ready to describe the optimal POVMs for dPBT and pPBT from \cite{studzinski2017port,mozrzymas2018optimal,leditzky2020optimality} in the Gelfand--Tsetlin basis. They are closely related to the standard PGM $E$ from \cref{def:PGM_PBT}.

\subsubsection{POVM for deterministic PBT}

In the case of deterministic PBT there is no failure outcome corresponding to $k=0$. The optimal POVMs for optimized and EPR resource states turn out to be the same \cite{leditzky2020optimality,mozrzymas2018optimal} and are closely related to the standard PGM $E$ from \cref{def:PGM_PBT,def:PGM_PBT_schur}. This POVM $E^{\hollowstar} = \set{E^{\hollowstar}_k}_{k=1}^{n}$ is defined as follows:
\begin{align}
    \label{def:dpbt_povm}
    E^{\hollowstar}_k \defeq E_k + \frac{E_0}{n},
\end{align}
and in the Gelfand--Tsetlin basis for every irrep $\Lambda \in \Irr{\A_{n,1}^d}$ and every $k \in [n]$ it is given as
\begin{align}
    \label{def:dpbt_povm_schur}
    \psi_{\Lambda}(E^{\hollowstar}_k) &=
    \begin{cases}
        \dfrac{I}{n} &
        \text{if $\Lambda = (\mu,\square)$ where $\mu \pt_{d-1} n$}, \\
       \psi_{(\lambda,\0)}(E_k) &
        \text{if $\Lambda = (\lambda,\0)$ where $\lambda \pt_d n-1$},
    \end{cases} 
\end{align}
where $\psi_{(\lambda,\0)}(E_k)$ are given in \cref{def:PGM_PBT_schur}.

\subsubsection{Generic PBT measurements}

To describe measurements for pPBT, we first need to explain how one can parameterize \emph{generic} PBT POVMs. Thanks to \cite[Propositions~1.7 and~3.4]{christandl2021asymptotic}, we can motivate the following definition: a \emph{generic} PBT measurement $E^{\star} = \set{E^{\star}_k}_{k=0}^{n}$ is defined in the Gelfand--Tsetlin basis as follows for $k \in [n]$:
\begin{align}
    \label{def:povm_generic}
    E^{\star}_k \defeq \sqrt{G} E^{\hollowstar}_k \sqrt{G}, \quad E^{\star}_0 \defeq I - G, \quad I \succeq G \succeq 0
\end{align}
for some choice of $G \in \A_{n,1}^d$ such that $G$ commutes with $\A_{n,0}^d$. Equivalently, $G \in \A_{n,1}^d$ is a diagonal matrix in the Gelfand--Tsetlin basis:
\begin{align}
    \psi_{\Lambda}(G) \defeq
    \begin{cases}
        g_{\mu} I &
        \text{if $\Lambda = (\mu,\square)$ where $\mu \pt_{d-1} n$}, \\
        \sum_{a \in \AC_d(\lambda)} g_{\lambda,a} \Pi_{\lambda,a} &
        \text{if $\Lambda = (\lambda,\0)$ where $\lambda \pt_d n-1$},
    \end{cases}
    \label{def:G_lambda_povm_generic}
\end{align}
where the projectors $\Pi_{\lambda,a}$ are defined in \cref{def:rho}. Diagonal entries of matrices $\psi_{\Lambda}(G)$
\begin{align}
    g_{\lambda,a} &\defeq \bra{S \circ (\lambda \cup a) \circ (\lambda,\0)} \psi_{(\lambda,\0)}(G) \ket{S \circ (\lambda \cup a) \circ (\lambda,\0)}, \\
    g_\mu &\defeq \bra{T \circ (\mu,\square)} \psi_{(\mu,\square)}(G) \ket{T \circ (\mu,\square)}
\end{align} for every $\lambda \pt_{d} n-1, \, S \in \Paths_{n-1}(\lambda,\Brat), \, a \in \AC_{d}(\lambda)$ and for every $\mu \pt_{d-1} n, \, T \in \Paths_{n}(\mu,\Brat)$ respectively satisfy 
\begin{equation}\label{eq:diag_G_condition}
    1 \geq g_{\mu} \geq 0, \qquad 1 \geq g_{\lambda,a} \geq 0.
\end{equation} 

In particular, the POVM $E^{\hollowstar}$ of dPBT corresponds to $G = I$, i.e. $\psi_{\Lambda}(G) = I$ for every $\Lambda \in \Irr{\A_{n,1}^d}$, and the standard PGM $E$ corresponds to
\begin{align}
    \psi_{\Lambda}(G) =
    \begin{cases}
        0 &
        \text{if $\Lambda = (\mu,\square)$ where $\mu \pt_{d-1} n$}, \\
        I &
        \text{if $\Lambda = (\lambda,\0)$ where $\lambda \pt_d n-1$}.
    \end{cases}
    \label{def:G_lambda_std_pgm}
\end{align}

\subsubsection{POVMs for probabilistic PBT}

The optimal pPBT measurement for optimized resource state turns out to be equal to the standard PGM $E$ from \cref{def:PGM_PBT,def:PGM_PBT_schur}, see \cite{studzinski2017port}.\footnote{Note that the statements regarding optimal POVMs for PBT protocols with optimized resource states in the original papers \cite{studzinski2017port,mozrzymas2018optimal} refer to so-called "dressed" effect operators, not the true effects $E^{\star}_i$.} 

In contrast, the optimal measurement for pPBT with EPR resource state is not the standard PGM $E$, but it is still closely related. According to \cite{studzinski2017port}, the optimal POVM for pPBT with EPR resource state corresponds to the POVM $E^{\star}$ from \cref{def:povm_generic} with the operator $G \in \A_{n,1}^d$ defined as
\begin{align}\label{def:G_lambda_ppbt}
    \psi_{\Lambda}(G) \defeq
    \begin{cases}
        0 &
        \text{if $\Lambda = (\mu,\square)$ where $\mu \pt_{d-1} n$}, \\
        \dfrac{\psi_{(\lambda,\0)}(\rho)}{d + \lambda_1} &
        \text{if $\Lambda = (\lambda,\0)$ where $\lambda \pt_d n-1$},
    \end{cases}
\end{align}
where $\psi_{\Lambda}(\rho)$ is given in \cref{def:rho}. In other words, the matrix $\psi_{\Lambda}(G)$ is diagonal for every $\Lambda \in \Irr{\A_{n,1}^d}$ and its diagonal entries corresponding to \cref{def:G_lambda_povm_generic} can be written as 
\begin{equation}\label{def:G_lambda_a_ppbt}
    g_{\lambda,a} = \frac{d + \cont(a)}{d + \lambda_1}, \qquad g_{\mu} = 0.
\end{equation}
Note that $\lambda_1$ is the highest possible content among $a \in \AC_d(\lambda)$, so \cref{eq:diag_G_condition} is satisfied.

\subsection{Naimark dilations and implementation of measurements}\label{sec:easy_facts}

In this section, we state some facts about projection-valued measures, which we will need in the later sections to work with \emph{Naimark dilations} of PBT POVMs. Naimark dilation is a realization of a given POVM as a projection-valued measure in a larger, so-called \emph{dilated}, Hilbert space. First, recall that

\begin{definition}
    $\Pi \defeq \set{\Pi_i}_{i=1}^n$ is a projection-valued measure (PVM) on a given Hilbert space $\mathcal{H}$ if $\sum_{i=1}^n \Pi_i = I_{\mathcal{H}}$ and $\Pi_i$ is an orthogonal projection, i.e. $\Pi_i^2 = \Pi_i = \Pi_i^\dagger \succeq 0$, for every $i \in [n]$.
\end{definition}

\begin{remark} If $\set{\Pi_i}_{i=1}^n$ is a PVM then for every pair $i,j \in [n], \, i \neq j$ the projectors $\Pi_i$ and      $\Pi_j$ are mutually orthogonal, i.e. $\Pi_i \Pi_j = \delta_{i,j} \Pi_i$. Indeed, take a vector $\ket{v} \in \mathrm{im}(\Pi_i) \subseteq \mathcal{H}$, and since $\sum_{j} \Pi_j = I_{\mathcal{H}}$ we have $1 + \sum_{j \neq i} \bra{v} \Pi_j \ket{v} = 1$.
Therefore, since each $\Pi_j$ is positive semidefinite it must be $\bra{v} \Pi_j \ket{v} = 0$ for every $j \neq i$, which implies $\Pi_i \Pi_j = \delta_{i,j} \Pi_i$. The converse is also true: if every pair of operators in a given POVM are mutually orthogonal then this POVM is a PVM.
\end{remark}

We can now describe an easy example of a Naimark dilation, which we will need later to implement the measurement for dPBT.

\begin{lemma}\label{lem:dpbt_dilation}
    Let $E \defeq \set{E_i}_{i=1}^n$ be a POMV on a Hilbert space $\mathcal{H}$ such that $E_i = \frac{I_{\mathcal{H}}}{n}$ for every $i \in n$. Then its dilation on the Hilbert space $\hat{\mathcal{H}} \defeq \C^n \otimes \mathcal{H}$ is given by $\Pi_i = \proj{i} \otimes I_{\mathcal{H}}$ and isometry $V : \mathcal{H} \rightarrow \hat{\mathcal{H}}$ defined as $V\ket{\psi} = \ket{+} \otimes \ket{\psi}$, where $\ket{+} = \frac{1}{\sqrt{n}} \sum_{i=1}^n \ket{i}$. 
\end{lemma}
\begin{proof}
    The isometry $V : \ket{\psi} \mapsto \ket{+} \otimes \ket{\psi}$ means that the input density matrix on $\mathcal{H}$ is embedded in $\hat{\mathcal{H}}$ as $\rho \otimes \proj{+}$. It is easy to see that $\Tr[(\rho \otimes \proj{+}) \Pi_i] = \Tr[\rho E_i]$.
\end{proof}


Now assume that for a certain POVM with $n+1$ outcomes $k \in \{0,1,\dotsc,n\}$ we managed to dilate outcomes $k \in [n]$ to some PVM $\Pi$. Next lemma explains how to dilate a ``leftover'' POVM $E$, related to $\Pi$ in a certain way. This result will be used in \cref{sec:pPBT_circuits}.

\begin{lemma}\label{lem:ppbt_dilation}
    Let $\Pi \defeq \set{\Pi_i}_{i=1}^n$ be a PVM on a Hilbert space $\mathcal{H}$ and let $G$ be a positive semidefinite operator on $\mathcal{H}$ such that $G \preceq I_\mathcal{H}$. 
    Suppose that there is also a POVM $E \defeq \set{E_i}_{i=0}^n$ on $\mathcal{H}$ defined as 
    \begin{align}
        E_i &\defeq \sqrt{G} \, \Pi_i \sqrt{G} \text{ for every } i \in [n], \qquad E_0 \defeq I_\mathcal{H} - G.
    \end{align}
    Then one can dilate the POVM $E$ on $\mathcal{H}$ to a PVM $\hat{\Pi} \defeq \set{\hat{\Pi}_i}_{i=0}^n$ on $\hat{\mathcal{H}} \defeq \C^2 \otimes \mathcal{H} $ via isometric embedding $\ket{\psi} \mapsto \ket{0} \otimes \ket{\psi}$ as follows:
    \begin{align}
        \hat{\Pi}_i &\defeq
        \begin{pNiceMatrix}
            \sqrt{G}\, \Pi_i \sqrt{G} & -\sqrt{G}\, \Pi_i \sqrt{I_\mathcal{H} - G}  \\
            - \sqrt{I_\mathcal{H} - G}\, \Pi_i \sqrt{G} & \sqrt{I_\mathcal{H} - G}\, \Pi_i \sqrt{I_\mathcal{H} - G}
        \end{pNiceMatrix} \quad
        \forall\, i \in [n], \quad 
        \hat{\Pi}_0 \defeq
        \begin{pNiceMatrix}
            I_\mathcal{H} - G & \sqrt{G} \sqrt{I_\mathcal{H} - G}  \\
            \sqrt{G} \sqrt{I_\mathcal{H} - G} & G 
        \end{pNiceMatrix}.
    \end{align}
\end{lemma}
\begin{proof}
    It is easy to see that the dilated PVM $\set{\hat{\Pi}_i}_{i=0}^n$ on $\hat{\mathcal{H}}$ has the form
    \begin{equation}
        \hat{\Pi}_i = U \begin{pNiceMatrix} \Pi_i & 0  \\ 0 & 0 \end{pNiceMatrix} U^\dagger \quad \forall\, i\in[n],\, \text{ and }  \, \hat{\Pi}_0 = U \begin{pNiceMatrix} 0 & 0  \\ 0 & I_\mathcal{H} \end{pNiceMatrix} U^\dagger,
    \end{equation}
    where the unitary $U$ is defined as 
    \begin{equation}
        U \defeq
        \begin{pNiceMatrix}
            \sqrt{G} & -\sqrt{I_\mathcal{H} - G}  \\
            \sqrt{I_\mathcal{H} - G} & \sqrt{G}
        \end{pNiceMatrix}.
    \end{equation}
    It is easy to see that $\hat{\Pi}_i^2=\hat{\Pi}_i=\hat{\Pi}_i^\dagger$ for every $i$ and $\sum_{i=0}^{n} \hat{\Pi}_i = I_{\hat{\mathcal{H}}}$. Therefore $\hat{\Pi}$ is indeed a PVM on $\hat{\mathcal{H}}$. This PVM manifestly acts on vectors $\ket{0} \otimes \ket{\psi} \in \hat{\mathcal{H}}$ in the same way as the POVM $E$ acts on $\ket{\psi} \in \mathcal{H}$, so it is indeed a Naimark's dilation of the POVM $E$.
\end{proof}



Now we explain how to implement a PVM of a certain form, which we will later use in \cref{sec:circuit_for_pgm}, assuming one can efficiently implement unitaries which define this PVM. Namely, assume that we have a PVM $\Pi = \set{\Pi_k}_{k=0}^n$ on $\mathcal{H} \defeq \mathcal{H}_1 \otimes \mathcal{H}_2$ for some Hilbert spaces $\mathcal{H}_1,\mathcal{H}_2$ such that
\begin{align}
    \label{def:general_Pi_k}
    \Pi_k &\defeq U_k \Pi_n U_k^\dagger \text{ for every } k \in [n] , \qquad
    \Pi_n \defeq I_{\mathcal{H}_1} \otimes \of*{W \proj{0} W^\dagger}_{\mathcal{H}_2}, \qquad \Pi_0 = I_{\mathcal{H}} - \sum_{k=1}^n \Pi_k
\end{align}
where $U_n \defeq I_{\mathcal{H}}$, and we assume that all $U_k$ are some easy-to-implement unitaries on $\mathcal{H} = \mathcal{H}_1 \otimes \mathcal{H}_2$, $W$ is some easy-to-implement unitary on $\mathcal{H}_2$ and $\ket{0}$ is a computational basis vector on $\mathcal{H}_2$.

We can implement the PVM $\Pi$ as follows:
\begin{itemize}
    \item Define a unitary $V$ on $\mathcal{H}$ as
        \begin{equation}\label{def:general_V}
            V \defeq \sum_{k=0}^{n} \omega^{k}_{n+1} \Pi_k,
        \end{equation}
    where $\omega_{n+1}$ is root of unity of order $n+1$. Observe, that we can implement the unitary $V$ efficiently thanks to our assumptions via the following circuit on $\mathcal{H} = \mathcal{H}_1 \otimes \mathcal{H}_2$:
    \begin{equation*}
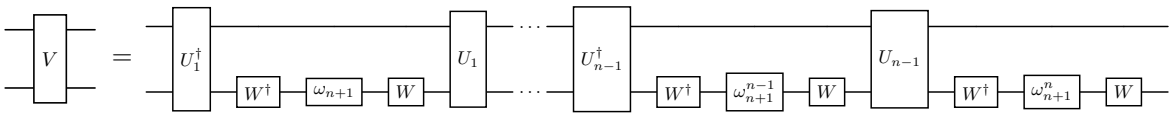
        \begin{adjustbox}{width=0.08\textwidth}
        \begin{quantikz}
        &\gate[2]{V} & \\
        & &
        \end{quantikz}
        \end{adjustbox}
        \hspace{0.1em}=
        \begin{adjustbox}{width=0.8\textwidth}
        \begin{quantikz}
        &\gate[2]{U^\dagger_1} &                  &                     &          & \gate[2]{U_1}  & \HD &\gate[2]{U^\dagger_{n-1}} &                  &                           &          &\gate[2]{U_{n-1}} &                   &                       &          & \\
        &                      & \gate{W^\dagger} & \gate{\omega_{n+1}} & \gate{W} &                & \HD &                          & \gate{W^\dagger} & \gate{\omega^{n-1}_{n+1}} & \gate{W} &                  &  \gate{W^\dagger} & \gate{\omega^n_{n+1}} & \gate{W} &
        \end{quantikz}
        \end{adjustbox}
    \end{equation*}
    The gate \begin{quantikz}[scale=0.6]\gate{\omega^{k}_{n+1}}\end{quantikz} represents the operator $\omega^{k}_{n+1} \proj{0} + \of*{I-\proj{0}}$ on $\mathcal{H}_2$. We also used $U_n = I_{\mathcal{H}}$.
    
    \item Note that implementing $V^i$ is also easy:
    \begin{equation*}
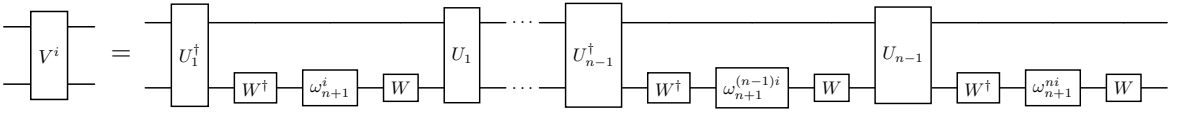
        \begin{adjustbox}{width=0.08\textwidth}
        \begin{quantikz}
        &\gate[2]{V^i} & \\
        & &
        \end{quantikz}
        \end{adjustbox}
        \hspace{0.1em}=
        \begin{adjustbox}{width=0.8\textwidth}
        \begin{quantikz}
        &\gate[2]{U^\dagger_1} &                  &                       &          & \gate[2]{U_1}  & \HD &\gate[2]{U^\dagger_{n-1}} &                  &                              &          &\gate[2]{U_{n-1}} &                   &                          &          & \\
        &                      & \gate{W^\dagger} & \gate{\omega^i_{n+1}} & \gate{W} &                & \HD &                          & \gate{W^\dagger} & \gate{\omega^{(n-1)i}_{n+1}} & \gate{W} &                  &  \gate{W^\dagger} & \gate{\omega^{ni}_{n+1}} & \gate{W} &
        \end{quantikz}
        \end{adjustbox}
    \end{equation*}
    
    \item Now we can run the phase estimation circuit to measure a given state $\ket{\psi} \in \mathcal{H}$ with the PVM $\Pi$:
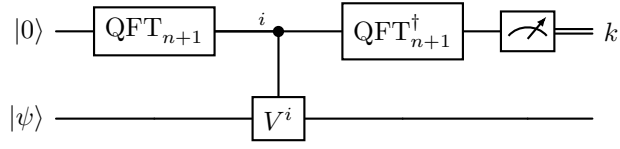
    \begin{equation*}
    \begin{quantikz}
        \lstick{$\ket{0}$}  & \gate{\mathrm{QFT}_{n+1}} & \ctrl{1} \wire[l][1]["i"{above,pos=0.2}]{a} & \gate{\mathrm{QFT}_{n+1}^\dagger} & \meter{} & \setwiretype{c} \rstick{$k$} \\
        \lstick{$\ket{\psi}$} & & \gate{V^i} & & &
    \end{quantikz}
    \end{equation*}
    Indeed, this circuit implements the following unitary evolution for every input state $\ket{\psi} \in \mathcal{H}$:
    \begin{align}
        \ket{\psi}\ket{0} \mapsto  \sum_{i=0}^n \tfrac{1}{\sqrt{n+1}} \ket{\psi} \ket{i}  \mapsto \sum_{i=0}^n \tfrac{1}{\sqrt{n+1}} V^i\ket{\psi} \ket{i} &= \sum_{i=0}^n \tfrac{1}{\sqrt{n+1}} \sum_{k=0}^{n} \omega^{ik}_{n+1} \Pi_k\ket{\psi} \ket{i} \nonumber \\
        &= \sum_{k=0}^{n} \Pi_k\ket{\psi} \sum_{i=0}^n \tfrac{\omega^{ik}_{n+1} }{\sqrt{n+1}} \ket{i} \mapsto \sum_{k=0}^{n} \Pi_k \ket{\psi} \ket{k},
    \end{align}
    so measuring $k$ returns the projected state $\frac{\Pi_k \ket{\psi}}{\norm{\Pi_k \ket{\psi}}}$ according to the Born rule.
\end{itemize}

\section{Efficient quantum algorithms for PBT in standard encoding}
\label{sec:std_enc}

\subsection{Naimark's dilation of the standard PGM}
\label{sec:pgm:dilation}

Before presenting our circuits, we need to explain how to dilate the POVM $E$ from \cref{def:PGM_PBT} to a projective measurement $\Pi$. In principle, that is possible for any POVM due to \emph{Naimark's dilation} theorem. However, in general it is not obvious how to achieve a conceptually simple dilation that can be efficiently implemented. We first explain how to construct such dilation explicitly, and then in \cref{sec:circuit_for_pgm} present an efficient circuit for $E$.

Recall that the Bratteli diagram of the symmetric group is the \emph{Young lattice}\footnote{Sometimes Young lattice is also called \emph{Young graph}. Young graph usually includes all possible partitions $\nu \pt k$ at a given level $k$. However, our Bratteli diagram $\Brat$ up to level $n$ coincides with a subset of the Young graph: only the partitions $\nu \pt k$ with $\len(\nu) \leq d$ are included in our $\Brat$.} \cite{Sagan}, and the following identity holds for every $\lambda \pt n-1$ \cite{stanley2013algebraic} in the Young lattice:
\begin{equation}
    n \cdot d_\lambda = \sum_{a \in \AC(\lambda)}  d_{\lambda \cup a},
    \label{pbt:naimrak_sym_induction}
\end{equation}
where the notation $\lambda \cup a$ denotes the Young diagram in the Young lattice obtained by adding a box $a$ to $\lambda$, and $d_\lambda$ is the dimension of the symmetric group irrep $\lambda$.\footnote{The dimension $d_\lambda$ can be both understood as the number of paths from the root to a vertex $\lambda$ in the Young lattice as well as in the Bratteli diagram of $\A_{n,0}^d$, since up to level $n$ the Bratteli diagram is a subset of the full Young lattice and the procedure of adding a cell is monotonic with respect to the number of rows in $\lambda$ along a given path in the Young lattice.}

The main observation of this section is that for a Young diagram $\lambda \pt n-1$, we have $\AC_d(\lambda) = \AC(\lambda)$ if $\lambda_d = 0$ and $\AC_d(\lambda) \neq \AC(\lambda)$ if $\lambda_d > 0$. In particular, when $\lambda_d = 0$ this implies that
\begin{equation}
    \norm{\ket{w_{S,\lambda}}}^2 = \sum_{a \in \AC_d(\lambda)} \frac{d_{\lambda \cup a}}{n \cdot d_\lambda} = \sum_{a \in \AC(\lambda)} \frac{d_{\lambda \cup a}}{n \cdot d_\lambda} = 1,
\end{equation}
so $\psi_{(\lambda,\0)}(E_n)$ is an orthogonal projector. Since the cyclic shift $\pi$ acts unitarily, all $\psi_{(\lambda,\0)}(E_i)$ are orthogonal projectors as well. Because $E$ provides a resolution of the identity in the irreducible representation $(\lambda,\0) \in \Irr{\A_{n,1}^d}$, the POVM $E$ restricted to the irreducible representation $(\lambda,\0)$ with $\lambda_d = 0$ is actually a PVM there. We will replace $\psi_{(\lambda,\0)}(E_n)$ by $\psi_{(\lambda,\0)}(\Pi_n)$ from now on to indicate that $E$ is actually a PVM on $(\lambda,\0) \in \Irr{\A_{n,1}^d}$.

However, for the irreps $\lambda$ with $\lambda_d > 0$ the POVM $E^\lambda$ is not a PVM because $\AC(\lambda) = \AC_d(\lambda) \sqcup \set{(d+1,1)}$ and the vectors $\ket{w_{S,\lambda}}$ are not normalized anymore:
\begin{equation}\label{eq:pbt_norm_<1}
    \norm{\ket{w_{S,\lambda}}}^2 = \sum_{a \in \AC_d(\lambda)} \frac{d_{\lambda \cup a}}{n \cdot d_\lambda} = \of[\bigg]{\sum_{a \in \AC(\lambda)} \frac{d_{\lambda \cup a}}{n \cdot d_\lambda}} - \frac{d_{\lambda \cup (d+1,1)}}{n \cdot d_\lambda} = 1 - \frac{d_{\lambda \cup (d+1,1)}}{n \cdot d_\lambda} < 1,
\end{equation}
where $\lambda \cup (d+1,1)$ denotes the Young diagram obtained from $\lambda$ by adding a cell with coordinates $(d+1,1)$, so that $\len \of*{\lambda \cup (d+1,1)}=d+1$. The vertex corresponding to this Young diagram does not exist in the Bratteli diagram $\Brat$. Fortunately, \cref{eq:pbt_norm_<1} suggests immediately how to construct a Naimark's dilation $\Pi^\lambda$ of $E^\lambda$ for $\lambda$ with $\lambda_d > 0$. For this construction, one needs to modify the Bratteli diagram $\Brat$ by adding vertices to each level $\geq n$. Then the set of all paths in this modified Bratteli diagram will define a new basis for the Naimark dilated Hilbert space.\footnote{We have provided a \emph{Wolfram Mathematica} notebook implementing our construction on \href{https://github.com/dgrinko/walledbrauer-gtbasis}{GitHub}.}

More concretely, to each level $k \leq p$ of the Bratteli diagram $\Brat$ we add all possible vertices labelled by all Young diagrams $\nu \pt k$ such that $\nu_{d+1} = 1$. An edge between a pair of Young diagrams in two consecutive levels is added if the latter diagram can be obtained by adding a cell to the previous one. This procedure ensures that all the levels up to $n$ of the new Bratteli diagram form a subset of the Young lattice, such that for every vertex at level $n$ the irrep dimensions still satisfy \cref{pbt:naimrak_sym_induction}. We call the new Bratteli diagram $\widetilde{\Brat}$. The basis for the Naimark dilated Hilbert space for the irrep $\Lambda \in \Irr{\A_{n,1}^{d}}$ consists of all paths from the root to the leaf $\Lambda$ in this modified Bratteli diagram, which we denote by $\Paths(\Lambda,\widetilde{\Brat})$. Formally, for every $\Lambda \in \Irr{\A_{n,1}^{d}}$, if $\Lambda = (\lambda,\0) $ for $\lambda \pt n-1, \, \len(\lambda) \geq d$ we define
\begin{equation}
    \Paths(\Lambda,\widetilde{\Brat}) \defeq
        \set*{ T = (T^{0},T^{1},\dotsc,T^{n},\Lambda) \in \Irr{\A_{1,0}^{d+1}} \times \dotsb \times \Irr{\A_{n,0}^{d+1}} \times \Irr{\A_{n,1}^{d}} \mid T \text{ satisfies \cref{def:bratteli_rule}} },
\end{equation}
where
\begin{equation}\label{def:bratteli_rule}
    T^{k}_{d+1} \leq 1 \; \forall k \, \in [n], \text{ and } T^{k-1} \rightarrow T^{k} \; \forall k \in [n], \text{ and } T^{n} = \lambda \cup a \text{ for some } a \in \AC(\lambda).
\end{equation}
If $\Lambda_r \neq \0$ then
\begin{equation}
     \Paths(\Lambda,\widetilde{\Brat}) \defeq \Paths(\Lambda,\Brat).
\end{equation}
This extension of the Bratteli diagram is illustrated in \cref{fig:A51}.

The action $\widetilde{\psi}_{\Lambda}$ of the transposition generators $\sigma_1,\dotsc,\sigma_{n-1}$ of $\A_{n,1}^d$ in this dilated Bratteli diagram $\widetilde{\Brat}$ is given by the generalization of \cref{gtbasis:transpositions} to all paths in $\Paths(\Lambda,\widetilde{\Brat})$, i.e. for every $T \in \Paths(\Lambda,\widetilde{\Brat})$ and every $i \in [n-1]$ we define:
\begin{align}
     \widetilde{\psi}_{\Lambda}(\sigma_i) \, \ket{T}
    &= \frac{1}{r_i(T)} \, \ket{T} + \sqrt{1 - \frac{1}{r_i(T)^2}} \, \ket{\sigma_i T}
    \qquad
    \text{for $i \neq n$}.
    \label{def:action_dilated_perm}
\end{align}
For this new Bratteli diagram $\widetilde{\Brat}$, we define the dilated versions $\ket{\widetilde{w}_{S,\lambda}}$ of vectors $\ket{w_{S,\lambda}}$ for $S \in \Paths_{n-1}(\lambda,\widetilde{\Brat})$ as
\begin{equation}
    \ket{\widetilde{w}_{S,\lambda}} \defeq \sum_{a \in \AC(\lambda)} \sqrt{\frac{d_{\lambda \cup a}}{n \cdot d_{\lambda}}} \ket{S \circ \of{\lambda \cup a} \circ (\lambda,\0)}.
\end{equation}
Therefore in the dilated space since \cref{pbt:naimrak_sym_induction} holds we have
\begin{equation}
    \norm{\ket{\widetilde{w}_{S,\lambda}}}^2 = \sum_{a \in \AC(\lambda)} \frac{d_{\lambda \cup a}}{n \cdot d_{\lambda}} = 1.
\end{equation}
More importantly, we have the following 
\begin{lemma}
    For every $\Lambda = (\lambda, \0)$ in the dilated Hilbert space spanned by $\Paths(\Lambda,\widetilde{\Brat})$ we have
    \begin{equation}
        \sum_{k=1}^n \sum_{S \in \Paths_{n-1}(\lambda,\widetilde{\Brat})} \widetilde{\psi}_{\Lambda}(\pi^k) \proj{\widetilde{w}_{S,\lambda}} \widetilde{\psi}_{\Lambda}(\pi^{-k}) = I
    \end{equation}
\end{lemma}
\begin{proof}
    Denote 
    \begin{equation}
        A \defeq \sum_{k=1}^n \sum_{S \in \Paths_{n-1}(\lambda,\widetilde{\Brat})} \widetilde{\psi}_{\Lambda}(\pi^k) \proj{\widetilde{w}_{S,\lambda}} \widetilde{\psi}_{\Lambda}(\pi^{-k}).
    \end{equation}
    Note that by construction \cref{def:action_dilated_perm} any $\sigma \in \S_{n-1}$ (we think of $\sigma$ as an element of the algebra $\A_{n-1,0}^d$) commutes with the following element:
    \begin{equation}
        \sum_{S \in \Paths_{n-1}(\lambda,\widetilde{\Brat})} \proj{\widetilde{w}_{S,\lambda}} \widetilde{\psi}_{\Lambda}(\sigma) = 
        \widetilde{\psi}_{\Lambda}(\sigma) \sum_{S \in \Paths_{n-1}(\lambda,\widetilde{\Brat})} \proj{\widetilde{w}_{S,\lambda}},  
    \end{equation}
    since the above element acts as identity on $\Paths_{n-1}(\lambda,\widetilde{\Brat})$.
    Therefore $A$ must commute with $\A_{n,0}^d$ since $\pi^{k}$ are transversals for cosets of $\A_{n,0}^d$ over $\A_{n-1,0}^d$:
    \begin{align}
        A &= \frac{1}{(n-1)!}
        \sum_{\sigma \in \S_{n-1}} \sum_{k=1}^n \sum_{S \in \Paths_{n-1}(\lambda,\widetilde{\Brat})} \widetilde{\psi}_{\Lambda}(\pi^k \sigma) \proj{\widetilde{w}_{S,\lambda}} \widetilde{\psi}_{\Lambda}((\pi^{k}\sigma)^{-1}) \\
        &= \frac{1}{(n-1)!} \sum_{\sigma \in \S_{n}} \sum_{S \in \Paths_{n-1}(\lambda,\widetilde{\Brat})} \widetilde{\psi}_{\Lambda}(\sigma) \proj{\widetilde{w}_{S,\lambda}} \widetilde{\psi}_{\Lambda}(\sigma^{-1}).
    \end{align}
    This means that $A$ is a diagonal matrix and $\bra{T} A \ket{T}$ depends only on $T^n$, so for every $a \in \AC(\lambda)$ and every $T \in \Paths(\Lambda,\widetilde{\Brat})$ with $T^n = \lambda \cup a$ we can write 
    \begin{equation}
        \bra{T} A \ket{T} = \frac{1}{d_{\lambda \cup a}} \sum_{\substack{T \in \Paths(\Lambda,\widetilde{\Brat}) \\ T^n = \lambda \cup a}} \bra{T} A \ket{T} = \frac{1}{d_{\lambda \cup a} (n-1)!} \sum_{\substack{S \in \Paths_{n-1}(\lambda,\widetilde{\Brat}) \\ T \in \Paths(\Lambda,\widetilde{\Brat}) \\ T^n = \lambda \cup a}} \sum_{\sigma \in \S_{n}} \bra{\widetilde{w}_{S,\lambda}} \widetilde{\psi}_{\Lambda}(\sigma^{-1}) \proj{T} \widetilde{\psi}_{\Lambda}(\sigma) \ket{\widetilde{w}_{S,\lambda}},
    \end{equation}
    but in the above formula $\widetilde{\psi}_{\Lambda}(\sigma)$ commutes with $\sum_{T} \proj{T}$. Therefore,
    \begin{equation}
        \bra{T} A \ket{T} = \frac{1}{d_{\lambda \cup a} (n-1)!} \sum_{\sigma \in \S_{n}} \sum_{\substack{S \in \Paths_{n-1}(\lambda,\widetilde{\Brat}) \\ T \in \Paths(\Lambda,\widetilde{\Brat}) \\ T^n = \lambda \cup a}} \abs{\braket{T}{\widetilde{w}_{S,\lambda}}}^2 = \frac{n \cdot d_{\lambda}}{d_{\lambda \cup a}} \frac{d_{\lambda \cup a}}{n \cdot d_{\lambda}} = 1.
    \end{equation}
    Since every diagonal element of $A$ is $1$, so $A=I$.
\end{proof}

Consequently, in the dilated space our POVM $E$ becomes a PVM, which we denote by $\Pi$. From now on assume that we work in the dilated Gelfand--Tsetlin basis spanned by $T \in \Paths(\Lambda,\widetilde{\Brat})$ and we want to implement the PVM $\Pi = \set{\Pi_k}_{k=0}^n$, where for every $k \in [n]$:
\begin{align}\label{def:pbt_dilated_pvm}
    \Pi_k \defeq \pi^{k} \Pi_n \pi^{-k} \text{ and } \widetilde{\psi}_{(\lambda,\0)}(\Pi_n) = \hspace{-1em} \sum_{S \in \Paths_{n-1}(\lambda,\widetilde{\Brat})} \hspace{-10pt} \proj{\widetilde{w}_{S,\lambda}}, \qquad  \Pi_0 = I - \sum_{k=1}^{n}\Pi_k.
\end{align}

\subsection{Quantum circuit for standard PGM}
\label{sec:circuit_for_pgm}

Using the results of \cref{sec:pgm:dilation}, our task now is to implement the PVM $\Pi$ from \cref{def:pbt_dilated_pvm}. We embed $\Paths(\Brat)$ in a Hilbert space with appropriate tensor product structure dictated by the mixed quantum Schur transform \cite{grinko2023gelfandtsetlin}. Mixed quantum Schur transform can be implemented using two different encodings for the Gelfand--Tsetlin basis of $\A_{n,1}^d$: the \emph{standard} encoding or the \emph{Yamanouchi} encoding, see \cite{grinko2023gelfandtsetlin} for more details. In this section, we discuss how to implement $\Pi$ with the standard encoding. The Yamanouchi encoding implementation is explained in the \cref{sec:yaman_enc}.

Before presenting our circuit for $\Pi$, we need to define a unitary $\widetilde{W}$ acting on registers $T^{n-1},T^n,T^{n+1}$, which can be used to prepare the states $\ket{\widetilde{w}_{S,\lambda}} \defeq \ket{S}\ket{\widetilde{w}_{\lambda}}$ for every $(\lambda,\0) \in \Irr{\A_{n,1}^{d}}$ and $S \in \Paths(\Lambda,\widetilde{\Brat}) $. Namely, we first define
\begin{equation}
    \label{def:w_lambda_tilde}
    \widetilde{W}_\lambda \ket{0} \defeq \ket{\widetilde{w}_{\lambda}}, \qquad \ket{\widetilde{w}_{\lambda}} \defeq \sum_{a \in \AC(\lambda)} \sqrt{\frac{d_{\lambda \cup a}}{n \cdot d_{\lambda}}} \ket{\lambda \cup a},
\end{equation}
where $\widetilde{W}_\lambda$ is a unitary matrix of size at most $(d+1) \times (d+1)$ with easy-to-compute entries $\ket{\widetilde{w}_{\lambda}}$ in its first column. Now we define $\widetilde{W}$, a controlled version of $\widetilde{W}_\lambda$, as a unitary which prepares $\ket{\widetilde{w}_{\lambda}}$ conditioned on $\lambda$:
\begin{align}
    \label{def:W_lambda}
    &\widetilde{W} \defeq \sum_{\substack{\lambda \pt_d n - 1 }} \of[\bigg]{ \of[\bigg]{\sum_{\substack{\lambda' \neq \lambda \\ \lambda' \pt_d n - 1}} \proj{\lambda'}} \otimes I + \proj{\lambda} \otimes \widetilde{W}_\lambda} \otimes \proj{(\lambda,\0)} + \sum_{\substack{\mu \pt_{d-1} n }} I \otimes I \otimes \proj{(\mu,\square)}.
\end{align}
This transformation can be implemented via a sequence of Givens rotations. Similarly to \cref{def:w_lambda_tilde}, we can define a unitary $W_\lambda$ which prepares the normalized version of $\ket{w_{\lambda}}$:
\begin{align}
    \label{def:W_lambda_wo_tilde}
    W_\lambda \ket{0} &\defeq \frac{\ket{w_{\lambda}}}{\norm{\ket{w_{\lambda}}}}, \qquad
    \ket{w_{\lambda}} = \sum_{a \in \AC_d(\lambda)} \sqrt{\frac{d_{\lambda \cup a}}{n \cdot d_{\lambda}}} \ket{\lambda \cup a}.
\end{align}
The gate $W$ is now defined for \cref{def:W_lambda_wo_tilde} in the same way as $\widetilde{W}$ in \cref{def:W_lambda}.

Assume the initial state is $\ket{S}\ket{0}\ket{\lambda} \defeq \ket{S^{2}} \dotsc \ket{S^{n-2}} \ket{\lambda} \ket{0} \ket{(\lambda,\0)}$ for arbitrary $S \in \Paths_{n-1}(\lambda,\widetilde{\Brat})$, where $\ket{0}$ is some basis state of the register, corresponding to the $n$-th level of the dilated Bratteli diagram $\widetilde{\Brat}$. Then we can prepare a state $\ket{S}\ket{\widetilde{w}_{\lambda}}\ket{(\lambda,\0)}$ as follows:
\begin{equation}
    \of*{I \otimes \widetilde{W}} \ket{S} \ket{0} \ket{\lambda} = \ket{S}\ket{\widetilde{w}_{\lambda}}\ket{\lambda},
\end{equation}
where identity $I$ acts on the registers $\ket{S^{2}} \dotsc \ket{S^{n-2}}$. Moreover, note that the amplitudes of the state $\ket{\widetilde{w}_{\lambda}}$ are easy to calculate on a classical computer in time $\widetilde{O}(d)$ due to:

\begin{lemma}[{\cite{kosuda2003new}}]\label{lem:sym_irrep_dim_ration_contents}
    For every $\lambda \pt n-1$ and $a \in \AC(\lambda)$ there holds
    \begin{equation}
        \frac{ \prod_{c \in \RC(\lambda)} \of{\cont(a)-\cont(c)}}{\prod_{c \in \AC(\lambda) \setminus a} \of{\cont(a)-\cont(c)}} = \frac{d_{\lambda \cup a}}{n \cdot d_\lambda}.
    \end{equation}
\end{lemma}

Each $\widetilde{W}_\lambda$ gate can be implemented as a sequence of simple controlled $\mathrm{R}_i$ gates, as shown in \cref{fig:W_gate}, which we define as
\begin{equation}\label{eq:R_i_def}
    \mathrm{R}_i \ket{0} \defeq
    \sqrt{
        \frac{1 - \sum_{j=1}^{i} \eta_j}{1 - \sum_{j=1}^{i-1} \eta_j}
    } \; \ket{0}
    + \sqrt{
        \frac{\eta_i}{1 - \sum_{j=1}^{i-1} \eta_j}
    } \; \ket{1},
    \qquad
    \eta_j \defeq \frac{d_{\lambda \cup a_j}}{n \cdot d_\lambda}
\end{equation}
where the cell $a_j \in \AC(\lambda)$ is located in the $j$-th row of $\lambda$. If for a given $j$ there is no such cell then the control on $\mathrm{R}_i$ in \cref{fig:W_gate} is not triggered. In other words, $\mathrm{R}_i$ is triggered only when $\lambda_{i-1} > \lambda_i$ ($\mathrm{R}_1$ is always triggered). Note that due to \cref{lem:sym_irrep_dim_ration_contents} the amplitudes in \cref{eq:R_i_def} are easy to compute classically in time $\widetilde{O}(d)$.

Following the prescription outlined above, we construct in \cref{fig:pbt_std_pgm_circuit} an efficient circuit for the PGM $E$ dilated as $\Pi$. The circuit mainly acts on the dilated space spanned by $(T^{0},T^{1},T^{2},\dotsc,T^{n},T^{n+1}) \in \Paths(\widetilde{\Brat})$, which form the Gelfand--Tsetlin basis. The ancilla registers are used implicitly in the circuit.

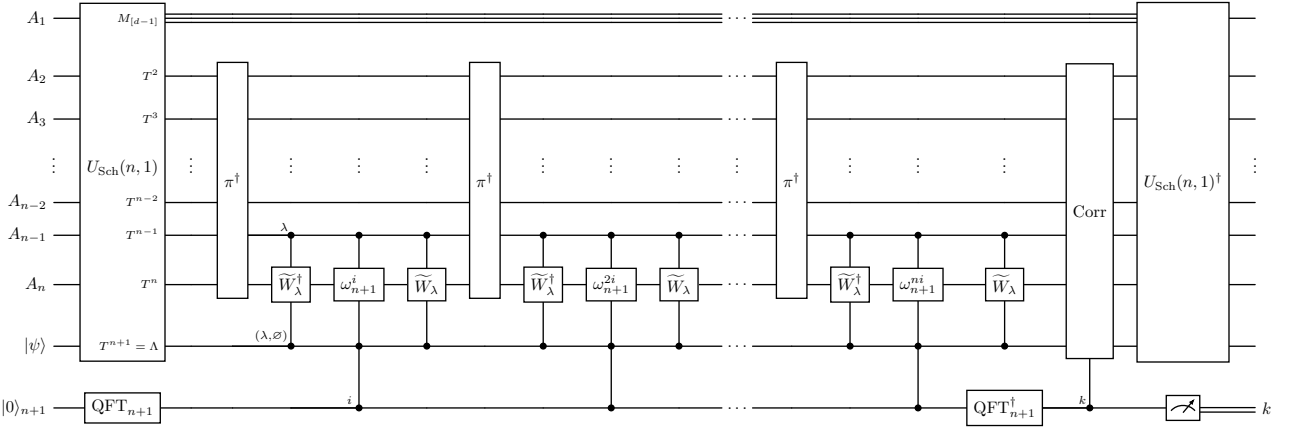
\begin{figure}[H]
\centering
\newcommand{\Sgi}{\gate[wires=8]{U_{\mathrm{Sch}}(n,1)^\dagger}}
\newcommand{\Sg}{\gate[wires=8,label style={yshift=0.3cm}]{U_{\mathrm{Sch}}(n,1)}}
\newcommand{\go}[1]{\gateoutput{#1}}
\newcommand{\iA}[1]{\lstick{$A_{#1}$}}
\newcommand{\ipsi}{\lstick{$\ket{\psi}$}}
\newcommand{\keto}{\lstick{$\ket{0}_{n+1}$}}
\newcommand{\namel}[1]{\wire[l][1]["(\lambda {,} \0)"{above,pos=#1}]{a} }
\newcommand{\name}[1]{\wire[l][1]["#1"{above,pos=0.1}]{a} }
\newcommand{\Cyc}[1]{\gate[wires=#1]{\pi^\dagger}}
\newcommand{\gW}[1]{\gate[wires=#1]{\widetilde{W}_{\lambda}}}
\newcommand{\gWinv}[1]{\gate[wires=#1]{\widetilde{W}_{\lambda}^\dagger}}
\newcommand{\gateCorr}{\gate[wires=7]{\mathrm{Corr}}}
\newcommand{\gateR}[1]{\gate{\omega_{n+1}^{#1}}}

\resizebox{\textwidth}{!}{
\begin{quantikz}[classical gap = 2.5pt, row sep = 0.6cm]
\iA{1}  &\Sg \go{$\gtM{d-1}$}  &\qb& & & & & & & & & \HD & & & & & & \Sgi & \qq \\
\iA{2}  &\go{$T^{2}$}          &\qq& \Cyc{6}  & & & & \Cyc{6}  & & & & \HD & \Cyc{6} & & & & \gateCorr & & \\
\iA{3}  &\go{$T^{3}$}          &\qq&           & & & &           & & & & \HD &          & & & & & & \\
\qn\vd  &                      &\vd&           & \vd & \vd & \vd & & \vd & \vd & \vd & \vd & & \vd & \vd & \vd & & & \vd \\
\iA{n-2}&\go{$T^{n-2}$}        &\qq&           & & & & & & & & \HD & & & & & & & \\
\iA{n-1}&\go{$T^{n-1}$}        &\qq&           & \ctrl{ 1} \name{\lambda} & \ctrl{1}  & \ctrl{1} & & \ctrl{1} & \ctrl{1} & \ctrl{1} & \HD & & \ctrl{1} & \ctrl{1} & \ctrl{1} & & & \\
\iA{n}  &\go{$T^{n}$}          &\qq&           & \gWinv{1} & \gateR{i} & \gW{1}  & & \gWinv{1} & \gateR{2i} & \gW{1} & \HD & & \gWinv{1} & \gateR{ni} & \gW{1} & & & \\
\ipsi   &\go{$T^{n+1}=\Lambda$}&\qq&           & \ctrl{-1} \namel{0.3} & \ctrl{-1} & \ctrl{-1} & & \ctrl{-1} & \ctrl{-1} & \ctrl{-1} & \HD & & \ctrl{-1} & \ctrl{-1} & \ctrl{-1} & & & \\
\keto   &\gateQFT{n+1}              &   & && \ctrl{-1} \name{i} & & & & \ctrl{-1} & & \HD & & & \ctrl{-1} & \gateQFTinv{n+1} & \ctrl{-1} \name{k} & \meter{} \wire[r][1][classical gap = 0.2pt]{c} & \qn \rstick{$k$}
\end{quantikz}
}
\caption{The circuit implementation of the PGM $E$ from \cref{def:PGM_PBT,def:PGM_PBT_schur} in standard encoding. The registers $T^{2},T^{3},\dotsc,T^{n}$ are dilated as per \cref{sec:pgm:dilation}. The correction gate $\mathrm{Corr}$ together with $U_{\mathrm{Sch}}(n,1)^\dagger$ transform is optional: it is used to bring the post-measurement state to the form defined by the original PGM measurement $E$.
}
\label{fig:pbt_std_pgm_circuit}
\end{figure}

First, a mixed quantum Schur transform maps the computational basis to the mixed Schur basis which is usually labelled by $\ket{\gtM{d-1},T}$, where $M$ is a Gelfand--Tsetlin pattern and $T$ is a path in $\Brat$. We assume a tensor product structure for different vertices $T^{i}$ of the path $T \in \Paths(\Brat)$ and we use the standard encoding for $\ket{T}$. Moreover, all registers $T^{2},\dotsc,T^{n}$ are assumed to be dilated, according to the procedure explained in \cref{sec:pgm:dilation}. Since $T^{0}$ and $T^{1}$ can only have one possible value, we omit those registers from the diagram since they are one-dimensional. The last level $T^{n+1}$ of the path $T$ is labelled by $\lambda$ and indicates an irreducible representation. The cyclic permutation gate $\pi = (1 2 \dotsc n) = \sigma_{1}\sigma_{2} \cdots \sigma_{n-1}$ acts only on $n-1$ wires of the dilated Gelfand--Tsetlin basis, and each of the transpositions $\sigma_i$ acts only locally on the registers $T^{i-1}, T^{i}, T^{i+1}$ ($\sigma_1$ acts only on $T^{2}$, and $\sigma_2$ acts only on $T^{2}, T^{3}$). $\widetilde{W}$ prepares the state $\ket{\widetilde{w}_{\lambda}}$ conditioned on $\lambda$, i.e., $\widetilde{W}_{\lambda} \ket{0} = \ket{\widetilde{w}_{\lambda}}$ where $\widetilde{W}_{\lambda}$ is controlled on $\lambda$. The phase gates $\omega^{ki}_{n+1}$ act non-trivially only on $\ket{0}$ in the register $T^{n}$, and are controlled on the condition $T^{n-1} = T^{n+1} = \lambda$. Finally, the measurement outcome $k = 0$ corresponds to the failure of the protocol, otherwise $k \in [n]$ indicates the port where Bob can find the teleported state $\ket{\psi}$ of dimension $d$.

We now argue that the gate complexity of our circuit in \cref{fig:pbt_std_pgm_circuit} is $\widetilde{O}(n d^4)$:

\begin{enumerate}[leftmargin=*]
    \item The complexity of implementing the mixed quantum Schur transform $U_{\text{Sch}}(n,1)$ is $\widetilde{O}(nd^4)$ \cite{nguyen2023mixed,grinko2023gelfandtsetlin}.
    \item The complexity of implementing $\pi = \sigma_{1}\sigma_{2} \dotsc \sigma_{n-1}$ based on \cref{fig:cyclic_circuit_std} is $\widetilde{O}(n d^2)$.
    The factor $n$ comes from the number of transpositions $\sigma_i$ in $\pi$.
    Each transposition $\sigma_i$ is 3-local: it acts only on registers $T^{i-1}, T^{i}, T^{i+1}$.
    More specifically, $\sigma_i$ is a $2 \times 2$ rotation on $T^{i}$ controlled by $T^{i-1}$ and $T^{i+1}$.
    According to \cref{fig:transposition_circuit_std}, each $\sigma_i$ can be implemented with $\widetilde{O}(d^2)$ gates, each of which decompose into $\widetilde{O}(1)$ elementary gates.
    In particular, each of the $\mathrm{R}_{j,k}$ gates
    \begin{align}
        \mathrm{R}_{j,k} \defeq
        \mx{
            \frac{1}{r_{j,k}} &
            \sqrt{1 - \frac{1}{r_{j,k}^2}} \\
            \sqrt{1 - \frac{1}{r_{j,k}^2}} &
          - \frac{1}{r_{j,k}}
        },
        \qquad
        r_{j,k} \defeq \lambda_j - \lambda_k + k - j
        \label{eq:Rjk}
    \end{align}
    appearing in \cref{fig:transposition_circuit_std} can be implemented with $\widetilde{O}(1)$ elementary gates and $\widetilde{O}(1)$ auxiliary qubits for computation of rotation parameters $r_{j,k}$.
        
        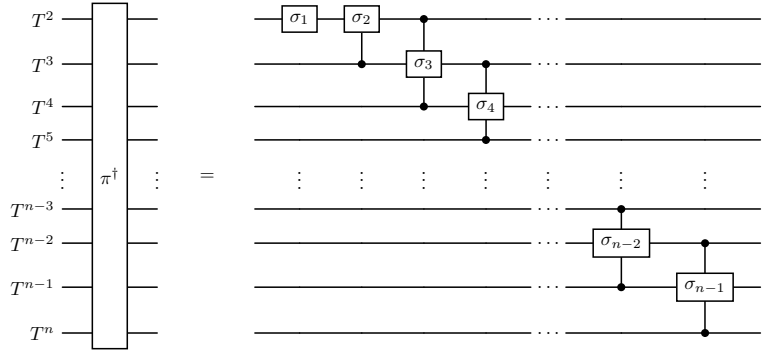
\begin{figure}[!ht]
        \centering
        \newcommand{\mideq}{\midstick[9,brackets=none]{\text{$=$}}}
        \newcommand{\Cyc}[1]{\gate[wires=#1]{\pi^\dagger}}
        \newcommand{\T}[1]{\lstick{$T^{#1}$}}
        \newcommand{\sgm}[1]{\gate{\sigma_{#1}}}
        
        \resizebox{0.6\textwidth}{!}{
        \begin{quantikz}[classical gap = 2pt, row sep = 0.3cm]
        \T{2}  &\Cyc{9}&   &\mideq\qn&&\sgm{1}\qq& \sgm{2} & \ctrl{1}&         &\HD&         &         &\\
        \T{3}  &       &   &      \qn&&\qq       &\ctrl{-1}& \sgm{3} &\ctrl{1} &\HD&         &         &\\
        \T{4}  &       &   &      \qn&&\qq       &         &\ctrl{-1}&\sgm{4}  &\HD&         &         &\\
        \T{5}  &       &   &      \qn&&\qq       &         &         &\ctrl{-1}&\HD&         &         &\\
        \qn\vd &       &\vd&      \qn&&     \vd  &  \vd    &   \vd   &    \vd  &\vd&    \vd  &  \vd    &\\
        \T{n-3}&       &   &      \qn&&\qq       &         &         &         &\HD&\ctrl{1} &         &\\
        \T{n-2}&       &   &      \qn&&\qq       &         &         &         &\HD&\sgm{n-2}&\ctrl{1} &\\
        \T{n-1}&       &   &      \qn&&\qq       &         &         &         &\HD&\ctrl{-1}&\sgm{n-1}&\\
        \T{n}  &       &   &      \qn&&\qq       &         &         &         &\HD&         &\ctrl{-1}&
        \end{quantikz}
        }
        \caption{Circuit for the cyclic permutation $\pi^\dagger = \sigma_{n-1} \sigma_{n-2} \cdots \sigma_2 \sigma_1$ in the Gelfand--Tsetlin basis. Each transposition $\sigma_i$ acts locally on registers $T^{i-1}, T^{i}, T^{i+1}$, with $T^{i-1}$ and $T^{i+1}$ used as controls. Note, that $\sigma_1$ acts only on $T^{2}$, and $\sigma_2$ acts only on $T^{2}, T^{3}$ because we have dropped the registers $T^0, T^1$ (they are always one-dimensional).}
        \label{fig:cyclic_circuit_std}
        \end{figure}
        
        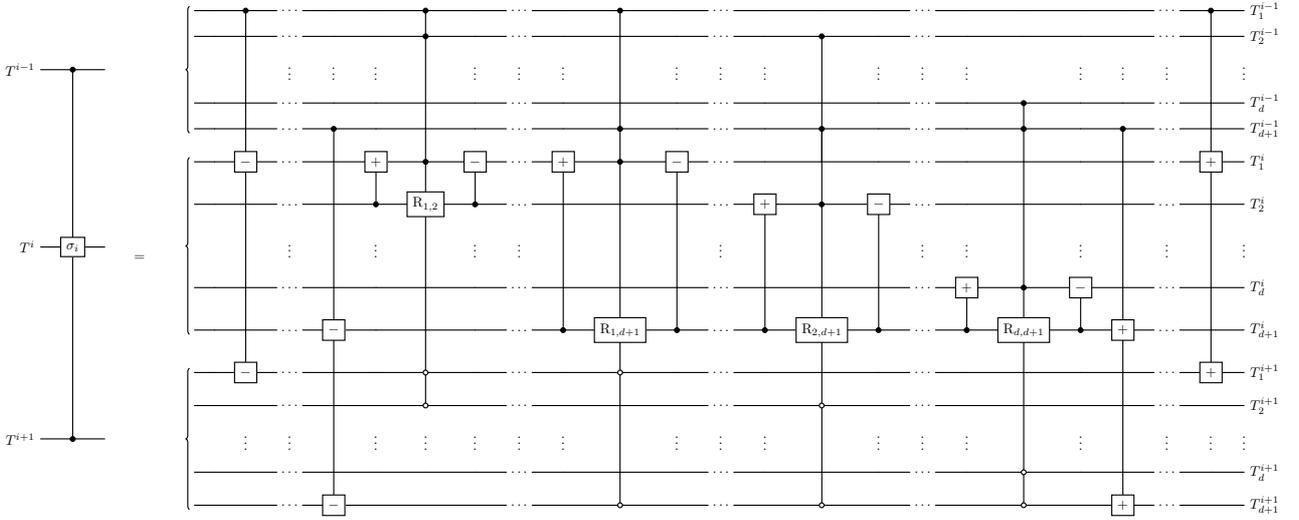
\begin{figure}[!ht]
        \centering
        \newcommand{\y}[1]{\lstick{$y_{#1}$}}
        \newcommand{\Tu}{\lstick{$T^{i-1}$}}
        \newcommand{\Tm}{\lstick{$T^{i}$}}
        \newcommand{\Td}{\lstick{$T^{i+1}$}}
        
        \newcommand{\Tust}{\lstick[5]{$T^{i-1}$}}
        \newcommand{\Tmst}{\lstick[5]{$T^{i}$}}
        \newcommand{\Tdst}{\lstick[5]{$T^{i+1}$}}
        
        \renewcommand{\Tust}{\lstick[5]{}}
        \renewcommand{\Tmst}{\lstick[5]{}}
        \renewcommand{\Tdst}{\lstick[5]{}}
        
        \newcommand{\Tui}[1]{\rstick{$T^{i-1}_{#1}$}}
        \newcommand{\Tmi}[1]{\rstick{$T^{i}_{#1}$}}
        \newcommand{\Tdi}[1]{\rstick{$T^{i+1}_{#1}$}}
        
        \newcommand{\grm}[1]{\gate{\mathrm{#1}}}
        
        \newcommand{\sgma}[1]{\gate[wires=1]{\sigma_{#1}}}
        \newcommand{\sgm}[1]{\gate[wires=1]{\mathrm{R}_{#1}}}
        
        \newcommand{\clo}[1]{\ctrl[open]{#1}}  
        
        \newcommand{\mideq}{\midstick[15,brackets=none]{\text{$=$}}}
        
        \resizebox{\textwidth}{!}{
        \begin{quantikz}[classical gap = 2.5pt, row sep = 0.5cm]
        \qn&         &&\mideq\qn&&\Tust&\qq&\ctrl{10}&\HD&         &         &\ctrl{1} &         &\HD&         & \ctrl{4}  &         &\HD&         &           &         &\HD&         &           &         &         &\HD&\ctrl{10}&\Tui{1}  \\
        \qn&         &&      \qn&&     &\qq&         &\HD&         &         &\ctrl{4} &         &\HD&         &           &         &\HD&         & \ctrl{4}  &         &\HD&         &           &         &         &\HD&         &\Tui{2}  \\
        \Tu&\ctrl{5} &&      \qn&&     &\qn&         &\vd&  \vd    &  \vd    &         &  \vd    &\vd&  \vd    &           &  \vd    &\vd& \vd     &           &  \vd    &\vd&   \vd   &           &  \vd    &  \vd    &\vd&         &\vd      \\
        \qn&         &&      \qn&&     &\qq&         &\HD&         &         &         &         &\HD&         &           &         &\HD&         &           &         &\HD&         & \ctrl{1}  &         &         &\HD&         &\Tui{d}  \\
        \qn&         &&      \qn&&     &\qq&         &\HD&\ctrl{10}&         &         &         &\HD&         & \ctrl{1}  &         &\HD&         & \ctrl{2}  &         &\HD&         & \ctrl{4}  &         &\ctrl{10}&\HD&         &\Tui{d+1}\\
        \qn&         &&      \qn&&\Tmst&\qq&\grm{-}  &\HD&         & \grm{+} &\ctrl{1} & \grm{-} &\HD& \grm{+} & \ctrl{4}  & \grm{-} &\HD&         &           &         &\HD&         &           &         &         &\HD& \grm{+} &\Tmi{1}  \\
        \qn&         &&      \qn&&     &\qq&         &\HD&         &\ctrl{-1}&\sgm{1,2}&\ctrl{-1}&\HD&         &           &         &\HD& \grm{+} & \ctrl{3}  & \grm{-} &\HD&         &           &         &         &\HD&         &\Tmi{2}  \\
        \Tm&\sgma{i} &&      \qn&&     &\qn&         &\vd&         &  \vd    &         &  \vd    &\vd&         &           &         &\vd&         &           &         &\vd&   \vd   &           &   \vd   &         &\vd&         &\vd      \\
        \qn&         &&      \qn&&     &\qq&         &\HD&         &         &         &         &\HD&         &           &         &\HD&         &           &         &\HD& \grm{+} & \ctrl{1}  & \grm{-} &         &\HD&         &\Tmi{d}  \\
        \qn&         &&      \qn&&     &\qq&         &\HD& \grm{-} &         &         &         &\HD&\ctrl{-4}&\sgm{1,d+1}&\ctrl{-4}&\HD&\ctrl{-3}&\sgm{2,d+1}&\ctrl{-3}&\HD&\ctrl{-1}&\sgm{d,d+1}&\ctrl{-1}& \grm{+} &\HD&         &\Tmi{d+1}\\
        \qn&         &&      \qn&&\Tdst&\qq&\grm{-}  &\HD&         &         &\clo{-4} &         &\HD&         &\clo{-1}   &         &\HD&         &           &         &\HD&         &           &         &         &\HD&\grm{+}  &\Tdi{1}  \\
        \qn&         &&      \qn&&     &\qq&         &\HD&         &         &\clo{-1} &         &\HD&         &           &         &\HD&         &\clo{-2}   &         &\HD&         &           &         &         &\HD&         &\Tdi{2}  \\
        \Td&\ctrl{-5}&&      \qn&&     &\qn&  \vd    &\vd&         &  \vd    &  \vd    &  \vd    &\vd&  \vd    &           &  \vd    &\vd&  \vd    &           &  \vd    &\vd&  \vd    &           &  \vd    &         &\vd&  \vd    &\vd      \\
        \qn&         &&      \qn&&     &\qq&         &\HD&         &         &         &         &\HD&         &           &         &\HD&         &           &         &\HD&         &\clo{-4}   &         &         &\HD&         &\Tdi{d}  \\
        \qn&         &&      \qn&&     &\qq&         &\HD&\grm{-}  &         &         &         &\HD&         &\clo{-4}   &         &\HD&         &\clo{-3}   &         &\HD&         &\clo{-1}   &         &\grm{+}  &\HD&         &\Tdi{d+1}
        \end{quantikz}   
        }
        \caption{
        Quantum circuit for implementing the transposition $\sigma_i$ in the Gelfand--Tsetlin basis. The controlled ``$\pm$'' gates perform arithmetic operations that add / subtract the control register from the target registers.
        For each $j,k$ such that $1 \leq j < k \leq d$, the corresponding $\mathrm{R}_{j,k}$ gate defined in \cref{eq:Rjk} acts on the $k$-th register of $T^i$ and has five controls:
        the $j$-th and $k$-th registers of $T^{i-1}$ and $T^{i+1}$, and the $j$-th register of $T^i$.
        At this level of abstraction, the implementation of $\sigma_i$ contains $\widetilde{O}(d^2)$ gates.
        }
        \label{fig:transposition_circuit_std}
        \end{figure}

        \begin{figure}[!ht]
        \centering
        
        \newcommand{\grm}[1]{\gate{\mathrm{#1}}}
        
        \newcommand{\Tust}{\lstick[5]{$T^{n-1}$}}
        \newcommand{\Tmst}{\lstick[5]{$T^{n}$}}
        \newcommand{\Tdst}{\lstick[5]{$T^{n+1}$}}
        
        \renewcommand{\Tust}{\lstick[5]{}}
        \renewcommand{\Tmst}{\lstick[5]{}}
        \renewcommand{\Tdst}{\lstick[5]{}}
        
        \newcommand{\Tu}{\lstick{$T^{n-1}$}}
        \newcommand{\Tm}{\lstick{$T^{n}$}}
        \newcommand{\Td}{\lstick{$T^{n+1}$}}

        \newcommand{\Tui}[1]{\rstick{$T^{n-1}_{#1}$}}
        \newcommand{\Tmi}[1]{\rstick{$T^{n}_{#1}$}}
        \newcommand{\Tdi}[1]{\rstick{$T^{n+1}_{#1}$}}
        
        \newcommand{\gW}{\gate{\widetilde{W}_{\lambda}}}
        \newcommand{\namelbd}[1]{\wire[l][1]["(\lambda {,} \0)"{above,pos=#1}]{a}}  
        \newcommand{\namel}[1]{\wire[l][1]["\lambda_{#1}"{above,pos=0.1}]{a} }      
        \newcommand{\name}[1]{\wire[l][1]["#1"{above,pos=0.1}]{a} }                 
        
        \newcommand{\clnm}[2]{\ctrl{#1}\namel{#2}}                                  
        \newcommand{\clo}[1]{\ctrl[open]{#1}}  
        \newcommand{\cl}[1]{\ctrl{#1}} 
        \newcommand{\gR}[1]{\gate{\mathrm{R}_{#1}}}

        \newcommand{\mideq}{\midstick[15,brackets=none]{\text{$=$}}}
        
        \resizebox{\linewidth}{!}{
        \begin{quantikz}[classical gap = 2.5pt, row sep = 0.5cm]
        \qn&         &&\mideq\qn&&\Tust&\qq&\clnm{10}{1}&            &\HD&            &              &\cl{1}  &\cl{1}  &\HD&\cl{1}  &\cl{1}  &\cl{10}&       &\HD&       &       &\Tui{1}  \\
        \qn&         &&      \qn&&     &\qq&            &\clnm{10}{2}&\HD&            &              &\cl{2}  &\cl{2}  &\HD&\cl{2}  &\cl{2}  &       &\cl{10}&\HD&       &       &\Tui{2}  \\
        \Tu&\ctrl{5} &&      \qn&&     &\qn&            &            &\vd&  \vd       &  \vd         &        &        &\vd&        &        &       &       &\vd&  \vd  &  \vd  &\vd      \\
        \qn&         &&      \qn&&     &\qq&            &            &\HD&\clnm{10}{d}&              &\cl{1}  &\cl{1}  &\HD&\cl{1}  &\cl{1}  &       &       &\HD&\cl{10}&       &\Tui{d}  \\
        \qn&         &&      \qn&&     &\qq&            &            &\HD&            &\clnm{10}{d+1}&\cl{1}  &\cl{1}  &\HD&\cl{1}  &\cl{1}  &       &       &\HD&       &\cl{10}&\Tui{d+1}\\
        \qn&         &&      \qn&&\Tmst&\qq&            &            &\HD&            &              &\gR{1}  &\clo{1} &\HD&\clo{1} &\clo{1} &\grm{+}&       &\HD&       &       &\Tmi{1}  \\
        \qn&         &&      \qn&&     &\qq&            &            &\HD&            &              &        &\gR{2}  &\HD&\clo{2} &\clo{2} &       &\grm{+}&\HD&       &       &\Tmi{2}  \\
        \Tm&\gW      &&      \qn&&     &\qn&            &            &\vd&            &              &        &        &\vd&        &        &       &       &\vd&       &       &\vd      \\
        \qn&         &&      \qn&&     &\qq&            &            &\HD&            &              &        &        &\HD&\gR{d}  &\clo{1} &       &       &\HD&\grm{+}&       &\Tmi{d}  \\
        \qn&         &&      \qn&&     &\qq&            &            &\HD&            &              &        &        &\HD&        &\gR{d+1}&       &       &\HD&       &\grm{+}&\Tmi{d+1}\\
        \qn&         &&      \qn&&\Tdst&\qq&\grm{-}     &            &\HD&            &              &\clo{-5}&\clo{-4}&\HD&\clo{-2}&\clo{-1}&\grm{+}&       &\HD&       &       &\Tdi{1}  \\
        \qn&         &&      \qn&&     &\qq&            &\grm{-}     &\HD&            &              &\clo{-1}&\clo{-1}&\HD&\clo{-1}&\clo{-1}&       &\grm{+}&\HD&       &       &\Tdi{2}  \\
        \Td&\ctrl{-5}&&      \qn&&     &\qn&  \vd       &  \vd       &\vd&            &              &        &        &\vd&        &        &  \vd  &  \vd  &\vd&       &       &\vd      \\
        \qn&         &&      \qn&&     &\qq&            &            &\HD&\grm{-}     &              &\clo{-2}&\clo{-2}&\HD&\clo{-2}&\clo{-2}&       &       &\HD&\grm{+}&       &\Tdi{d}  \\
        \qn&         &&      \qn&&     &\qq&            &            &\HD&            &\grm{-}       &\clo{-1}&\clo{-1}&\HD&\clo{-1}&\clo{-1}&       &       &\HD&       &\grm{+}&\Tdi{d+1}
        \end{quantikz}   
        }
        \caption{Quantum circuit for implementing the unitary $\widetilde{W}$ defined by (\ref{def:W_lambda}) in the Gelfand--Tsetlin basis. The controlled ``$\pm$'' gates perform arithmetic operations that add / subtract the control register from the target registers, just as on \cref{fig:transposition_circuit_std}. \cref{eq:Rjk} defines rotation gates $\mathrm{R}_{j}$ that are controlled on $2(d+1+j)$ registers in total: all registers of $T^{i-1}$ and $T^{i+1}$ and first $(j-1)$ registers of $T^i$. The overall gate complexity of $\widetilde{W}$ is $\widetilde{O}(d^2)$.}
        \label{fig:W_gate}
        \end{figure}
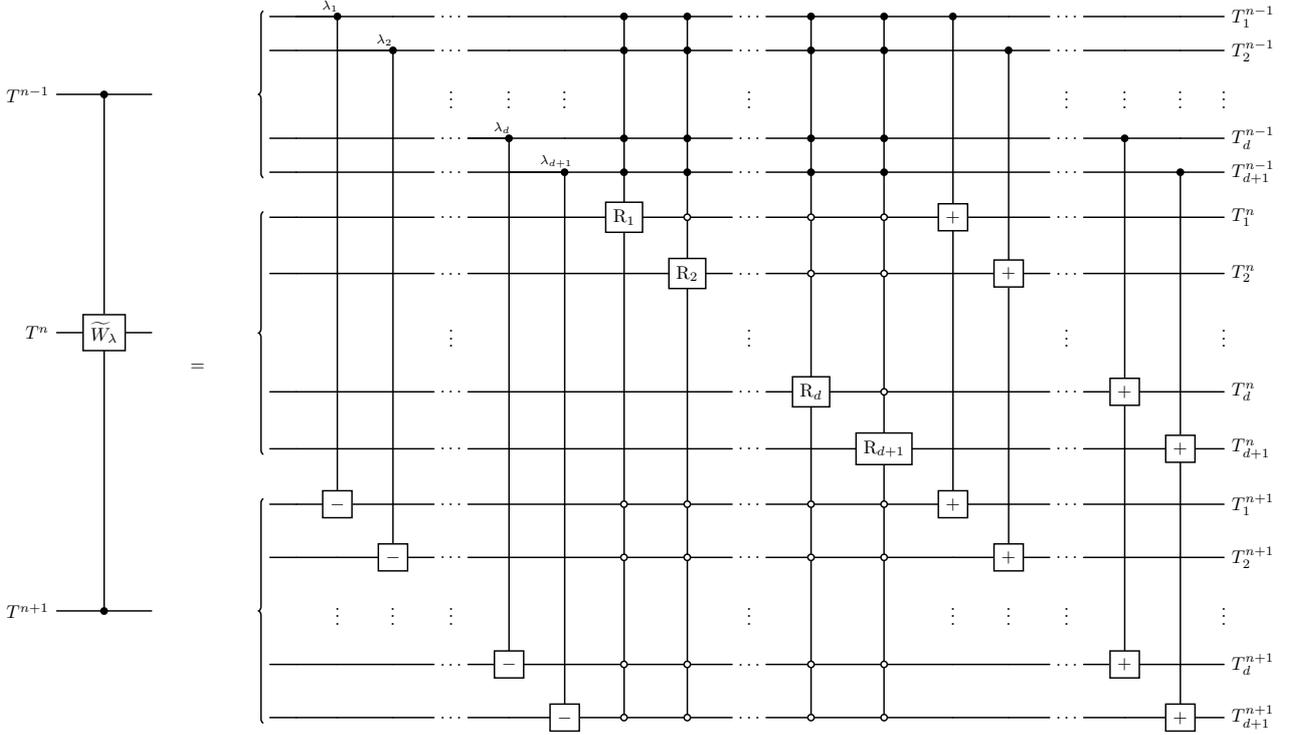
    \item The operation $\widetilde{W}$ defined in \cref{def:W_lambda} is a $\lambda$-controlled unitary $\widetilde{W}_\lambda$ that acts non-trivially only on a $(d+1)$-dimensional subspace of the register $T^{n}$. This register consists of $d+1$ wires corresponding to $d+1$ rows of a Young diagram, see \cref{fig:W_gate}. The matrix entries of $\widetilde{W}_\lambda$ (see \cref{def:w_lambda_tilde}) can be reversibly computed on the fly using gates $\mathrm{R}_i, \, i \in [d+1]$ from \cref{eq:R_i_def} in classical time $\widetilde{O}(d)$, which must be implemented in a quantum circuit coherently. Therefore, implementing $\widetilde{W}$ would have the gate complexity $\widetilde{O}(d^2)$.
    \item $\omega^{ki}_{n+1}$ denotes the gate $\omega^{ki}_{n+1} \proj{0} + (I-\proj{0})$ on register $T^{n}$ conditioned on the registers $T^{n-1}=\lambda,\, T^{n+1}=\Lambda=(\lambda,\0)$. This has complexity $\widetilde{O}(1)$.
    \item The complexity of the Quantum Fourier Transform $\mathrm{QFT}_{n+1}$ is $\widetilde{O}(1)$.
    \item One can optionally implement the correction gate $\mathrm{Corr}$ together with inverse mixed Schur transform at the end to get the right post-measurement state (according to the definition of PGM $E$ from \cref{def:PGM_PBT}). For that one needs to uncompute the gates $\pi^k$ for $k \in [n]$ and $\widetilde{W}_\lambda$, and instead run the gate $W_\lambda$ from \cref{def:W_lambda} followed by $\pi^k$. The complexity of implementing the correction gate does not change both the total gate and time complexities of the full circuit, adding only a constant factor overhead.
        \begin{figure}[!ht]
        \centering
        \newcommand{\mideq}{\midstick[8,brackets=none]{\text{$=$}}}
        \newcommand{\Corr}{\gate[wires=7]{\mathrm{Corr}}}
        \newcommand{\T}[1]{\lstick{$T^{#1}$}}
        \newcommand{\name}[1]{\wire[l][1]["#1"{above,pos=0.1}]{a} }  
        \newcommand{\Cyc}[1]{\gate[wires=6]{\pi^{#1}}}
        \newcommand{\gWt}{\gate{\widetilde{W}_{\lambda}^\dagger}}
        \newcommand{\gW}{\gate{W_{\lambda}}}
        
        \resizebox{0.5\textwidth}{!}{
        \begin{quantikz}[classical gap = 2pt]
        \T{2}  &  \Corr          &   &\mideq\qn&&\qq&    \Cyc{-k}     &         &         &    \Cyc{k}      &   &\\
        \T{3}  &                 &   &      \qn&&\qq&                 &         &         &                 &   &\\
        \qn\vd &                 &\vd&      \qn&&\vd&                 &  \vd    &  \vd    &                 &\vd&\\
        \T{n-2}&                 &   &      \qn&&\qq&                 &         &         &                 &   &\\
        \T{n-1}&                 &   &      \qn&&\qq&                 &\ctrl{1} &\ctrl{1} &                 &   &\\
        \T{n}  &                 &   &      \qn&&\qq&                 &\gWt     &\gW      &                 &   &\\
        \L     &                 &   &      \qn&&\qq&                 &\ctrl{-1}&\ctrl{-1}&                 &   &\\
               &\ctrl{-1}\name{k}&   &      \qn&&\qq&\ctrl{-2}\name{k}&         &         &\ctrl{-2}\name{k}&   &
        \end{quantikz}
        }
        \caption{Circuit for the correction operation $\mathrm{Corr}$. $W_\lambda$ gate is defined in \cref{def:W_lambda} and its implementation is completely analogous to $\widetilde{W}_\lambda$ from \cref{fig:W_gate}.}
        \label{fig:correction_gate}
        \end{figure}
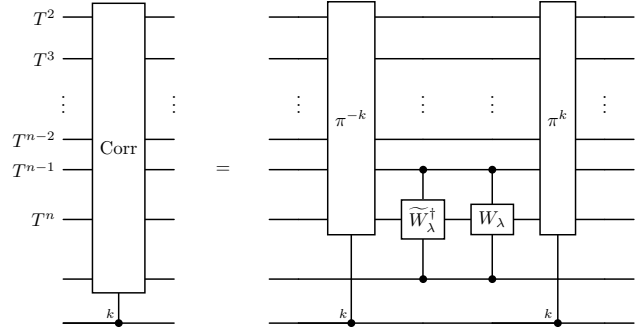
    \item Now we are ready to count the total gate and time complexity. For that, notice that each gate $\pi^\dagger$ in \cref{fig:pbt_std_pgm_circuit} consisting of local gates $\sigma_i$ as in \cref{fig:cyclic_circuit_std} can be pushed to the left of the circuit. This will reduce the naive complexity $\widetilde{O}(n^2d^2)$ of implementing $n$ sequential gates $\pi^\dagger$ to just $\widetilde{O}(n d^2)$. Counting everything together gives $\widetilde{O}(n d^4)$ total time complexity. The gate complexity is $\widetilde{O}(n^2 d^4)$.
\end{enumerate}

\noindent Similarly, we can count the number of auxiliary qubits needed to implement our circuit \cref{fig:pbt_std_pgm_circuit} in the standard encoding. However, we must be careful with the precision of implementing each gate. If the total target precision is $\epsilon$ then if the circuit has $\poly(n,d)$ gates it implies the precision per each gate must be $\epsilon_g \defeq \epsilon / \poly(n,d)$. This will translate into the size of auxiliary registers for classical computation of matrix entries of the unitaries which need to be implemented coherently: the number of qubits needed for such computations will scale as $\log(1 / \epsilon_g) = O(\log(1 / \epsilon) + \log(n) + \log(d))$. With that in mind, we can summarize the total space complexity:

\begin{enumerate}[leftmargin=*]
    \item The number of auxiliary qubits needed to implement the mixed Schur transform isometry in the standard encoding \cite{grinko2023gelfandtsetlin} and create a Naimark's dilation from \cref{sec:pgm:dilation} after the mixed Schur transfer is $(n+d)d\log(n)\polylog(d,1/\epsilon)$.
    \item The number of auxiliary qubits needed to implement each gate $\sigma_i$ from \cref{fig:transposition_circuit_std} inside $\pi^\dagger$ gate from \cref{fig:cyclic_circuit_std} is $\log(n)\polylog(d,1/\epsilon)$. However, when $O(n)$ gates $\sigma_i$ are implemented in parallel with the shift trick of $\pi^\dagger$ gates we need to have $n\log(n)\polylog(d,1/\epsilon)$ auxiliary qubits available.
    \item The gates $\widetilde{W}_\lambda$ and $W_\lambda$ require $d^2\log(n)\polylog(d,1/\epsilon)$ qubits to implement gates $\mathrm{R}_i$, see \cref{fig:W_gate}.
    \item Overall, the total number of auxiliary qubits needed is $(n+d)d\log(n)\polylog(d,1/\epsilon)$.
\end{enumerate}

The above analysis finishes the proof of the first statement of \cref{thm:pbt} for standard PGM. In the next section, we describe how to extend the construction for standard PGM for generic PBT measurements, including the measurement needed to implement pPBT with the EPR resource state.

\subsection{Quantum circuit for deterministic PBT measurement}
\label{sec:dPBT_circuits}

Using the quantum circuit for the standard PGM $E$ \cref{fig:pbt_std_pgm_circuit} we can easily implement the POVM $E^{\hollowstar}$ for the dPBT protocols from \cref{def:dpbt_povm,def:dpbt_povm_schur}. For that, we need to use \Cref{lem:dpbt_dilation}, which tells us to use one additional auxiliary qudit of dimension $n$, prepared in the state $\ket{+}_n = \sum_{i=0}^{n-1} \frac{1}{\sqrt{n}} \ket{i}$. The resulting circuit is presented in \Cref{fig:dpbt_circuit}. Notice that the circuit is almost the same as the one for the standard PGM $E$ presented in \cref{fig:pbt_std_pgm_circuit}. The only difference is an additional register of dimension $n$ and a controlled Pauli gate $Z \defeq \sum_{i=0}^{n-1} \omega_{n}^{i} \proj{i}$.

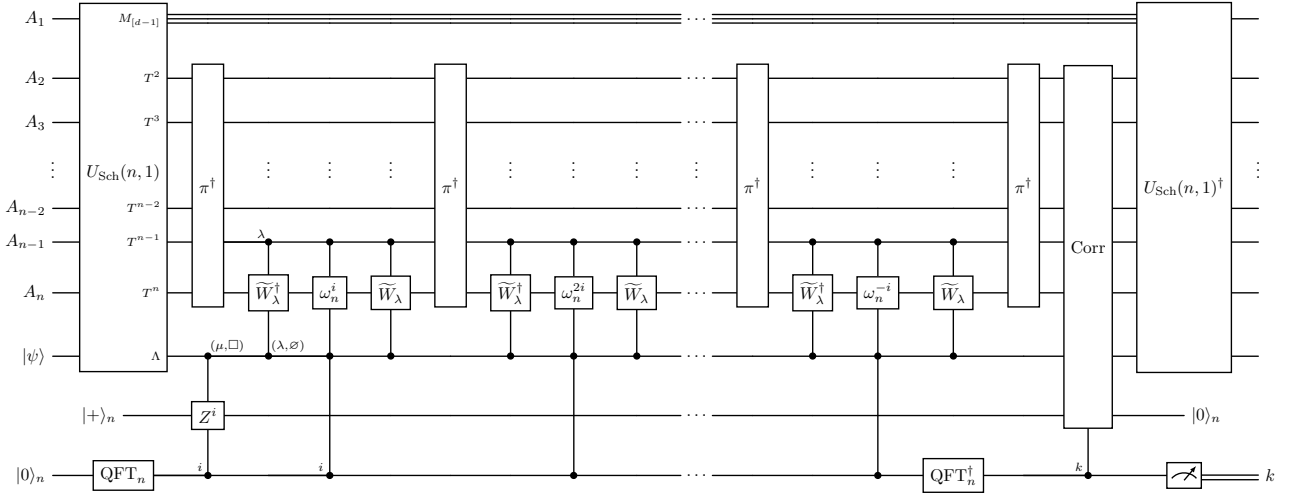
\begin{figure}[H]
\centering
\newcommand{\Sgi}{\gate[wires=8]{U_{\mathrm{Sch}}(n,1)^\dagger}}
\newcommand{\Sg}{\gate[wires=8,label style={yshift=0.3cm}]{U_{\mathrm{Sch}}(n,1)}}
\newcommand{\go}[1]{\gateoutput{#1}}
\newcommand{\iA}[1]{\lstick{$A_{#1}$}}
\newcommand{\ipsi}{\lstick{$\ket{\psi}$}}
\newcommand{\keto}{\lstick{$\ket{0}_n$}}
\newcommand{\Zgate}[1]{\gate{Z^{#1}}}
\newcommand{\namel}[1]{\wire[l][1]["(\lambda {,} \0)"{above,pos=#1}]{a} }
\newcommand{\namer}[1]{\wire[r][1]["(\lambda {,} \0)"{above,pos=#1}]{a} }
\newcommand{\namerm}[1]{\wire[r][1]["(\mu {,} \square)"{above,pos=#1}]{a} }
\newcommand{\name}[1]{\wire[l][1]["#1"{above,pos=0.1}]{a} }
\newcommand{\Cyc}[1]{\gate[wires=#1]{\pi^\dagger}}
\newcommand{\gW}[1]{\gate[wires=#1]{\widetilde{W}_{\lambda}}}
\newcommand{\gWinv}[1]{\gate[wires=#1]{\widetilde{W}_{\lambda}^\dagger}}
\newcommand{\gateCorr}{\gate[wires=8]{\mathrm{Corr}}}
\newcommand{\clo}[1]{\ctrl[open]{#1}}  
\newcommand{\gateR}[1]{\gate{\omega_{n}^{#1}}}

\resizebox{\textwidth}{!}{
\begin{quantikz}[classical gap = 2.5pt, row sep = 0.6cm]
\iA{1}  &\Sg \go{$\gtM{d-1}$} &\qb                   & & & & & & & & \HD & & & &                                                                                                     & & & \Sgi & \qq \\
\iA{2}  &\go{$T^{2}$}         &                     \Cyc{6}  & & & & \Cyc{6}  & & & & \HD & \Cyc{6} & & &                                                                & \Cyc{6} &  \gateCorr & & \\
\iA{3}  &\go{$T^{3}$}         &                                & & & &           & & & & \HD &          & & &                                                                        & & & & \\
\qn\vd  &                     &                            & \vd & \vd & \vd & & \vd & \vd & \vd & \vd & & \vd & \vd & \vd                            & & & & \vd \\
\iA{n-2}&\go{$T^{n-2}$}       &                                & & & & & & & & \HD & & & &                                                                                            & & & & \\
\iA{n-1}&\go{$T^{n-1}$}       &                                & \ctrl{1} \name{\lambda} & \ctrl{1}  & \ctrl{1} & & \ctrl{1} & \ctrl{1} & \ctrl{1} & \HD & & \ctrl{1} & \ctrl{1} & \ctrl{1}        & & & & \\
\iA{n}  &\go{$T^{n}$}         &                                & \gWinv{1} & \gateR{i} & \gW{1}  & & \gWinv{1} & \gateR{2i} & \gW{1} & \HD & & \gWinv{1} & \gateR{-i} & \gW{1}     & & & & \\
\ipsi   &\go{$\Lambda$}       &\ctrl{1}\namerm{0.3}  & \ctrl{-1} \namer{0.3} & \ctrl{-1} & \ctrl{-1} & & \ctrl{-1} & \ctrl{-1} & \ctrl{-1} & \HD & & \ctrl{-1} & \ctrl{-1} & \ctrl{-1} & &      & & \\
\qn     &\lstick{$\ket{+}_n$} &   \Zgate{i}  \qq     &           &          &           & &           &          &           & \HD & &                      &  & &   &  & \rstick{$\ket{0}_n$}   & \qn   \\
\keto   &\gateQFT{n}          &\ctrl{-1} \name{i}    & &\ctrl{-2} \name{i} & & & & \ctrl{-2} & & \HD & & & \ctrl{-2} & \gateQFTinv{n} & & \ctrl{-1} \name{k} & \meter{} \wire[r][1][classical gap = 0.2pt]{c} & \qn \rstick{$k$}
\end{quantikz}
}
\caption{The circuit implementation of the POVM $E^{\hollowstar}$ from \cref{def:dpbt_povm,def:dpbt_povm_schur} for dPBT in standard encoding. }
\label{fig:dpbt_circuit}
\end{figure}

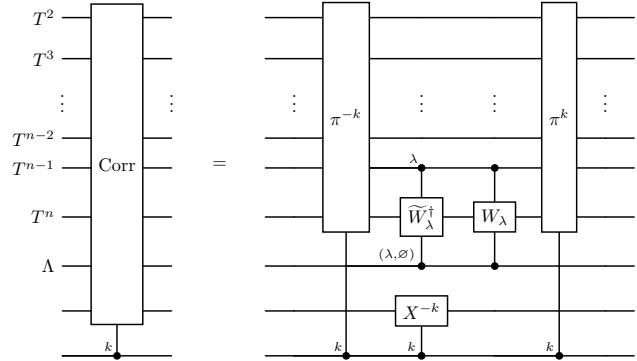
\begin{figure}[H]
    \centering
    \newcommand{\mideq}{\midstick[8,brackets=none]{\text{$=$}}}
    \newcommand{\Corr}{\gate[wires=8]{\mathrm{Corr}}}
    \newcommand{\T}[1]{\lstick{$T^{#1}$}}
    \newcommand{\Linit}{\lstick{$\Lambda$}}
    \newcommand{\name}[1]{\wire[l][1]["#1"{above,pos=0.1}]{a} }  
    \newcommand{\Cyc}[1]{\gate[wires=6]{\pi^{#1}}}
    \newcommand{\gWt}{\gate{\widetilde{W}_{\lambda}^\dagger}}
    \newcommand{\gW}{\gate{W_{\lambda}}}
    \newcommand{\Xg}[1]{\gate{X^{#1}}}
    \newcommand{\namep}[2]{\wire[l][1]["#1"{above,pos=#2}]{a} } 
    
    \resizebox{0.5\textwidth}{!}{
    \begin{quantikz}[classical gap = 2pt]
    \T{2}  &  \Corr          &   &\mideq\qn&&\qq&    \Cyc{-k}     &         &         &    \Cyc{k}      &   &\\
    \T{3}  &                 &   &      \qn&&\qq&                 &         &         &                 &   &\\
    \qn\vd &                 &\vd&      \qn&&\vd&                 &  \vd    &  \vd    &                 &\vd&\\
    \T{n-2}&                 &   &      \qn&&\qq&                 &         &         &                 &   &\\
    \T{n-1}&                 &   &      \qn&&\qq&&\ctrl{1}\namep{\lambda}{0.1}&\ctrl{1} &               &   &\\
    \T{n}  &                 &   &      \qn&&\qq&                 &\gWt     &\gW      &                 &   &\\
    \Linit &                 &   &      \qn&&\qq&&\ctrl{-1}\namep{(\lambda{,}\0)}{0.3}&\ctrl{-1}&       &   &\\
           &                 &   &      \qn&&\qq&                 & \Xg{-k} &         &                 &   &\\
           &\ctrl{-1}\name{k}&   &      \qn&&\qq&\ctrl{-3}\name{k}&\ctrl{-1}\name{k}& &\ctrl{-3}\name{k}&   &
    \end{quantikz}
    }
    \caption{Circuit for the correction gate $\mathrm{Corr}$ from \Cref{fig:dpbt_circuit}. It is a slight modification of the correction gate for standard PGM from \Cref{fig:correction_gate}.}
    \label{fig:correction_gate_dpbt}
\end{figure}

\subsection{Quantum circuit for probabilistic PBT measurement with EPR resource state}
\label{sec:pPBT_circuits}

The measurement for optimized resource state pPBT is the standard PGM $E$ from \cref{fig:pbt_std_pgm_circuit}. For EPR state case, it is straightforward now to extend the quantum circuit for the standard PGM to a quantum circuit for pPBT POVM defined via \cref{def:G_lambda_ppbt}. For that, we can employ \cref{lem:ppbt_dilation} for each irrep $(\lambda,\0) \in \Irr{\A_{n,1}^d}$, where $\widetilde{\psi}_{(\lambda,\0)}(U)$ is determined by the corresponding diagonal matrix $G$ from \cref{def:G_lambda_ppbt} extended to Naimark dilated space as $\bra{T}\widetilde{\psi}_{(\lambda,\0)}(G)\ket{T} = 0$ for all $T \in \Paths(\Lambda,\widetilde{\Brat}) \setminus \Paths(\Lambda,\Brat)$. We introduce an additional qubit on which we would act with a unitary $U_{\lambda,a}$ depending on the irrep $(\lambda,\0) \in \Irr{\A_{n,1}^d}$ and $a \in \AC(\lambda)$:
\begin{equation}
    U_{\lambda,a} = 
    \begin{pNiceMatrix}
            \sqrt{g_{\lambda,a}} & -\sqrt{1-g_{\lambda,a}}  \\
            \sqrt{1-g_{\lambda,a}} & \sqrt{g_{\lambda,a}} 
    \end{pNiceMatrix},
    \label{eq:U_lambda_a}
\end{equation}
where $g_{\lambda,a}$ are defined in \cref{def:G_lambda_a_ppbt} for pPBT with EPR resource state. The resulting circuit \cref{fig:p_pbt_epr_circuit} is almost the same as \cref{fig:pbt_std_pgm_circuit}. The only difference is an additional register and a controlled rotation matrix (\ref{eq:U_lambda_a}). 

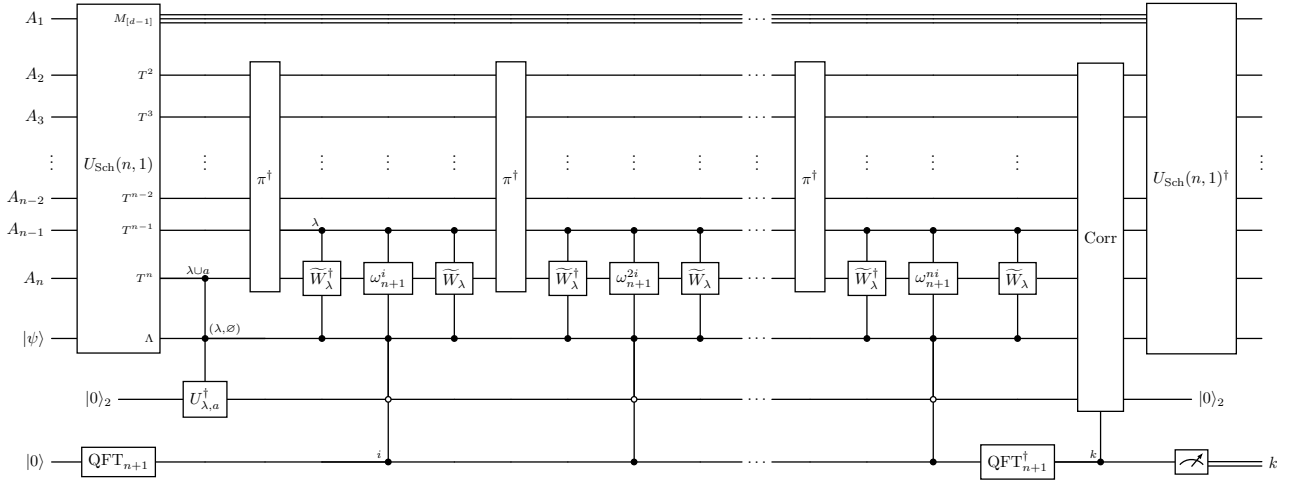
\begin{figure}[H]
\centering
\newcommand{\Sgi}{\gate[wires=8]{U_{\mathrm{Sch}}(n,1)^\dagger}}
\newcommand{\Sg}{\gate[wires=8,label style={yshift=0.3cm}]{U_{\mathrm{Sch}}(n,1)}}
\newcommand{\go}[1]{\gateoutput{#1}}
\newcommand{\iA}[1]{\lstick{$A_{#1}$}}
\newcommand{\ipsi}{\lstick{$\ket{\psi}$}}
\newcommand{\keto}{\lstick{$\ket{0}$}}
\newcommand{\Ginv}{\gate{U^\dagger_{\lambda,a}}}
\newcommand{\Ggate}{\gate{U_{\lambda,a}}}
\newcommand{\namel}[1]{\wire[l][1]["(\lambda {,} \0)"{above,pos=#1}]{a} }
\newcommand{\namer}[1]{\wire[r][1]["(\lambda {,} \0)"{above,pos=#1}]{a} }
\newcommand{\name}[1]{\wire[l][1]["#1"{above,pos=0.1}]{a} }
\newcommand{\Cyc}[1]{\gate[wires=#1]{\pi^\dagger}}
\newcommand{\gW}[1]{\gate[wires=#1]{\widetilde{W}_{\lambda}}}
\newcommand{\gWinv}[1]{\gate[wires=#1]{\widetilde{W}_{\lambda}^\dagger}}
\newcommand{\gateCorr}{\gate[wires=8]{\mathrm{Corr}}}
\newcommand{\clo}[1]{\ctrl[open]{#1}}  
\newcommand{\gateR}[1]{\gate{\omega_{n+1}^{#1}}}

\resizebox{\textwidth}{!}{
\begin{quantikz}[classical gap = 2.5pt, row sep = 0.6cm]
\iA{1}  &\Sg \go{$\gtM{d-1}$}&\qb                   & & & & & & & & & \HD & & & &                                                                                                     &  & \Sgi & \qq \\
\iA{2}  &\go{$T^{2}$}        &                      &\Cyc{6}  & & & & \Cyc{6}  & & & & \HD & \Cyc{6} & & &                                                                 & \gateCorr & & \\
\iA{3}  &\go{$T^{3}$}        &                      &           & & & &           & & & & \HD &          & & &                                                                         & & & \\
\qn\vd  &                    & \vd                  &           & \vd & \vd & \vd & & \vd & \vd & \vd & \vd & & \vd & \vd & \vd                            & & & \vd \\
\iA{n-2}&\go{$T^{n-2}$}      &                      &           & & & & & & & & \HD & & & &                                                                                            & & & \\
\iA{n-1}&\go{$T^{n-1}$}      &                      &           & \ctrl{1} \name{\lambda} & \ctrl{1}  & \ctrl{1} & & \ctrl{1} & \ctrl{1} & \ctrl{1} & \HD & & \ctrl{1} & \ctrl{1} & \ctrl{1}         & & & \\
\iA{n}  &\go{$T^{n}$}        &\ctrl{1}\name{\lambda \cup a}    &           & \gWinv{1} & \gateR{i} & \gW{1}  & & \gWinv{1} & \gateR{2i} & \gW{1} & \HD & & \gWinv{1} & \gateR{ni} & \gW{1}       & & & \\
\ipsi   &\go{$\Lambda$}      &\ctrl{1}\namer{0.3}   &           & \ctrl{-1} & \ctrl{-1} & \ctrl{-1} & & \ctrl{-1} & \ctrl{-1} & \ctrl{-1} & \HD & & \ctrl{-1} & \ctrl{-1} & \ctrl{-1} & & & \\
\qn     &\lstick{$\ket{0}_2$}  &\Ginv \qq             &           &           & \clo{-1} &           & &           & \clo{-1} &           & \HD & &           & \clo{-1} &           & & \rstick{$\ket{0}_2$}   & \qn   \\
\keto   &\gateQFT{n+1}       &                      & & &\ctrl{-2} \name{i} & & & & \ctrl{-2} & & \HD & & & \ctrl{-2} & \gateQFTinv{n+1} & \ctrl{-1} \name{k} & \meter{} \wire[r][1][classical gap = 0.2pt]{c} & \qn \rstick{$k$}
\end{quantikz}
}
\caption{The circuit implementation of the POVM $E^{\star}$ for pPBT with EPR resource state defined via \cref{def:G_lambda_ppbt} in standard encoding. }
\label{fig:p_pbt_epr_circuit}
\end{figure}

The complexity of computing $g_{\lambda,a}$ classically is $\widetilde{O}(1)$. Therefore the time, gate and space complexities of the circuit \cref{fig:p_pbt_epr_circuit} are the same as for \cref{fig:pbt_std_pgm_circuit}.
The correction gate $\mathrm{Corr}$ in \cref{fig:p_pbt_epr_circuit} is implemented in similar way as in \cref{fig:correction_gate}, except one needs to uncompute $U_{\lambda,a}$ on the additional qubit register and do the uncomputation conditioned on the outcome $k=0$.

\begin{figure}[H]
    \centering
    \newcommand{\mideq}{\midstick[8,brackets=none]{\text{$=$}}}
    \newcommand{\Corr}{\gate[wires=8]{\mathrm{Corr}}}
    \newcommand{\T}[1]{\lstick{$T^{#1}$}}
    \newcommand{\Linit}{\lstick{$\Lambda$}}
    \newcommand{\name}[1]{\wire[l][1]["#1"{above,pos=0.1}]{a} }  
    \newcommand{\Cyc}[1]{\gate[wires=6]{\pi^{#1}}}
    \newcommand{\gWt}{\gate{\widetilde{W}_{\lambda}^\dagger}}
    \newcommand{\gW}{\gate{W_{\lambda}}}
    \newcommand{\Xg}[1]{\gate{X^{#1}}}
    \newcommand{\namep}[2]{\wire[l][1]["#1"{above,pos=#2}]{a} } 
    
    \resizebox{0.5\textwidth}{!}{
    \begin{quantikz}[classical gap = 2pt]
    \T{2}  &  \Corr          &   &\mideq\qn&&\qq&    \Cyc{-k}     &         &         &    \Cyc{k}      &   &\\
    \T{3}  &                 &   &      \qn&&\qq&                 &         &         &                 &   &\\
    \qn\vd &                 &\vd&      \qn&&\vd&                 &  \vd    &  \vd    &                 &\vd&\\
    \T{n-2}&                 &   &      \qn&&\qq&                 &         &         &                 &   &\\
    \T{n-1}&                 &   &      \qn&&\qq&&\ctrl{1}\namep{\lambda}{0.1}&\ctrl{1} &               &   &\\
    \T{n}  &                 &   &      \qn&&\qq&                 &\gWt     &\gW      &                 &   &\\
    \Linit &                 &   &      \qn&&\qq&&\ctrl{-1}\namep{(\lambda{,}\0)}{0.25}&\ctrl{-1}&      &   &\\
           &                 &   &      \qn&&\qq&                 & \targ{} &         &                 &   &\\
           &\ctrl{-1}\name{k}&   &      \qn&&\qq&\ctrl{-3}\name{k}&\ctrl[open]{-1}    & &\ctrl{-3}\name{k}&   &
    \end{quantikz}
    }
    \caption{Circuit for the correction gate $\mathrm{Corr}$ from \Cref{fig:p_pbt_epr_circuit}. It is a slight modification of the correction gate for standard PGM from \Cref{fig:correction_gate}.}
    \label{fig:correction_gate_ppbt}
\end{figure}
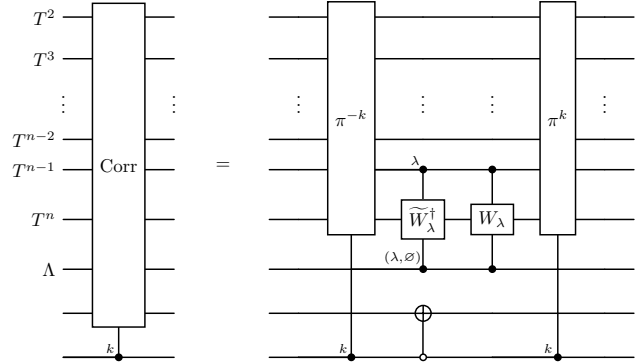

\subsection{Efficient quantum algorithms for generic PBT measurements}
\label{sec:G_PBT_circuits}

We can combine two implementations for dPBT and pPBT measurement and implement a generic measurement $E^{\star}$ from \cref{def:povm_generic} defined via $G$ operator, which is given in the Gelfand--Tsetlin basis via \cref{def:G_lambda_povm_generic}. 

Given a phase gate $\widetilde{Z} \defeq \sum_{i=1}^n \omega_{n+1}^i \proj{i}$ and a unitary $U_{\Lambda,\gamma}$ acting on auxiliary qubit defined as
\begin{equation}
    U_{\Lambda,\gamma} = 
    \begin{cases}
    R_y(g_{\lambda,a}) 
    &\text{if $\Lambda = (\lambda,\0), \gamma = \lambda \cup a$},\\
    R_y(g_{\mu}) 
    &\text{if $\Lambda = (\mu,\square), \gamma = \mu$ },
    \end{cases}
    \qquad
    R_y(g) \defeq 
    \begin{pNiceMatrix}
            \sqrt{g} & -\sqrt{1-g}  \\
            \sqrt{1-g} & \sqrt{g} 
    \end{pNiceMatrix}.
    \label{eq:U_Lambda_gamma}
\end{equation}
we can implement POVM $E^{\star}$ as in \Cref{fig:generic_pbt_circuit}. 

\begin{figure}[H]
\centering
\newcommand{\Sgi}{\gate[wires=8]{U_{\mathrm{Sch}}(n,1)^\dagger}}
\newcommand{\Sg}{\gate[wires=8,label style={yshift=0.3cm}]{U_{\mathrm{Sch}}(n,1)}}
\newcommand{\go}[1]{\gateoutput{#1}}
\newcommand{\iA}[1]{\lstick{$A_{#1}$}}
\newcommand{\ipsi}{\lstick{$\ket{\psi}$}}
\newcommand{\keto}{\lstick{$\ket{0}$}}
\newcommand{\Ginv}{\gate{U^\dagger_{\Lambda,\gamma}}}
\newcommand{\Ggate}{\gate{U_{\Lambda,\gamma}}}
\newcommand{\Zgate}[1]{\gate{\widetilde{Z}^{#1}}}
\newcommand{\namel}[1]{\wire[l][1]["(\lambda {,} \0)"{above,pos=#1}]{a} }
\newcommand{\namer}[1]{\wire[r][1]["\Lambda"{above,pos=#1}]{a} }
\newcommand{\namerm}[1]{\wire[r][1]["(\mu {,} \square)"{above,pos=#1}]{a} }
\newcommand{\name}[1]{\wire[l][1]["#1"{above,pos=0.1}]{a} }
\newcommand{\lambdar}[1]{\wire[r][1]["(\lambda {,} \0)"{above,pos=#1}]{a} }
\newcommand{\Cyc}[1]{\gate[wires=#1]{\pi^\dagger}}
\newcommand{\gW}[1]{\gate[wires=#1]{\widetilde{W}_{\lambda}}}
\newcommand{\gWinv}[1]{\gate[wires=#1]{\widetilde{W}_{\lambda}^\dagger}}
\newcommand{\gateCorr}{\gate[wires=9]{\mathrm{Corr}}}
\newcommand{\clo}[1]{\ctrl[open]{#1}} 
\newcommand{\gateR}[1]{\gate{\omega_{n+1}^{#1}}}

\resizebox{\textwidth}{!}{
\begin{quantikz}[classical gap = 2.5pt, row sep = 0.6cm]
\iA{1}  &\Sg \go{$\gtM{d-1}$}&\qb                  &                    &                      & & & & & & & \HD & & & &                                                                                   & & & \Sgi & \qq \\
\iA{2}  &\go{$T^{2}$}        &                     &\Cyc{6}             &                      & & & \Cyc{6}  & & & & \HD & \Cyc{6} & & &                                                                & & \gateCorr & & \\
\iA{3}  &\go{$T^{3}$}        &                     &                    &                      & & &           & & & & \HD &          & & &                                                                        & & & & \\
\qn\vd  &                    &\vd                  &                    &\vd                   & \vd & \vd & & \vd & \vd & \vd & \vd & & \vd & \vd & \vd                            &\vd& & & \vd \\
\iA{n-2}&\go{$T^{n-2}$}      &                     &                    &                      & & & & & & & \HD & & & &                                                                                           & & & & \\
\iA{n-1}&\go{$T^{n-1}$}      &                     &                    &\ctrl{1}\name{\lambda}& \ctrl{1}  & \ctrl{1} & & \ctrl{1} & \ctrl{1} & \ctrl{1} & \HD & & \ctrl{1} & \ctrl{1} & \ctrl{1}        & & & & \\
\iA{n}  &\go{$T^{n}$}        &\ctrl{1}\name{\gamma}&                    &\gWinv{1}             & \gateR{i} & \gW{1}  & & \gWinv{1} & \gateR{2i} & \gW{1} & \HD & & \gWinv{1} & \gateR{ni} & \gW{1}       &\ctrl{1}& & & \\
\ipsi   &\go{$\Lambda$}      &\ctrl{1}\namer{0.1}  &\ctrl{2}\namerm{0.3}&\ctrl{-1}\lambdar{0.3}& \ctrl{-1} & \ctrl{-1} & & \ctrl{-1} & \ctrl{-1} & \ctrl{-1} & \HD & & \ctrl{-1} & \ctrl{-1} & \ctrl{-1} &\ctrl{1}& & & \\
\qn     &\lstick{$\ket{0}_2$}&\Ginv     \qq        &                    &                      & \clo{-1} &           & &           & \clo{-1} &           & \HD & &           & \clo{-1} &           &\Ggate & & \rstick{$\ket{0}_2$}   & \qn  \\
\qn     &\lstick{$\ket{+}_n$}&\qq                  &\Zgate{i}           &                      &          &           & &           &          &           & \HD & &                      &   &   &  & & \rstick{$\ket{1}_n$}   & \qn   \\
\keto   &\gateQFT{n+1}       &                     &\ctrl{-1}\name{i}   &                      &\ctrl{-2} \name{i} & & & & \ctrl{-2} & & \HD & & & \ctrl{-2} & & \gateQFTinv{n+1} & \ctrl{-1} \name{k} & \meter{} \wire[r][1][classical gap = 0.2pt]{c} & \qn \rstick{$k$}
\end{quantikz}
}
\caption{The circuit implementation of a generic POVM $E^{\star}$ for PBT from \cref{def:povm_generic} in standard encoding.
}
\label{fig:generic_pbt_circuit}
\end{figure}
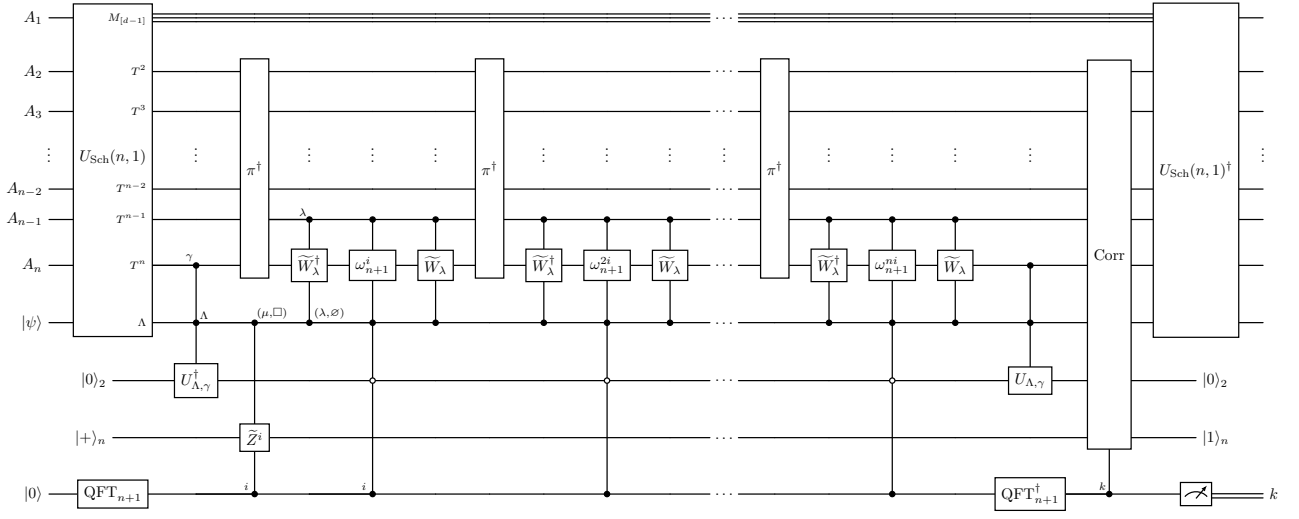

Note that our construction works for any diagonal matrix $\psi_{\Lambda}(G)$ as long as its diagonal entries $g_{\lambda,a},g_{\mu}$ are efficiently classically computable. The correction gate $\mathrm{Corr}$ in \cref{fig:generic_pbt_circuit} is implemented in a similar way as in \cref{fig:correction_gate,fig:correction_gate_dpbt,fig:correction_gate_ppbt}.

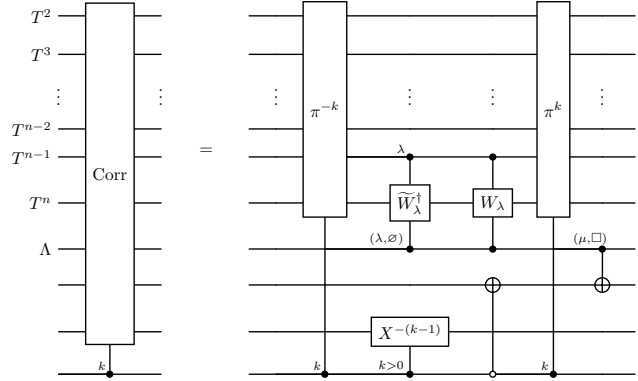
\begin{figure}[H]
    \centering
    \newcommand{\mideq}{\midstick[8,brackets=none]{\text{$=$}}}
    \newcommand{\Corr}{\gate[wires=9]{\mathrm{Corr}}}
    \newcommand{\T}[1]{\lstick{$T^{#1}$}}
    \newcommand{\Linit}{\lstick{$\Lambda$}}
    \newcommand{\name}[1]{\wire[l][1]["#1"{above,pos=0.1}]{a} }  
    \newcommand{\namep}[2]{\wire[l][1]["#1"{above,pos=#2}]{a} } 
    \newcommand{\Cyc}[1]{\gate[wires=6]{\pi^{#1}}}
    \newcommand{\gWt}{\gate{\widetilde{W}_{\lambda}^\dagger}}
    \newcommand{\gW}{\gate{W_{\lambda}}}
    \newcommand{\Xg}[1]{\gate{X^{#1}}}
    
    \resizebox{0.5\textwidth}{!}{
    \begin{quantikz}[classical gap = 2pt]
    \T{2}  &  \Corr          &   &\mideq\qn&&\qq&    \Cyc{-k}     &                                     &         &    \Cyc{k}      &                       &\\
    \T{3}  &                 &   &      \qn&&\qq&                 &                                     &         &                 &                       &\\
    \qn\vd &                 &\vd&      \qn&&\vd&                 &  \vd                                &  \vd    &                 &\vd                    &\\
    \T{n-2}&                 &   &      \qn&&\qq&                 &                                     &         &                 &                       &\\
    \T{n-1}&                 &   &      \qn&&\qq&                 &\ctrl{1}\namep{\lambda}{0.1}         &\ctrl{1} &                 &                       &\\
    \T{n}  &                 &   &      \qn&&\qq&                 &\gWt                                 &\gW      &                 &                       &\\
    \Linit &                 &   &      \qn&&\qq&                 &\ctrl{-1}\namep{(\lambda{,}\0)}{0.25}&\ctrl{-1}& &\ctrl{1}\namep{(\mu {,} \square)}{0.2} &\\
           &                 &   &      \qn&&\qq&                 &                                     & \targ{} &                 &\targ{}                &\\
           &                 &   &      \qn&&\qq&                 &\Xg{-(k-1)}                          &         &                 &                       &\\
           &\ctrl{-1}\name{k}&   &      \qn&&\qq&\ctrl{-4}\name{k}&  \ctrl{-1}\namep{k>0}{0.2}          &\ctrl[open]{-2}&\ctrl{-4}\name{k}&                 &
    \end{quantikz}
    }
    \caption{Circuit for the correction gate $\mathrm{Corr}$ from \Cref{fig:generic_pbt_circuit}.}
    \label{fig:correction_gate_generic}
\end{figure}

\subsection{Exponentially improved lower bound for non-local quantum computation}

Port-based teleportation has interesting applications in holography and non-local quantum computation \cite{may2019quantum,may2022complexity}, where it was argued that the complexity of the local operation controls the amount of entanglement needed to implement it non-locally, using ideas from AdS/CFT correspondence. In particular, it was derived in \cite[Lemma 9]{may2022complexity} that port-based teleportation can be used to lower bound the amount of entanglement needed to implement a given channel (from a large class of one-sided quantum channels) non-locally in terms of the so-called \emph{interaction-class circuit complexity} \cite[Definition 3]{may2022complexity} denoted by $\mathcal{C}$:
\begin{equation}
    \Omega \of*{ \log \log \mathcal{C}} \leq E_c,
\end{equation}
where $E_c$ is the entanglement cost needed to implement non-locally a unitary with complexity $\mathcal{C}$, see \cite[Equation 29]{may2022complexity}. Port-based teleportation can also be used to find an upper bound \cite{beigi2011simplified,speelman2016nonlocal,may2022complexity}. The derivation of the lower bound uses a trivial upper bound $\exp\of{O(p)}$ for the complexity of the port-based teleportation in terms of the number of ports $n$, see \cite[Equation 47]{may2022complexity}. It is already pointed out in \cite[page 28]{may2022complexity} that a better implementation of the port-based teleportation protocol would lead to a better lower bound. Complexities of our implementations of PBT protocols are $\widetilde{O}(nd^4)$, therefore this immediately translates, according to \cite[Lemma 9]{may2022complexity}, to a better lower bound:
\begin{equation}
    \Omega \of*{ \log \mathcal{C}} \leq E_c,
\end{equation}
thus improving exponentially upon the previous known bound.

\section{Logarithmic space PBT via Yamanouchi encoding}\label{sec:yaman_enc}

Using the Yamanouchi encoding possibility of the mixed quantum Schur transform \cite{grinko2023gelfandtsetlin} defined in \cref{eq:quantum_schur_domain_range_yamanouchi}, it is possible to reduce the space complexity of the constructions presented in \cref{sec:circuit_for_pgm} from $O(n\log(n))$ to $O(\log(n))$. The resulting circuits for generic PBT measurements $E^{\star}$, including standard PGM $E$, dPBT POVM $E^{\hollowstar}$ and the POVM for EPR pPBT, are presented in \cref{fig:p_pbt_epr_circuit_yamanouchi}. They look essentially the same as \cref{fig:generic_pbt_circuit} except for the differences in implementation of gates $\pi^\dagger$, $\sigma_i$ and $\widetilde{W}_\lambda$ which stem from the different type of encoding. 

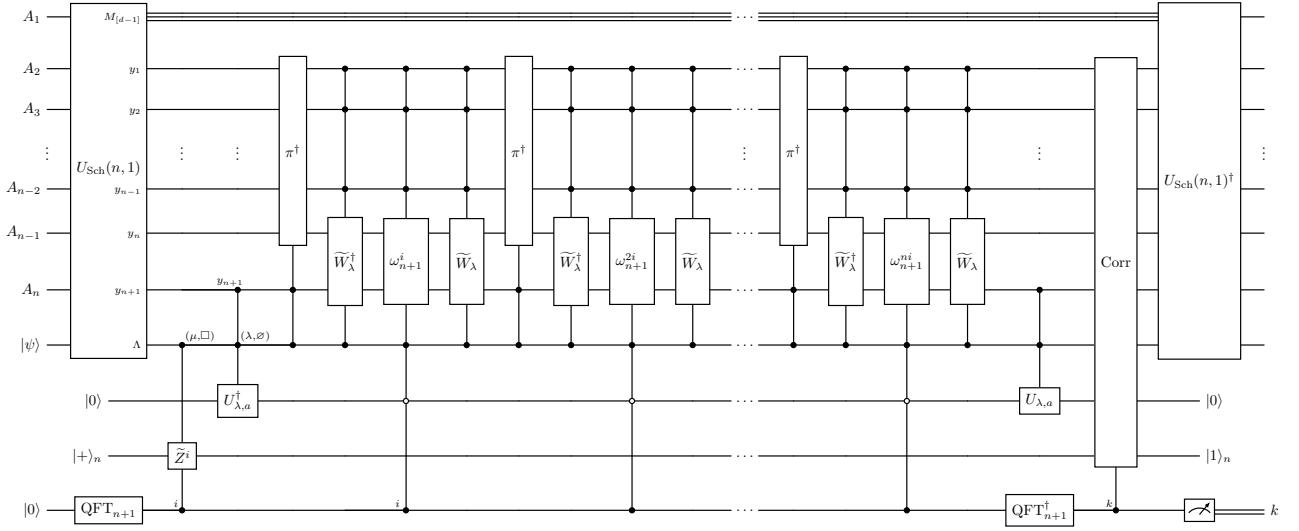
\begin{figure}[H]
\centering
\newcommand{\Sgi}{\gate[wires=8]{U_{\mathrm{Sch}}(n,1)^\dagger}}
\newcommand{\Sg}{\gate[wires=8,label style={yshift=0.3cm}]{U_{\mathrm{Sch}}(n,1)}}
\newcommand{\go}[1]{\gateoutput{#1}}
\newcommand{\iA}[1]{\lstick{$A_{#1}$}}
\newcommand{\ipsi}{\lstick{$\ket{\psi}$}}
\newcommand{\keto}{\lstick{$\ket{0}$}}
\newcommand{\Ginv}{\gate{U^\dagger_{\lambda,a}}}
\newcommand{\Ggate}{\gate{U_{\lambda,a}}}
\newcommand{\namel}[1]{\wire[l][1]["(\lambda {,} \0)"{above,pos=#1}]{a} }
\newcommand{\namer}[1]{\wire[r][1]["(\lambda {,} \0)"{above,pos=#1}]{a} }
\newcommand{\name}[1]{\wire[l][1]["#1"{above,pos=0.1}]{a} }
\newcommand{\gateRt}[1]{\gate[wires=2]{\omega_{n+1}^{#1}}}
\newcommand{\finalcw}{\wire[r][1][classical gap = 0.2pt]{c}}
\newcommand{\Cyc}[1]{\gate[wires=#1]{\pi^\dagger}}
\newcommand{\gW}[1]{\gate[wires=#1]{\widetilde{W}_{\lambda}}}
\newcommand{\gWinv}[1]{\gate[wires=#1]{\widetilde{W}_{\lambda}^\dagger}}
\newcommand{\gateCorr}{\gate[wires=9]{\mathrm{Corr}}}
\newcommand{\clo}[1]{\ctrl[open]{#1}}  
\newcommand{\Zgate}[1]{\gate{\widetilde{Z}^{#1}}}
\newcommand{\nrm}[1]{\wire[r][1]["(\mu {,} \square)"{above,pos=#1}]{a} }

\resizebox{\textwidth}{!}{
\begin{quantikz}[classical gap = 2.5pt, row sep = 0.6cm]
\iA{1}  &\Sg \go{$\gtM{d-1}$}&\qb&                   &         &            &          &           &         &            &           &         & \HD &         &            &           &         &           &         &\Sgi                &\qq \\
\iA{2}  &\go{$y_1$}          &   &                   &\Cyc{5}  &  \ctrl{1}  & \ctrl{1} & \ctrl{1}  &\Cyc{5}  &  \ctrl{1}  & \ctrl{1}  & \ctrl{1}& \HD &\Cyc{5}  &  \ctrl{1}  & \ctrl{1}  & \ctrl{1}&           &\gateCorr&                    &    \\
\iA{3}  &\go{$y_2$}          &   &                   &         &  \ctrl{2}  & \ctrl{2} & \ctrl{2}  &         &  \ctrl{2}  & \ctrl{2}  & \ctrl{2}& \HD &         &  \ctrl{2}  & \ctrl{2}  & \ctrl{2}&           &         &                    &    \\
\qn\vd  &                    &\vd&       \vd         &         &            &          &           &         &            &           &         & \vd &         &            &           &         &  \vd      &         &                    &\vd \\
\iA{n-2}&\go{$y_{n-1}$}      &   &                   &         &  \ctrl{1}  & \ctrl{1} & \ctrl{1}  &         &  \ctrl{1}  & \ctrl{1}  & \ctrl{1}& \HD &         &  \ctrl{1}  & \ctrl{1}  & \ctrl{1}&           &         &                    &    \\
\iA{n-1}&\go{$y_{n}$}        &   &                   &         & \gWinv{2}  &\gateRt{i}& \gW{2}    &         & \gWinv{2}  &\gateRt{2i}& \gW{2}  & \HD &         & \gWinv{2}  &\gateRt{ni}& \gW{2}  &           &         &                    &    \\
\iA{n}  &\go{$y_{n+1}$}      &   &\ctrl{1}\name{y_{n+1}}&\ctrl{-1}&         &          &           &\ctrl{-1}&            &           &         & \HD &\ctrl{-1}&            &           &         &  \ctrl{1} &         &                    &    \\
\ipsi&\go{$\Lambda$}&\ctrl{2}\nrm{0.3}&\ctrl{1}\namer{0.3}&\ctrl{-1}&\ctrl{-1}&\ctrl{-1}&\ctrl{-1} &\ctrl{-1}&  \ctrl{-1} & \ctrl{-1} &\ctrl{-1}& \HD &\ctrl{-1}&  \ctrl{-1} & \ctrl{-1} &\ctrl{-1}&  \ctrl{1} &         &                    &    \\
\qn     &\lstick{$\ket{0}$}  &\qq&\Ginv              &         &            & \clo{-1} &           &         &            & \clo{-1}  &         & \HD &         &            & \clo{-1}  &         &  \Ggate   &         &\rstick{$\ket{0}$}  &\qn \\
\qn     &\lstick{$\ket{+}_n$}&\Zgate{i} \qq&         &         &            &          &           &         &            &           &         & \HD &         &            &           &         &           &         &\rstick{$\ket{1}_n$}&\qn \\
\keto   &\gateQFT{n+1}       &\ctrl{-1}\name{i}&     &         &            & \ctrl{-2} \name{i}  &&         &            & \ctrl{-2} &         & \HD &         &            & \ctrl{-2} &         &\gateQFTinv{n+1}&\ctrl{-1}\name{k}&\meter{} \finalcw &\qn \rstick{$k$}
\end{quantikz}
}
\caption{The circuit implementation of generic POVM $E^{\star}$ in Yamanouchi encoding. Notice that the structure of the circuit is exactly the same as with standard encoding, see \cref{fig:generic_pbt_circuit}, except that wires $T^2,\ldots,T^n$ containing information about path $T=(T^0,\ldots,T^n)$ via standard encoding are replaced by wires $y_1,\ldots,y_{n+1}$ which contain the same information via Yamanouchi encoding. This is obtained by another form of mixed Schur isometry, see \cref{eq:Yamanouchi_encoding}. This requires reformulation for all subsequent gates from the standard to Yamanouchi encoding, which we present on \cref{fig:perm_circuit_yamanouchi,fig:transposition_circuit_yamanouchi_full,fig:transposition_circuit_yamanouchi}.}
\label{fig:p_pbt_epr_circuit_yamanouchi}
\end{figure}

Similarly to \cref{sec:circuit_for_pgm} we can count the total gate and time complexities in \cref{fig:p_pbt_epr_circuit_yamanouchi} as follows:
\begin{enumerate}[leftmargin=*]
    \item The complexity of implementing the mixed quantum Schur transform $U_{\text{Sch}}(n,1)$ in Yamanouchi encoding is $\widetilde{O}(nd^4)$ \cite{grinko2023gelfandtsetlin}.
    \item The complexity of implementing $\pi = \sigma_{1}\sigma_{2} \dotsc \sigma_{n-1}$ based on \cref{fig:perm_circuit_yamanouchi} is $\widetilde{O}(n d^2)$.
    The factor $n$ comes from the number of transpositions $\sigma_i$ in $\pi$, which are implemented sequentially.
        \begin{figure}[H]
        \centering
        \newcommand{\mideq}{\midstick[10,brackets=none]{\text{$=$}}}
        \newcommand{\y}[1]{\lstick{$y_{#1}$}}
        \newcommand{\Ld}{\lstick{$\Lambda$}}
        \newcommand{\sgm}[1]{\gate[wires=2]{\sigma_{#1}}}
        \newcommand{\Cyc}[1]{\gate[wires=#1]{\pi^\dagger}}
        
        \resizebox{0.6\textwidth}{!}{
        \begin{quantikz}[classical gap = 2.5pt, row sep = 0.3cm]
        \y{1}  &\Cyc{8}  &   &\mideq\qn&&\sgm{1}\qq  &\ctrl{1} & \ctrl{1}&\HD&\ctrl{1 }&\ctrl{1 }&\\
        \y{2}  &         &   &      \qn&&         \qq& \sgm{2} & \ctrl{1}&\HD&\ctrl{1 }&\ctrl{1 }&\\
        \y{3}  &         &   &      \qn&&\ctrl{-1}\qq&         & \sgm{3} &\HD&\ctrl{1 }&\ctrl{1 }&\\
        \y{4}  &         &   &      \qn&&\ctrl{-1}\qq&\ctrl{-1}&         &\HD&\ctrl{2 }&\ctrl{2 }&\\
        \qn\vd &         &\vd&      \qn&&            &         &         &\vd&         &         &\\
        \y{n-2}&         &   &      \qn&&\ctrl{-2}\qq&\ctrl{-2}&\ctrl{-2}&\HD&\sgm{n-2}&\ctrl{1 }&\\
        \y{n-1}&         &   &      \qn&&\ctrl{-1}\qq&\ctrl{-1}&\ctrl{-1}&\HD&         &\sgm{n-1}&\\
        \y{n}  &         &   &      \qn&&\ctrl{-1}\qq&\ctrl{-1}&\ctrl{-1}&\HD&\ctrl{-1}&         &\\
        \y{n+1}&\ctrl{-1}&   &      \qn&&\ctrl{-1}\qq&\ctrl{-1}&\ctrl{-1}&\HD&\ctrl{-1}&\ctrl{-1}&\\
        \Ld    &\ctrl{-1}&   &      \qn&&\ctrl{-1}\qq&\ctrl{-1}&\ctrl{-1}&\HD&\ctrl{-1}&\ctrl{-1}&
        \end{quantikz}
        }
        \caption{Quantum circuit for the cyclic permutation $\pi^\dagger = \sigma_{n-1} \sigma_{n-2} \cdots \sigma_2 \sigma_1$ in the Yamanouchi encoding. Each transposition $\sigma_i$ acts locally on registers $y_{i}$ and $ y_{i+1}$, while being controlled on all other registers $y_i$ and $\Lambda$. Note that in standard encoding, $\sigma_i$ was controlled only on two registers, compare with \cref{fig:cyclic_circuit_std}. \cref{fig:transposition_circuit_yamanouchi_full} presents the exact form of a transposition $\sigma_i$ gate in Yamanouchi encoding.}
        \label{fig:perm_circuit_yamanouchi}
        \end{figure}
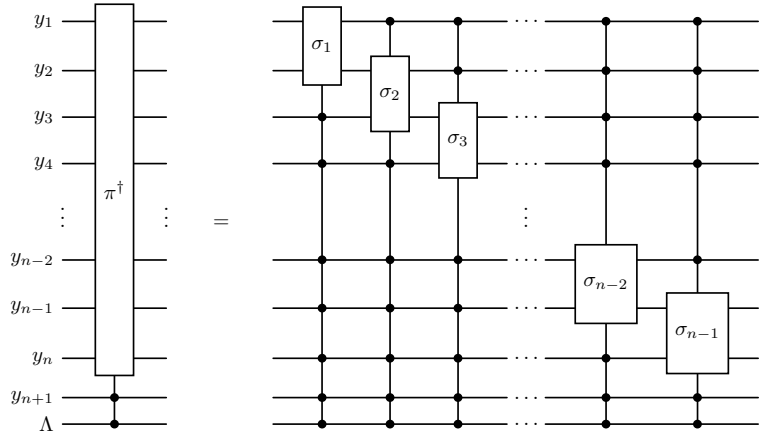
    \item Each transposition $\sigma_i$ is more tricky in Yamanouchi encoding, see \cref{fig:transposition_circuit_yamanouchi_full,fig:transposition_circuit_yamanouchi}.
    More specifically, to implement $\sigma_i$ we need to obtain the full information about the Young diagram $T^{i-1}$ by using an auxiliary space and recording the description of $T^{i-1}$ sequentially from registers $y_k$ for $k<i$ via a $\mathrm{Rec_k}$ gates, see \cref{fig:transposition_circuit_yamanouchi_full}.
    Now, according to \cref{fig:transposition_circuit_yamanouchi}, each $\sigma_i$ can be implemented with $\widetilde{O}(d^2)$ gates $\mathrm{R}_{i,j}$ from \cref{eq:Rjk}, acting on wires $y_i,\,y_{i+1}$ each of which decompose into $\widetilde{O}(1)$ elementary gates and $\widetilde{O}(1)$ auxiliary qubits for computation of rotation parameters $r_{j,k}$.
        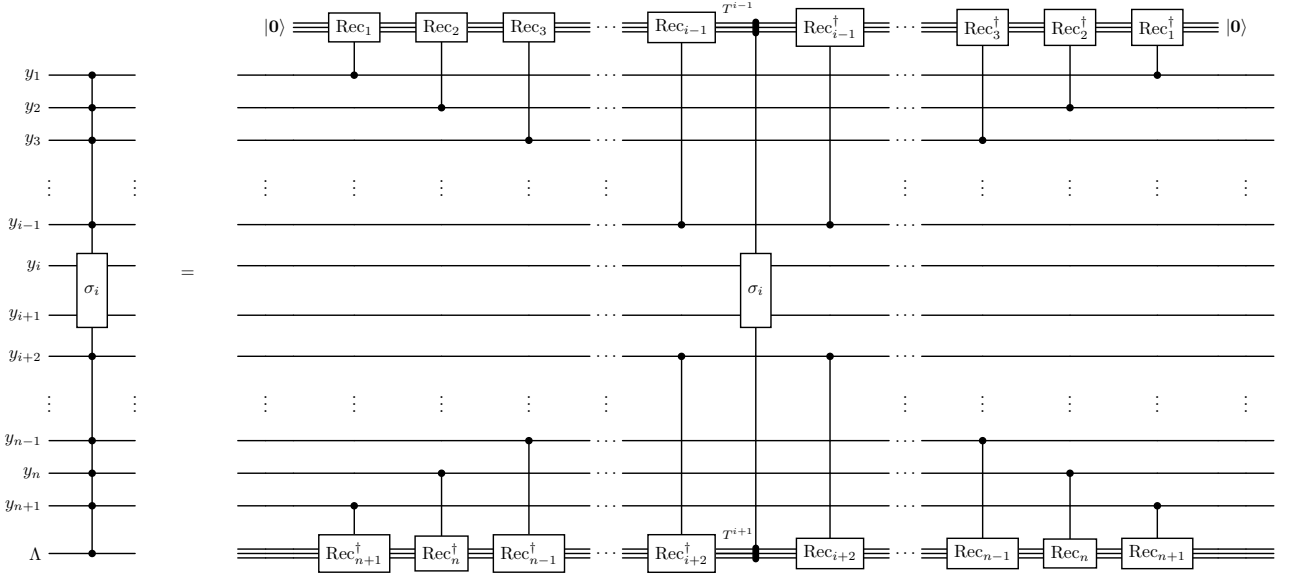
\begin{figure}[H]
        \centering
        \newcommand{\mideq}{\midstick[11,brackets=none]{\text{    $=$    }}}
        \newcommand{\y}[1]{\lstick{$y_{#1}$}}
        \newcommand{\Ld}{\lstick{$\Lambda$}}
        \newcommand{\sgm}[1]{\gate[wires=2]{\sigma_{#1}}}
        \newcommand{\rec}[1]{\gate{\mathrm{Rec}_{#1}}}
        \newcommand{\reci}[1]{\gate{\mathrm{Rec}_{#1}^\dagger}}
        \newcommand{\gcopy}{\gate{\mathrm{Copy}}}
        \newcommand{\gcopyi}{\gate{\mathrm{Copy}^\dagger}}
        \newcommand{\kol}{\lstick{$\ket{\bf{0}}$}}
        \newcommand{\kor}{\rstick{$\ket{\bf{0}}$}}
        \newcommand{\qwb}{\qwbundle{}}
        \newcommand{\nameu}{\wire[l][1]["T^{i-1}"{above=0.2,pos=0.4}]{q} }
        \newcommand{\named}{\wire[l][1]["T^{i+1}"{above=0.2,pos=0.4}]{q} }
        
        \resizebox{\textwidth}{!}{
        \begin{quantikz}[classical gap = 2.5pt, row sep = 0.5cm]
        \qn    &           &   &         &&   &\kol&\rec{1}\qb& \rec{2} & \rec{3}  &\HD&\rec{i-1} &\ctrl{6}\nameu &\reci{i-1}&\HD& \reci{3}& \reci{2}&\reci{1} &\kor&\qn& \\
        \y{1}  & \ctrl{1}  &   &\mideq\qn&&\qq&    &\ctrl{-1} &         &          &\HD&          &               &          &\HD&         &         &\ctrl{-1}&    &   & \\
        \y{2}  & \ctrl{1}  &   & \qn     &&\qq&    &          &\ctrl{-2}&          &\HD&          &               &          &\HD&         &\ctrl{-2}&         &    &   & \\
        \y{3}  & \ctrl{2}  &   & \qn     &&\qq&    &          &         &\ctrl{-3} &\HD&          &               &          &\HD&\ctrl{-3}&         &         &    &   & \\
        \qn\vd &           &\vd&         &&\vd&    &   \vd    &  \vd    &  \vd     &\vd&          &               &          &\vd&  \vd    &  \vd    &   \vd   &    &\vd& \\
        \y{i-1}& \ctrl{1}  &   & \qn     &&\qq&    &          &         &          &\HD& \ctrl{-5}&               & \ctrl{-5}&\HD&         &         &         &    &   & \\
        \y{i}  & \sgm{i}   &   & \qn     &&\qq&    &          &         &          &\HD&          &    \sgm{i}    &          &\HD&         &         &         &    &   & \\
        \y{i+1}&           &   & \qn     &&\qq&    &          &         &          &\HD&          &               &          &\HD&         &         &         &    &   & \\
        \y{i+2}& \ctrl{-1} &   & \qn     &&\qq&    &          &         &          &\HD& \ctrl{5} &               & \ctrl{5} &\HD&         &         &         &    &   & \\
        \qn\vd &           &\vd&         &&\vd&    &   \vd    &  \vd    &  \vd     &\vd&          &               &          &\vd&  \vd    &  \vd    &   \vd   &    &\vd& \\
        \y{n-1}& \ctrl{-2} &   & \qn     &&\qq&    &          &         & \ctrl{3} &\HD&          &               &          &\HD& \ctrl{3}&         &         &    &   & \\
        \y{n}  & \ctrl{-1} &   & \qn     &&\qq&    &          & \ctrl{2}&          &\HD&          &               &          &\HD&         & \ctrl{2}&         &    &   & \\
        \y{n+1}& \ctrl{-1} &   & \qn     &&\qq&    & \ctrl{1} &         &          &\HD&          &               &          &\HD&         &         & \ctrl{1}&    &   & \\
        \Ld    & \ctrl{-1} &   & \qn     &&\qb&    &\reci{n+1}& \reci{n}&\reci{n-1}&\HD&\reci{i+2}&\ctrl{-6}\named&\rec{i+2} &\HD&\rec{n-1}&  \rec{n}&\rec{n+1}&    &   & 
        \end{quantikz}   
        }
        \caption{Quantum circuit for implementing a transposition $\sigma_i$ in the Yamanouchi encoding. To correctly compute the rotation angles, we shall recover the information about parts $T^{i-1}$ and $T^{i+1}$ of a path $T$ from its form in Yamanouchi encoding. This is achieved by a sequence of recording gates $\text{Rec}_i$ performed on two auxiliary registers and controlled on $y_i$ register. After recovering information about $T^{i-1}$ and $T^{i+1}$, transpoition $\sigma_i$ can be simply performed as presented on \cref{fig:transposition_circuit_yamanouchi}.}
        \label{fig:transposition_circuit_yamanouchi_full}
        \end{figure}

        \begin{figure}[H]
        \centering
        \newcommand{\y}[1]{\lstick{$y_{#1}$}}
        \newcommand{\Tu}{\lstick{$T^{i-1}$}}
        \newcommand{\Td}{\lstick{$T^{i+1}$}}
        
        \newcommand{\Tust}{\lstick[5]{$T^{i-1}$}}
        \newcommand{\Tdst}{\lstick[5]{$T^{i+1}$}}
        
        \newcommand{\Tui}[1]{\rstick{$T^{i-1}_{#1}$}}
        \newcommand{\Tdi}[1]{\rstick{$T^{i+1}_{#1}$}}
        
        \newcommand{\grm}[1]{\gate{\mathrm{#1}}}
        \newcommand{\sgma}[1]{\gate[wires=2]{\sigma_{#1}}}
        \newcommand{\sgm}[1]{\gate[wires=2]{\mathrm{R}_{#1}}}
        
        \newcommand{\mideq}{\midstick[12,brackets=none]{\text{$=$}}}
        \newcommand{\clo}[1]{\ctrl[open]{#1}}  
        
        \resizebox{\textwidth}{!}{
        \begin{quantikz}[classical gap = 2.5pt]
        \qn    &         &&\mideq\qn&&\Tust  &\qq&\ctrl{7}&        &\HD&        &        &\ctrl{1} &\HD&\ctrl{3}   &\ctrl{4} &\HD&           &         &\HD&           &        &        &\HD&        &\ctrl{7}&\Tui{1}  \\
        \qn    &         &&      \qn&&       &\qq&        &\ctrl{7}&\HD&        &        &\ctrl{4} &\HD&           &         &\HD&\ctrl{2}   &\ctrl{3} &\HD&           &        &        &\HD&\ctrl{7}&        &\Tui{2}  \\
        \Tu \qb&\ctrl{3} &&      \qn&&       &\qn&        &        &\vd&  \vd   &  \vd   &         &\vd&           &         &\vd&           &         &\vd&   \vd     &  \vd   &  \vd   &\vd&        &        &\vd      \\
        \qn    &         &&      \qn&&       &\qq&        &        &\HD&\ctrl{7}&        &         &\HD&\ctrl{2}   &         &\HD&\ctrl{2}   &         &\HD&\ctrl{1}   &        &\ctrl{7}&\HD&        &        &\Tui{d-1}\\
        \qn    &         &&      \qn&&       &\qq&        &        &\HD&        &\ctrl{7}&         &\HD&           &\ctrl{1} &\HD&           &\ctrl{1} &\HD&\ctrl{1}   &\ctrl{7}&        &\HD&        &        &\Tui{d}  \\
        \y{i}  &\sgma{i} &&      \qn&&\y{i}  &\qq&        &        &\HD&        &        &\sgm{1,2}&\HD&\sgm{1,d-1}&\sgm{1,d}&\HD&\sgm{2,d-1}&\sgm{2,d}&\HD&\sgm{d-1,d}&        &        &\HD&        &        &         \\
        \y{i+1}&         &&      \qn&&\y{i+1}&\qq&        &        &\HD&        &        &         &\HD&           &         &\HD&           &         &\HD&           &        &        &\HD&        &        &         \\
        \qn    &         &&      \qn&&\Tdst  &\qq&\grm{-} &        &\HD&        &        & \clo{-1}&\HD& \clo{-1}  & \clo{-1}&\HD&           &         &\HD&           &        &        &\HD&        &\grm{+} &\Tdi{1}  \\
        \qn    &         &&      \qn&&       &\qq&        &\grm{-} &\HD&        &        & \clo{-1}&\HD&           & \clo{-1}&\HD& \clo{-2}  & \clo{-2}&\HD&           &        &        &\HD&\grm{+} &        &\Tdi{2}  \\
        \Td \qb&\ctrl{-3}&&      \qn&&       &\qn&  \vd   &  \vd   &\vd&        &        &  \vd    &\vd&           &         &\vd&           &         &\vd&           &        &        &\vd&  \vd   &  \vd   &\vd      \\
        \qn    &         &&      \qn&&       &\qq&        &        &\HD&\grm{-} &        &         &\HD& \clo{-3}  &         &\HD& \clo{-2}  &         &\HD& \clo{-4}  &        &\grm{+} &\HD&        &        &\Tdi{d-1}\\
        \qn    &         &&      \qn&&       &\qq&        &        &\HD&        &\grm{-} &         &\HD&           &\clo{-4} &\HD&           & \clo{-3}&\HD& \clo{-1}  &\grm{+} &        &\HD&        &        &\Tdi{d}  
        \end{quantikz}   
        }
        \caption{Quantum circuit implementing the transposition $\sigma_i$ in the Yamanouchi encoding is similar to the circuit in the standard encoding, see \cref{fig:transposition_circuit_std}. Controlled ``$\pm$'' gates and rotation $\mathrm{R}_{j,k}$ gates are exactly the same as on \cref{fig:transposition_circuit_std}.}
        \label{fig:transposition_circuit_yamanouchi}
        \end{figure}
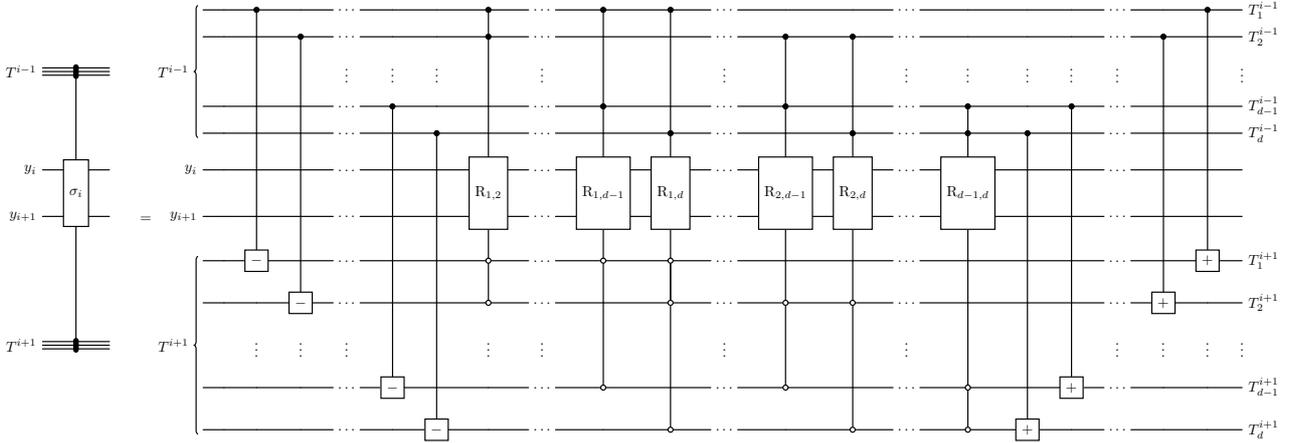

    \item The operation $\widetilde{W}$ defined in \cref{def:W_lambda} is implemented similarly to \cref{fig:W_gate} and the recording procedure described in \cref{fig:transposition_circuit_yamanouchi_full}. The time and gate complexities for that are $\widetilde{O}(n d^2)$.
    \item The implementation of $\omega^{ki}_{n+1}$ has complexity $\widetilde{O}(n)$.
    \item The complexity of the Quantum Fourier Transform $\mathrm{QFT}_{n+1}$ is $\widetilde{O}(1)$.
    \item One can optionally implement the correction gate $\mathrm{Corr}$ together with inverse mixed Schur transform at the end to get the right post-measurement state as in \cref{fig:generic_pbt_circuit}. The complexity of implementing this does not change both the total gate and time complexities of the full circuit, adding only a constant factor overhead.
    \item Overall, counting everything together gives $\widetilde{O}(n^2 d^4)$ for both total gate and time complexities in Yamanouchi encoding: essentially all nontrivial operations run sequentially.
\end{enumerate}

\noindent Similarly, we can count the number of auxiliary qubits needed to implement our circuit \cref{fig:pbt_std_pgm_circuit} in the Yamanouchi encoding similarly to \cref{sec:circuit_for_pgm}:

\begin{enumerate}[leftmargin=*]
    \item The number of auxiliary qubits needed to implement the mixed Schur transform isometry in the standard encoding \cite{grinko2023gelfandtsetlin} and create a Naimark's dilation from \cref{sec:pgm:dilation} after the mixed Schur transfer is $d^2\log(n)\polylog(d,1/\epsilon)$. One important technical remark regarding \cref{fig:p_pbt_epr_circuit_yamanouchi,fig:perm_circuit_yamanouchi,fig:transposition_circuit_yamanouchi,fig:transposition_circuit_yamanouchi_full} is that we do not depict an additional qudit of dimension $n+1$ which is needed to extend the Bratteli diagram from $\Brat$ to $\widetilde{\Brat}$ according to \cref{sec:pgm:dilation}: it extends the space of Yamanouchi words on alphabet $[d]$ to the space of Yamanouchi words with at most one symbol $d+1$ among $y_1,\dotsc,y_n$ by encoding the location $i$ of the value $d+1$ among $y_1,\dotsc,y_n$ (if there exist $i$ such that $y_i = d+1$ then it also must be $y_{n+1}=d+1$). Incorporation of this qubit is trivial and it does not change the time and gate complexities, however, the depiction of it is not convenient, so we omit it.
    \item The number of auxiliary qubits needed to implement each gate $\sigma_i$ from \cref{fig:transposition_circuit_yamanouchi} is $\log(n)\polylog(d,1/\epsilon)$. We implement all gates $\sigma_i$ sequentially in the Yamanouchi encoding so we can reuse $\log(n)\polylog(d,1/\epsilon)$ auxiliary qubits for each gate.
    \item The gates $\widetilde{W}_\lambda$ and $W_\lambda$ require $d^2\log(n)\polylog(d,1/\epsilon)$ qubits to implement gates $\mathrm{R}_i$, see \cref{fig:W_gate}.
    \item Overall, the total number of auxiliary qubits needed is $d^2\log(n)\polylog(d,1/\epsilon)$.
\end{enumerate}

The above analysis finishes the proof of the second statement of \cref{thm:pbt}.

\section{Quantum circuits for optimized resource states}
\label{sec:opt_states}

In this section, we describe efficient quantum circuits for preparing optimized resource states for the aforementioned PBT protocols. 
To write and present such circuits in a unified way, we shall introduce another variant of mixed Schur transform $ \Usch(n,n)$ corresponding to the matrix algebra $\A^d_{n,n}$ of partially transposed permutations acting on $n+n$ qudits, each of local dimension $d$. 
Note that this short introduction is analogous, yet not identical, to the one presented in \cref{sec:Representation theory of the partially transposed permutation algebra,sec:Mixed quantum Schur transform}. 
Again, we refer the more interested reader to \cite{grinko2023gelfandtsetlin} for more details. 

The matrix algebra $\A^d_{n,n}$ is generated by acting on $(\C^d)\xp{n+n}$ in a similar way as (\ref{eq:Brauer action}). 
Its irreducible representations are labelled by the following pairs of Young diagrams:
\begin{equation}
    \Irr{\A_{n,n}^{d}} \defeq
    \Set[\Big]{
    (\lambda,\mu) \given
    \lambda \pt n - k, \;
    \mu \pt n - k, \;
    \len(\lambda) + \len (\mu) \leq d, \;
    \text{for }k \in [n]
    }.
\end{equation}
Furthermore, we adapt the Okounkov--Vershik approach to describe representation theory of $\Irr{\A_{n,n}^{d}} $, together with the notion of the Bratteli diagram $\hat{\Brat}$ for the sequence of algebras
$\A_{0,0}^d \hookrightarrow \A_{1,0}^d \hookrightarrow \cdots \hookrightarrow \A_{n,0}^d \hookrightarrow \A_{n,1}^d\hookrightarrow \cdots \hookrightarrow \A_{n,n}^d$, and the notion of paths in the Bratteli diagram $\Paths(\Lambda,\hat{\Brat}) $ starting at the root and terminating at $\Lambda$. 
As we shall see, for our purpose it is enough to present only one special representation corresponding to irrep $\Lambda=(\0,\0)\in \Irr{\A_{n,n}^{d}}$. 
The set of paths 
\begin{align}
    \Paths((\0,\0),\hat{\Brat}) \defeq
     \Big\{(T^0,\dotsc,T^n,T^{n+1},\dotsc,T^{2n}) \,\Big|\,
     &
     T^{k} , T^{2n-k} \pt k \text{ for } k \leq n, \, \text{ and }
     \\&
      T^{i-1} \rightarrow T^{i},  T^{2n-i+1} \rightarrow T^{2n-i}  \text{ for } i \in [n]
     \Big\}
     \nonumber
     \label{eq:paths2}
\end{align}
span the Gelfand--Tsetlin basis $\set{\ket{T} \mid T \in \Paths(\Lambda,\Brat)}$ of irreducible representation $(\0,\0)\in \Irr{\A_{n,n}^{d}}$, see \cref{fig:StatePrep}. 
Furthermore, for $\mu\pt k$, $k\leq n$, we denote by $\Paths_n(\mu,\hat{\Brat}) $ the set of all paths which terminate in $\mu$ at the level $k$. 
Notice that for $i \in [n]$ the diagram $T^i$ is obtained from $T^{i-1}$ by adding some addable cell $a \in \AC_d(T^{i-1})$, while for $i >n$ the diagram $T^i$ is obtained from $T^{i-1}$ by removing some cell. 

\begin{figure}
\begin{tikzpicture}[> = latex,
cut/.style = {thick, dashed, blue, rounded corners = 12pt},
  every node/.style = {inner sep = 3pt, anchor = west},
  cut2/.style = {thick, dashed, gray, rounded corners = 12pt},
  every node/.style = {inner sep = 3pt, anchor = west},
 MT/.style = {draw = blue!40, line width = 3pt},
  every node/.style = {inner sep = 1pt}]
  \def\W{2.0cm}
  \def\H{1cm}
  \foreach \i/\p/\q in {0/0/0, 0.7/1/0, 1.5/2/0, 2.5/3/0, 3.5/3/1, 4.3/3/2, 5.0/3/3} {
    \node at (\i*\W,3.1*\H) {$\A^3_{\p,\q}$};
  }
  \draw[dashed] (3*\W,3.4*\H) -- (3*\W,-2.5*\H);
  \node (0)   at (0.0*\W, 0.0*\H) {$\0$};
  \node (1)   at (0.7*\W, 0.0*\H) {$\yd{1}$};
  \node (2)   at (1.5*\W, 1.0*\H) {$\yd{2}$};
  \node (11)  at (1.5*\W,-1.0*\H) {$\yd{1,1}$};
  \node (3)    at (2.5*\W, 2*\H) {$\yd{3}$};
  \node (21)   at (2.5*\W, 0*\H) {$\yd{2,1}$};
  \node (111)  at (2.5*\W,-2*\H) {$\yd{1,1,1}$};
   \node (2p)   at (3.5*\W, 1.0*\H) {$\yd{2}$};
  \node (11p)  at (3.5*\W,-1.0*\H) {$\yd{1,1}$};
  \node (1p)   at (4.3*\W, 0.0*\H) {$\yd{1}$};
   \node (0p)   at (5.0*\W, 0.0*\H) {$\0$};
  \draw[->] (0) -- (1);
  \draw[->] (1.north east) -- (2);
  \draw[->] (1.south east) -- (11);
  \draw[->] (2.north east) -- (3);
  \draw[->] (2.south east) -- (21);
  \draw[->] (11.north east) -- (21);
  \draw[->] (11.south east) -- (111);
  \draw[<-] (0p) -- (1p);
  \draw[<-] (1p.north west) -- (2p);
  \draw[<-] (1p.south west) -- (11p);
  \draw[<-] (2p.north west) -- (3);
  \draw[<-] (2p.south west) -- (21);
  \draw[<-] (11p.north west) -- (21);
  \draw[<-] (11p.south west) -- (111);
  \draw [decorate,decoration={brace,amplitude=8pt,mirror,raise=4pt},yshift=0pt,thick]
(5.2*\W, -0.3*\H) -- (5.2*\W, 0.3*\H) node [black,midway,xshift=1.3cm] {$\Lambda=(\0,\0)$};
\end{tikzpicture}
\caption{A part of Bratteli diagram $\hat{\Brat}$ corresponding to an irrep $\Lambda=(\0,\0)$ of algebra $\A_{3,3}^3$. The set of paths starting at the root and terminating at the end span the Gelfand--Tsetlin basis of $\Lambda=(\0,\0)$. A tensor product of $n$ copies of EPR states shared between the first and second half of the systems is fully supported in irrep $\Lambda=(\0,\0)$ in the Schur basis. In fact, it corresponds to a uniform superposition of all symmetric paths in the Gelfand--Tsetlin basis, see \cref{eq:EPR,eq:EPRprim}. Similarly, the optimized resource state for PBT protocols can be expressed as a weighted superposition of symmetric paths in the Gelfand--Tsetlin basis. In particular, for pPBT the exact formula for the weights might be computed in $\widetilde{O}(d)$ time, see \cref{eq:pPBT_amplitudes_resource_state}, while for dPBT, weights are the result of non-trivial optimization procedure, see \cref{eq:Opt_procedure_dPBT}. }
\label{fig:StatePrep}
\end{figure}
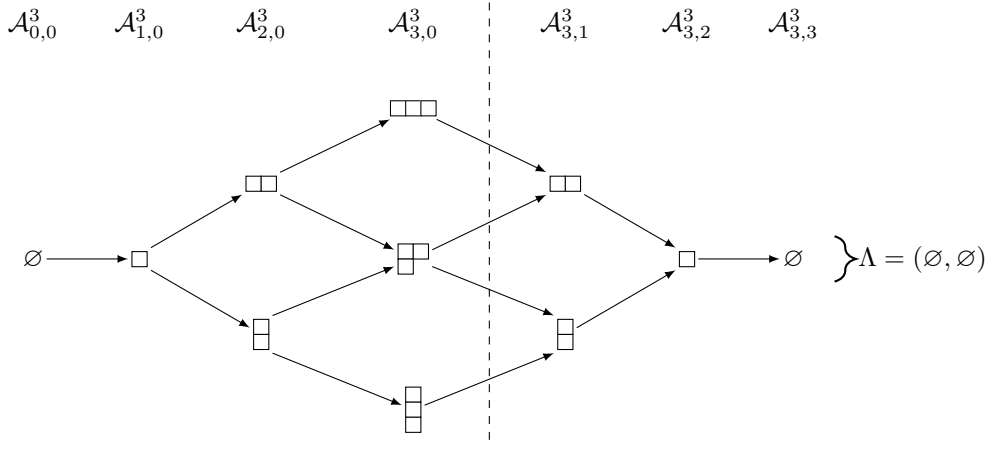

A variant of mixed Schur--Weyl duality partitions the space $(\C^d)\xp{n+n}$ into subspaces that are invariant under the natural $U^{\otimes n} \otimes \bar{U}^{\otimes n}$ action of $U \in \U{d}$ and the action of the matrix algebra $\A^d_{n,n}$. 
It implies existence of mixed quantum Schur transform $\Usch(n,n) \in \U{d^{n+n}}$, a unitary basis change $\Usch(n,n) \in \U{d^{n+n}}$ that maps the computational basis $\set{\ket{x} \mid x \in [d]^{n+n}}$ of $(\C^d)\xp{n+n}$ to a new basis composed of irreducible representations $\Uirrep{\Lambda}$ and $\Airrep{\Lambda}$ of the aforementioned actions of $\U{d}$ and $\A^d_{n,n}$, respectively:
\begin{equation}
    \label{eq:mSWII}
    \Usch(n,n) \colon
    (\C^d)\xp{n+n}
    \to
    \bigoplus_{\Lambda \in \Irr{\A_{n,n}^{d}}} \Airrep{\Lambda} \x \Uirrep{\Lambda}
    \quad\text{where}\quad
    \Airrep{\Lambda} \defeq \C^{\Paths(\Lambda,\hat{\Brat})}
    \text{ and }
    \Uirrep{\Lambda} \defeq \C^{\GT(\Lambda)},
\end{equation}
where the direct sum ranges over all irreducible representations $\Lambda$ of $\A_{n,n}^d$, and $\GT(\Lambda)$ denotes the set of, so-called, Gelfand--Tsetlin patterns, see \cite{grinko2023gelfandtsetlin} for more details. 
Similarly as in \cref{sec:Representation theory of the partially transposed permutation algebra}, transformation (\ref{eq:mSWII}) extends to the following isometry 
\begin{equation}
    \label{eq:quantum_schur_domain_rangeII}
    U_{\text{Sch}} (n,n) \colon (\C^d)\xp{n+1} \to
    \underbrace{\C^{\Irr{\A_{0,0}^d}} \x \cdots \x \C^{\Irr{\A_{n,n}^d}}  }_T
    \otimes
    \underbrace{\C^{\GT((n,n),d)}}_{\gtM{d-1}},
\end{equation}
where paths $T \in \Paths(\hat{\Brat})$ is stored as a tensor product state $\ket{T} = \ket{T^2} \otimes \cdots \otimes  \ket{T^{n+n}}$. 
In \cite{grinko2023gelfandtsetlin}, we described a quantum circuit implementing the mixed Schur isometry, which for any computational basis vector outputs the corresponding superposition of paths $T \in \Paths(\Brat)$ and Gelfand--Tsetlin patterns. 
The complexity of implementing the mixed quantum Schur transform isometry $U_{\text{Sch}} (n,n)$ is $\widetilde{O}(nd^4)$ \cite{nguyen2023mixed,grinko2023gelfandtsetlin}.

A tensor product of $n$ copies of EPR states shared between the first and second half of the systems in $(\C^d)\xp{n+n}$ has a relatively simple form in mixed Schur basis \cite{studzinski2017port}. In particular, it is supported only on one irrep, namely $\Lambda = (\0,\0) \in \Irr{\A^d_{n,n} }$. Notice that the unitary group representation corresponding to $\Lambda = (\0,\0)  $ is one dimensional. 
With a small abuse of notation, we have
\begin{equation}
\label{eq:EPR}
    \Usch(n,n) \; \Big( \bigotimes_{i=1}^n  \ket{\Phi^+}_{A_i B_i} \Big)
    =
    \sum_{\substack{\mu \pt_d n }} \sqrt{\frac{d_\mu m_\mu}{d^n}} \, \ket{\mathrm{EPR^{[n-1]}_\mu}} \, \ket{\mu}_{T^n},
\end{equation}
where 
\begin{equation}
\label{eq:EPRprim}
\ket{\mathrm{EPR^{[n-1]}_\mu}} \defeq 
\sum_{S \in \Paths_n(\mu,\hat{\Brat})} \sqrt{\frac{1}{d_\mu}} \ket{S_{0}}_{T^{0}} \dotsb \ket{S_{n-1}}_{T^{n-1}} \ket{S_{n-1}}_{T^{n+1}} \dotsb \ket{S_{0}}_{T^{2n}}
\end{equation}
and $\ket{\Phi^+}_{A_i B_i} \defeq \tfrac{1}{\sqrt{d}}\sum_{k=1}^{d} \ket{k}_{A_i}\ket{k}_{B_i}$ is an EPR pair shared between Alice's register $i$ and Bob's register $i$. 
Notice that the order of the path registers is slightly changed compared to (\ref{eq:quantum_schur_domain_rangeII}), and all paths present in the formula above are symmetric with respect to the middle vertex. 
The analytical expressions for optimized resource states in dPBT and pPBT protocols were developed in Ref~\cite{studzinski2017port,mozrzymas2018optimal,christandl2021asymptotic}, and in a mixed Schur basis have a similar form to $n$ copies of EPR pairs \eqref{eq:EPR}. 
Indeed, in both cases, they are of the following form
\begin{align}
    \label{def:resource_state}
    \ket{\Psi}_T &= \sum_{\substack{\mu \pt_d n}} \sqrt{f_\mu} \, \ket{\mathrm{EPR^{[n-1]}_\mu}} \, \ket{\mu}_{T^n},
\end{align}
where $\set*{f_\mu}_{\mu \pt_d n}$ is some probability distribution satisfying $\sum_{\mu \pt_d n} f_\mu = 1$.\footnote{In PBT literature sometimes a different parametrization used: $f_\mu = \tfrac{c_\mu d_\mu m_\mu}{d^n}$ where $c_\mu$ are variables. EPR resource state corresponds to the choice $c_\mu = 1$ for every $\mu \pt_d n$.} 
Furthermore, we can rewrite this expression as \begin{align}
    \ket{\Psi}_T &= \sum_{\substack{\lambda \pt_d n-1}} \sqrt{f_\lambda} \ket{\mathrm{EPR^{[n-2]}_\lambda}} \ket{\lambda}_{T^{n-1}} \ket{\lambda}_{T^{n+1}} \of*{\sum_{a \in \AC_d(\lambda)} \sqrt{\frac{f_{\lambda \cup a}}{f_{\lambda}} \frac{d_{\lambda}}{d_{\lambda \cup a}}} \ket{\lambda \cup a}_{T^n}}, \\
    \ket{\mathrm{EPR^{[n-2]}_\lambda}} &\defeq \sum_{S \in \Paths_{n-1}(\lambda,\hat{\Brat})} \sqrt{\frac{1}{d_\lambda}} \ket{S_{0}}_{T^{0}} \dotsb \ket{S_{n-2}}_{T^{n-2}} \ket{S_{n-2}}_{T^{n+2}} \dotsb \ket{S_{0}}_{T^{2n}},
\end{align}
where $f_\lambda$ are defined in such a way that the state $\sum_{a \in \AC_d(\lambda)} \sqrt{\frac{f_{\lambda \cup a}}{f_{\lambda}} \frac{d_{\lambda}}{d_{\lambda \cup a}}} \ket{\lambda \cup a}_{T^n}$ is normalized, namely
\begin{equation}
    \frac{f_\lambda}{d_\lambda} \defeq \sum_{a \in \AC_d(\lambda)} \frac{f_{\lambda \cup a}}{d_{\lambda \cup a}}.
\end{equation}
We continue doing this rewriting recursively. That leads to the circuit for the preparation of $\ket{\Psi}_T$ presented in \cref{fig:state_prep}, where gates $\mathrm{F}_i$ prepare the following states controlled on $T^{i-1}$:
\begin{equation}
    \label{def:F_gates}
    \mathrm{F}_i \ket{\nu}_{T^{i-1}} \ket{0}_{T^{i}} = \ket{\nu}_{T^{i-1}} \of*{\sum_{a \in \AC_d(\nu)} \sqrt{\frac{f_{\nu \cup a}}{f_{\nu}} \frac{d_{\nu}}{d_{\nu \cup a}}} \ket{\nu \cup a}_{T^i}}
\end{equation}
In particular, for pPBT \cite{studzinski2017port} the formulas for $f_{\nu}$ for every $\nu \pt_d k$ turn out to be as follows:
\begin{equation}
    \label{def:f_resource_state}
    f_\nu = \frac{m^2_\nu}{\sum_{\substack{\chi \pt_d k}} m^2_\chi}
\end{equation}
Therefore due to \cite[Proposition 25]{studzinski2017port} and \cref{eq:cont_ratios} the amplitudes in \cref{def:F_gates} for every $a \in \AC_d(\nu)$ and $\nu \pt k$ can be computed efficiently in time $\widetilde{O}(d)$:
\begin{equation}
    \label{eq:pPBT_amplitudes_resource_state}
    \frac{f_{\nu \cup a}}{f_{\nu}} \frac{d_{\nu}}{d_{\nu \cup a}} = \frac{d + \cont(a)}{d^2 + k} \frac{m_{\nu \cup a}}{m_{\nu}} = \frac{d + \cont(a)}{d^2 + k} \of*{\prod_{\substack{i \, : \, i \neq r \\ 1 \leq i \leq d }} \frac{\nu_r - \nu_i + i - r +1}{\nu_r - \nu_i + i - r}},
\end{equation}
where the last equality is due to the Weyl dimension formula \cref{eq:unitary_dimension} and $r$ denotes the row number where the box $a$ was added to Young diagram $\nu$. 

\begin{figure}[H]
\centering
\newcommand{\Sch}{\gate[wires=14]{U_{\mathrm{Sch}}(n,n)^\dagger}}
\newcommand{\go}[1]{\gateoutput{#1}}
\newcommand{\gi}[1]{\gateinput{#1}}
\newcommand{\Cpy}{\gate{\mathrm{Copy}}}
\newcommand{\Fgt}[1]{\gate{\mathrm{F}_{#1}}}

\newcommand{\keto}{\lstick{$\ket{0}$}}
\newcommand{\namel}[1]{\wire[l][1]["(\lambda {,} \0)"{above,pos=#1}]{a} }
\newcommand{\name}[1]{\wire[l][1]["#1"{above,pos=0.1}]{a} }

\resizebox{0.6\textwidth}{!}{
\begin{quantikz}[classical gap = 2.5pt]
\keto   &       &        &        &        &\HD&        &        &\Sch \gi{$T^0$}\go{$A_1$}        &   &\rstick[7]{Alice}\\
\keto   &       &        &        &        &\HD&        &        &     \gi{$T^1$}\go{$A_2$}        &   &                 \\
\keto   &\Fgt{2}&\ctrl{8}&\ctrl{1}&        &\HD&        &        &     \gi{$T^2$}\go{$A_3$}        &   &                 \\
\keto   &       &        &\Fgt{3} &\ctrl{6}&\HD&        &        &     \gi{$T^3$}\go{$A_4$}        &   &                 \\
\qn\vd  &\vd    &        &\vd     &        &\vd&  \vd   & \vd    &                                 &\vd&                 \\
\keto   &       &        &        &        &\HD&\ctrl{2}&\ctrl{1}&     \gi{$T^{n-1}$}\go{$A_{n-1}$}&   &                 \\
\keto   &       &        &        &        &\HD&        &\Fgt{n} &     \gi{$T^{n}$}  \go{$A_n$}    &   &                 \\
\keto   &       &        &        &        &\HD&\gate{+}&        &     \gi{$T^{n+1}$}\go{$B_n$}    &   &\rstick[7]{Bob}  \\
\qn\vd  &\vd    &        & \vd    &        &\vd& \vd    & \vd    &                                 &\vd&                 \\
\keto   &       &        &        &\gate{+}&\HD&        &        &     \gi{$T^{2n-3}$} \go{$B_{5}$}&   &                 \\
\keto   &       &\gate{+}&        &        &\HD&        &        &     \gi{$T^{2n-2}$} \go{$B_{4}$}&   &                 \\
\keto   &       &        &        &        &\HD&        &        &     \gi{$T^{2n-1}$} \go{$B_{3}$}&   &                 \\
\keto   &       &        &        &        &\HD&        &        &     \gi{$T^{2n}$}   \go{$B_{2}$}&   &                 \\
\keto\qb&       &        &        &        &\HD&        &        &     \gi{$\gtM{d-1}$}\go{$B_1$}  &\qq&       
\end{quantikz}
}
\caption{Circuit for the preparation of the resource state $\ket{\Psi}_T \ket{0}_{\gtM{d-1}}$ from \cref{def:resource_state}. Gates $\mathrm{F}_i$ are defined in \cref{def:F_gates} via \cref{def:f_resource_state,eq:pPBT_amplitudes_resource_state}.}
\label{fig:state_prep}
\end{figure}
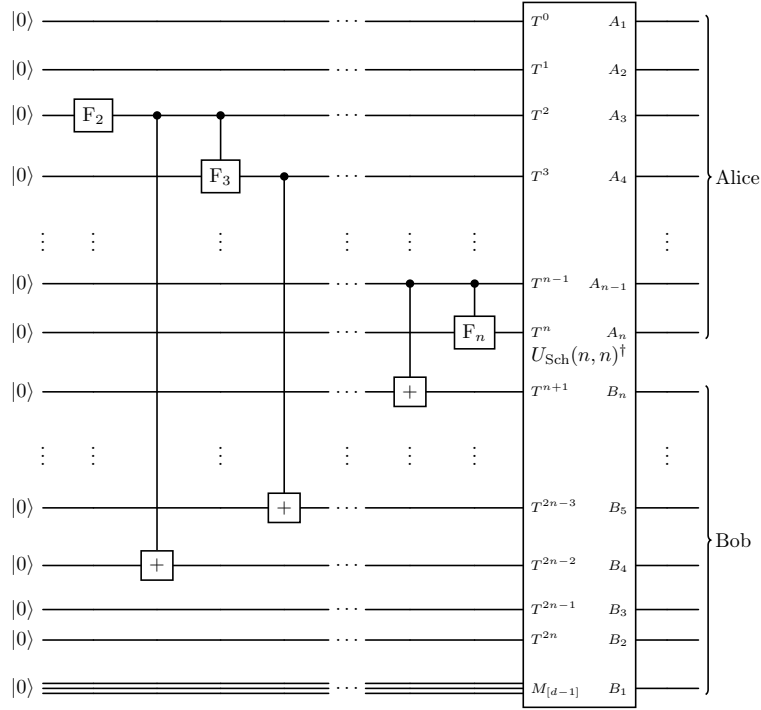

However, we cannot do the same analysis for the dPBT protocol \cite{mozrzymas2018optimal,leditzky2020optimality} since $f_\mu$ for $\mu \pt_d n$ are defined via non-trivial optimization problem \cite[Equation 6.3]{leditzky2020optimality}:
\begin{equation}
\label{eq:Opt_procedure_dPBT}
    \set*{f_\mu}_{\mu \pt_d n} = \argmax_{\sum_{\mu} f_\mu = 1} \sum_{\lambda \pt_d n-1} \of[\Bigg]{\sum_{\substack{a \in \AC_d(\lambda)}} \sqrt{f_{\lambda \cup a}}}^2.
\end{equation}
We leave it as an open question for future work to understand how to efficiently compute the amplitudes in \cref{def:F_gates} for dPBT. At the same time, a suboptimal choice for dPBT resource states, described in \cite[Section B]{christandl2021asymptotic}, are implementable via our scheme presented above.

\section*{Acknowledgements}

DG thanks Tudor Giurgica-Tiron, Quynh Nguyen, Aram Harrow, Hari Krovi, Yanlin Chen and John van de Wetering for useful general discussions. We thank Jiani Fei for her comments and for spotting an error in the dPBT post-measurement state presented in version 1 of the manuscript. We thank Adam Wills for his comments about known results on port-based teleportation, for pointing out logarithmic space qubit Schur transform from \cite{wills2023generalised}, and for useful general discussions. We also thank Rene Allerstorfer, Harry Burhman, Florian Speelman and Philip Verduyn Lunel for discussing the relationship between port-based teleportation and non-local quantum computation. Finally, we thank all authors of \cite{nguyen2023mixed,fei2023efficient} for the coordination with version 1 of \cite{grinko2023gelfandtsetlin} which initially described the efficient construction of standard PGM, now presented in the current paper. DG, AB, and MO were supported by an NWO Vidi grant (Project No VI.Vidi.192.109).


\printbibliography

\end{document}